\documentclass[conference]{IEEEtran}

\newif\ifLongVersion\LongVersionfalse

\ifLongVersion 
\usepackage[createShortEnv, conf={normal}]{proof-at-the-end}
\newenvironment{proofsTextEnd}{}{}
\newenvironment{recTextEnd}{}{}
\newenvironment{spTextEnd}{}{}
\else 
\usepackage[createShortEnv,conf={end,link to proof,restate}]{proof-at-the-end}
\newenvironment{proofsTextEnd}{\begin{textAtEnd}[category=proofs]}{\end{textAtEnd}}
\newenvironment{recTextEnd}{\begin{textAtEnd}[category=rec]}{\end{textAtEnd}}
\newenvironment{spTextEnd}{\begin{textAtEnd}[category=sp]}{\end{textAtEnd}}
\fi

\usepackage{cite}
\usepackage{paralist}
\usepackage{amssymb}
\usepackage{amsmath}
\usepackage{mathtools}
\usepackage{todonotes}
\usepackage{algorithm}
\usepackage{algorithmicx}
\usepackage[noend]{algpseudocode}
\usepackage{multicol}
\usepackage{bussproofs}

\newenvironment{proof}{}{\qed}
\newcommand{\qed}{\hfill$\square$}

\newenvironment{example}{\noindent\emph{Example}.}{}

\newtheorem{definition}{Definition}
\newtheorem{lemma}{Lemma}
\newtheorem{theorem}{Theorem}
\newtheorem{corollary}{Corollary}
\newtheorem{fact}{Fact}

\newcommand{\nat}{\mathbb{N}}

\newcommand{\arity}{\#}
\newcommand{\arityof}[1]{{\arity{#1}}}
\newcommand{\rankof}[1]{\alpha({#1})}
\newcommand{\sortof}[1]{\sort({#1})}
\newcommand{\homof}[2]{h_{{\mathcal{#1}},{\mathcal{#2}}}}

\newcommand{\lenof}[1]{|{#1}|}
\newcommand{\cardof}[1]{{\mathrm{card}({#1})}}
\newcommand{\sizeof}[1]{\mathrm{size}({#1})}
\newcommand{\width}[1]{\mathrm{wd}({#1})}
\newcommand{\twd}[1]{\mathrm{twd}({#1})}

\newcommand{\universeOf}[1]{\mathsf{#1}}

\newcommand{\vars}{\mathcal{V}}
\newcommand{\Vars}{\mathcal{X}}

\newcommand{\isdef}{\stackrel{\scalebox{0.5}{$\mathsf{def}$}}{=}}

\newcommand{\interv}[2]{[{#1},{#2}]}
\newcommand{\tuple}[1]{\langle {#1} \rangle}

\newcommand{\set}[1]{\{ {#1} \}}
\newcommand{\mset}[1]{\{\!\!\{ {#1} \}\!\!\}}
\newcommand{\mcup}{+}
\newcommand{\emptymset}{\emptyset_{\mathrm{m}}}
\newcommand{\trunk}[1]{\lceil {#1} \rceil}
\newcommand{\lcm}{\mathrm{lcm}}

\newcommand{\pow}[1]{\mathrm{pow}({#1})}
\newcommand{\finpow}[1]{\mathrm{pow}_{\mathit{fin}}({#1})}
\newcommand{\mpow}[1]{\mathrm{mpow}({#1})}
\newcommand{\dom}[1]{\mathrm{dom}({#1})}
\newcommand{\img}[1]{\mathrm{img}({#1})}

\newcommand{\finsubseteq}{\subseteq_{\mathit{fin}}}

\newcommand{\bigO}{\mathcal{O}}
\newcommand{\twoexptime}{$2\mathsf{EXPTIME}$}

\newcommand{\np}{$\mathsf{NP}$}
\newcommand{\ptime}{$\mathsf{PTIME}$}
\newcommand{\conp}{$\mathsf{co}$-$\mathsf{NP}$}
\newcommand{\exptime}{$\mathsf{EXPTIME}$}

\newcommand{\fsignature}{\mathcal{F}}

\newcommand{\spsignature}{\fsignature_{\mathcal{SP}}}
\newcommand{\dspsignature}{\fsignature_{\mathcal{DSP}}}

\newcommand{\alphabet}{\mathbb{A}}
\newcommand{\alphabetTwo}{\mathbb{B}}
\newcommand{\alphabetParse}{\alphabet_\class}

\newcommand{\states}{\mathcal{Q}}

\newcommand{\initstates}{\mathcal{I}}

\newcommand{\arrow}[2]{\xrightarrow{{\scriptscriptstyle #1}}_{{\scriptstyle #2}}}

\newcommand{\auto}[2]{\mathcal{A}_{
    {#1}
    \ifthenelse{\equal{#2}{}}{}{,{#2}}
}}

\newcommand{\autsat}[2]{\mathcal{A}^{\scriptscriptstyle\mathsf{sat}}_{
    {#1}
    \ifthenelse{\equal{#2}{}}{}{,{#2}}
}}

\newcommand{\autcut}[2]{\mathcal{A}^{\scriptscriptstyle\mathsf{cut}}_{
    {#1}
    \ifthenelse{\equal{#2}{}}{}{,{#2}}
}}

\newcommand{\autcst}[2]{\mathcal{A}^{\scriptscriptstyle\mathsf{cst}}_{
    {#1}
    \ifthenelse{\equal{#2}{}}{}{,{#2}}
  }
}

\makeatletter
\newcommand*{\da@rightarrow}{\mathchar"0\hexnumber@\symAMSa 4B }
\newcommand*{\da@leftarrow}{\mathchar"0\hexnumber@\symAMSa 4C }
\newcommand*{\xdashrightarrow}[2][]{%
  \mathrel{%
    \mathpalette{\da@xarrow{#1}{#2}{}\da@rightarrow{\,}{}}{}%
  }%
}
\newcommand{\xdashleftarrow}[2][]{%
  \mathrel{%
    \mathpalette{\da@xarrow{#1}{#2}\da@leftarrow{}{}{\,}}{}%
  }%
}
\newcommand*{\da@xarrow}[7]{%
  \sbox0{$\ifx#7\scriptstyle\scriptscriptstyle\else\scriptstyle\fi#5#1#6\m@th$}%
  \sbox2{$\ifx#7\scriptstyle\scriptscriptstyle\else\scriptstyle\fi#5#2#6\m@th$}%
  \sbox4{$#7\dabar@\m@th$}%
  \dimen@=\wd0 %
  \ifdim\wd2 >\dimen@
    \dimen@=\wd2 %
  \fi
  \count@=2 %
  \def\da@bars{\dabar@\dabar@}%
  \@whiledim\count@\wd4<\dimen@\do{%
    \advance\count@\@ne
    \expandafter\def\expandafter\da@bars\expandafter{%
      \da@bars
      \dabar@
    }%
  }%
  \mathrel{#3}%
  \mathrel{%
    \mathop{\da@bars}\limits
    \ifx\\#1\\%
    \else
      _{\copy0}%
    \fi
    \ifx\\#2\\%
    \else
      ^{\copy2}%
    \fi
  }%
  \mathrel{#4}%
}
\makeatother

\newcommand{\store}{\mathfrak{s}}

\newcommand{\cardconstr}[3]{\mathsf{card}_{{#2},{#3}}({#1})}

\newcommand{\cmso}{$\mathsf{CMSO}$}

\newcommand{\mso}{$\mathsf{MSO}$}
\newcommand{\mscmso}{$\mathsf{(C)MSO}$}

\newcommand{\poly}[1]{{\mathsf{poly}({#1})}}

\newcommand{\size}[1]{{\mathrm{size}({#1})}}

\renewcommand{\mod}{~\mathrm{mod}~}

\newcommand{\sintfusion}[2]{\widetilde{\mathtt{IF}}({#1}\ifthenelse{\equal{#2}{}}{}{,{#2}})}

\newcommand{\step}[1]{\Rightarrow_{\scriptscriptstyle{#1}}}
\newcommand{\assoc}{\stackrel{\scriptstyle{\mathsf{a}}}{\sim}}

\newcommand{\astep}[1]{\stackrel{\scriptstyle{\mathsf{a}~}}{\Rightarrow}_{\scriptscriptstyle{#1}}}

\newcommand{\langof}[2]{\mathcal{L}_{#1}({#2})}
\newcommand{\alangof}[3]{\mathcal{L}_{#1}^{\scriptscriptstyle{{\mathcal{#2}}}}({#3})}

\newcommand{\requiv}[1]{{\simeq_{\scriptscriptstyle{#1}}}}

\newcommand{\supp}[1]{\mathrm{supp}({#1})}

\newcounter{index}

\newcommand{\graphof}[3]{\mathrm{subgraph}_{#1}[{#2}
\ifthenelse{\equal{#3}{}}{]}{,{#3}]}}

\newcommand{\graph}{G}

\newcommand{\vertices}{V}
\newcommand{\vertof}[1]{\vertices_{\scriptscriptstyle{#1}}}
\newcommand{\vertavoid}[1]{\widetilde{\vertices}_{\scriptscriptstyle{#1}}}
\newcommand{\edgeof}[1]{\edges_{\scriptscriptstyle{#1}}}
\newcommand{\labels}{\lambda}
\newcommand{\labof}[1]{\labels_{\scriptscriptstyle{#1}}}
\newcommand{\edgerel}{\upsilon}
\newcommand{\edgerelof}[1]{\edgerel_{\scriptscriptstyle{#1}}}

\newcommand{\slabs}{\tau}

\newcommand{\sources}{\xi}
\newcommand{\sourceof}[1]{\sources_{\scriptscriptstyle{#1}}}

\newcommand{\algof}[1]{\mathcal{#1}}

\newcommand{\treealgebra}{\algof{T}}

\newcommand{\repalgebra}{\algof{P}}
\newcommand{\repdomain}{\universeOf{P}}
\newcommand{\emptygraph}{\mathbf{0}}

\newcommand{\symbOf}[1]{\overline{#1}}
\newcommand{\sgraph}[1]{#1}

\newcommand{\grammar}{\Gamma}

\newcommand{\hval}{\mathbf{val}}

\newcommand{\trans}{\delta}

\newcommand{\hr}{$\mathsf{HR}$}

\newcommand{\rules}{\mathcal{R}}
\newcommand{\nonterm}{\mathcal{N}}

\newcommand{\edges}{{E}}

\newcommand{\tree}{T}
\newcommand{\trees}{\universeOf{T}}

\newcommand{\subtree}[2]{{#1}|_{{#2}}}

\newcommand{\twof}[1]{\mathrm{tw}({#1})}

\newcommand{\adhof}[2]{\mathrm{adh}_{#1}({#2})}

\newcommand{\unit}[2]{\overline{1}_{
        {#1}
        \ifthenelse{\equal{#2}{}}{}{,{#2}}
}}

\newcommand{\zero}[2]{\overline{0}_{
        {#1}
        \ifthenelse{\equal{#2}{}}{}{,{#2}}
}}

\newcommand{\pop}{\parallel}
\newcommand{\pexp}[2]{{#1}^{\sharp{#2}}}
\newcommand{\pexpalg}[3]{{#2}^{\sharp^{\algof{#1}}{#3}}}
\newcommand{\idemof}[1]{\mathrm{idem}({#1})}
\newcommand{\sop}{\circ}
\newcommand{\rop}{\rhd}
\newcommand{\sort}{\sigma}
\newcommand{\sorts}{\Sigma}

\newcommand{\restrict}[1]{\mathsf{restrict}_{{#1}}}

\newcommand{\rename}[1]{\mathsf{rename}_{{#1}}}

\newcommand{\extend}[1]{\mathsf{append}_{#1}}

\newcommand{\parsefunc}{\pi}

\newcommand*{\langu}{\mathcal{L}}
\newcommand*{\class}{\algof{C}}

\newcommand{\fpof}[1]{\mathtt{FP}(#1)}
\newcommand{\afp}{\mathtt{FP}}
\newcommand{\X}{\mathtt{X}}
\newcommand{\Y}{\mathtt{Y}}
\renewcommand{\P}{\mathtt{P}}
\renewcommand{\S}{\mathtt{S}}

\newcommand{\encof}[1]{\mathbf{enc}({#1})}

\newcommand{\spof}[1]{\mathbf{sp}({#1})}

\newcommand{\btwclass}[1]{\algof{G}^{\scriptscriptstyle{\leq {#1}}}}
\newcommand{\btwuniverse}[1]{\universeOf{G}^{\scriptscriptstyle{\leq {#1}}}}
\newcommand{\btwsignature}[1]{\fsignature_\algof{G}^{\scriptscriptstyle{\leq {#1}}}}
\newcommand{\singledge}[2]{{#1}_{#2}}

\newcommand{\nontermof}[1]{\theta({#1})}

\begin{document}

\title{Regular Grammars for Sets of Graphs of Tree-Width 2}

 \author{
   \IEEEauthorblockN{Marius Bozga}
   \IEEEauthorblockA{\textit{Verimag}, \textit{CNRS} \\
     Grenoble, France, \\
   marius.bozga@univ-grenoble-alpes.fr}
   \and
   \IEEEauthorblockN{Radu Iosif}
   \IEEEauthorblockA{\textit{Verimag}, \textit{CNRS} \\
     Grenoble, France, \\
     radu.iosif@univ-grenoble-alpes.fr}
   \and
   \IEEEauthorblockN{Florian Zuleger}
   \IEEEauthorblockA{\textit{TU Wien} \\
     Wien, Austria \\
     florian.zuleger@tuwien.ac.at}
 }


\maketitle

\begin{abstract}
  Regular word grammars are restricted context-free grammars that
  define all the recognizable languages of words. This paper
  generalizes regular grammars from words to certain classes of
  graphs, by defining regular grammars for unordered unranked trees
  and graphs of tree-width $2$ at most. The qualifier ``regular'' is
  justified because these grammars define precisely the recognizable
  (equivalently, \cmso-definable) sets of the respective graph
  classes. The proof of equivalence between regular and recognizable
  sets of graphs relies on the effective construction of a recognizer
  algebra of size doubly-exponential in the size of the grammar. This
  sets a \twoexptime\ upper bound on the (\exptime-hard) problem of
  inclusion of a context-free language in a regular language, for
  graphs of tree-width $2$ at most. A further syntactic restriction of
  regular grammars suffices to capture precisely the \mso-definable
  sets of graphs of tree-width $2$ at most, i.e., the sets defined by
  \cmso\ formul{\ae} without cardinality constraints. Moreover, we
  show that \mso-definability coincides with recognizability by
  algebras having an aperiodic parallel composition semigroup, for
  each class of graphs defined by a bound on the tree-width.
\end{abstract}

\section{Introduction}
Recognizability is a central pillar of the theory of formal
languages. This foundational role is underpinned by the variety of
well-established formalisms that define recognizable languages. For
instance, recognizable word languages can be represented by finite
automata, semigroups, regular expressions, regular grammars and
monadic-second order logic. These formalisms provide a robust
framework for understanding the structure of recognizable languages.

While the cases of words and trees are largely understood, the case of
graphs is currently facing several open problems. One such problem is
the lack of a notion of automaton defining a class of graph languages
that is closed under boolean operations and has a logical
characterization. A seminal development in this direction is the work
of Courcelle that introduced graph algebras
~\cite{CourcelleI,courcelle_engelfriet_2012}. In particular, the
\emph{hyperedge replacement} (\hr) algebra generalizes the notion of
recognizability from words to graphs. Courcelle established that all
\cmso-definable sets of graphs are recognizable~\cite{CourcelleI}, but
there are recognizable sets of graphs that are not
\cmso-definable. Notably, the class of graphs of tree-width at most~$2$ is chronologically the first for which recognizability in the
\hr\ algebra has been shown to be equivalent to definability in
counting monadic-second order logic (\cmso)~\cite{CourcelleV}. The
picture has been since completed by the seminal papers of
Boja\'{n}czyk and
Pilipczuk~\cite{10.1145/2933575.2934508,journals/lmcs/BojanczykP22},
that established the equivalence between recognizability and
\cmso-definability, for all classes of graphs defined by a bound on
their tree-width.

Unlike for words and trees, the equivalence between recognizability
and definability for graphs does not provide finite syntactic
descriptions of the recognizable sets (e.g., automata, regular
expression or grammars) that can be used to compute boolean operations
(union, intersection, complement) and decide emptiness of a set,
membership in a set, or inclusion between sets. For the class of
graphs of tree-width at most~$2$, the finite representation problem
has been addressed by the development of regular expressions that
capture exactly the \cmso-definable sets
thereof~\cite{DBLP:conf/icalp/Doumane22}. Moreover, grammars that
capture the \cmso-definable sets of graphs of bounded
\emph{embeddable} tree-width (i.e., an over-approximation of
tree-width that considers only decompositions whose backbones are
spanning trees of the given graph) have been
defined~\cite{Lpar24}.

Despite these recent advances, the problems of computing boolean
operations and deciding emptiness, membership and inclusion for sets
of graphs (other than trees) have received scant attention. Moreover,
the definition of finite algebraic descriptions of the \cmso-definable
sets of graphs of tree-width bounded by $k\geq3$ is currently an open
problem.



\paragraph{Contributions} (1) We introduce grammars that capture
precisely the recognizable sets for the classes of (1) unordered and
unranked trees, (2) series-parallel graphs, and, more generally, (3)
graphs of tree-width at most~$2$. Following the initial terminology of
Courcelle~\cite{CourcelleV}, we call these grammars
\emph{regular}. For each of the mentioned classes, the regular
grammars are defined by composite (derived) \hr\ operations and by
simple syntactic restrictions on the form of their rules.

Just as a regular word grammar can be converted into a finite
automaton recognizing its language, a regular graph grammar can be
converted into a finite algebra that recognizes its language.  The
elements of this algebra are finite sets of either multisets or tuples
of nonterminals, that describe (sets of) derivations of the grammar,
used to interpret the algebraic operations on graphs. For each of the
three classes of graphs considered, the size of the recognizer algebra
for a regular grammar is at most doubly exponential in the size of the
grammar. This uniform upper bound provides a \twoexptime\ algorithm
for the problem of inclusion of a context-free graph language into a
regular language (i.e., the language of a regular grammar) for the
three classes of graphs considered. To the best of our efforts, we
could not find a matching lower bound, other than the simply
exponential blowup incurred by the determinization of automata on
words or ranked trees \cite{comon:hal-03367725}.

Conversely, we prove that each recognizable set is regular in its
respective class. Moreover, a restriction on the form of the rules of
regular grammars suffices to define precisely the sets of graphs
recognized by algebras having an aperiodic parallel composition
semigroup, or simply aperiodic algebras.

(2) We investigate the relation between recognizable and
\mso-definable sets, i.e., defined by logical formul{\ae} without
modulo constraints on the cardinality of set variables. Since the
equivalence of recognizability and \cmso-definability has been already
established for tree-width bounded classes of graphs
\cite{10.1145/2933575.2934508}, we prove that the sets recognized by
aperiodic algebras, i.e., for the three classes of graphs considered,
are the same as the ones definable in \mso. A roughly similar
situation occurs for words, where recognizability by aperiodic monoids
is equivalent to definability in first-order
logic~\cite{SCHUTZENBERGER1965190}. However, this elegant result does
not generalize to trees\footnote{A characterization of the aperiodic
sets of trees definable in first-order logic is given
in~\cite{10.1145/1614431.1614435}.}, let alone graphs. Furthermore, we
generalize the equivalence between recognizability by aperiodic
algebras and \mso-definability, for all classes of graphs defined by a
bound on the tree-width.

\paragraph{Related Work} Data tree description languages such as XML
(eXtensible Markup Language) or JSON (JavaScript Object Notation) have
motivated a systematic study of automata and logics for unranked
trees, see~\cite{10.1007/11523468_4} for an overview. Two-variable
first-order logics for unranked trees with total sibling order and
equality between data values are studied
in~\cite{DBLP:journals/jacm/BojanczykMSS09}. Deterministic automata on
unranked trees with total sibling order are investigated
in~\cite{10.1007/11537311_7}. Several definitions of unranked tree
automata, whose languages are equivalent to various fragments of
\cmso\ are given in~\cite{journals/iandc/BoiretHNT17}. Notably,
bottom-up tree automata with transitions triggered by counting
constraints on the multiset of states labeling the children of the
current node have the same expressivity as \cmso. However, there is no
result on the complexity of the inclusion problem for this particular
class of automata. Subclasses with more restrictive transition
triggers have membership problems ranging from \ptime\ and
\np-hardness and inclusion problems from \ptime\ and
\conp-completeness. Algebraic recognizers for unranked (ordered and
unordered) trees are also defined
in~\cite{DBLP:conf/birthday/BojanczykW08}, in terms of forest algebras
consisting of a horizontal monoid (with disjoint union) and a vertical
monoid (with context composition). However, the focus of this work is
logical definability of unranked tree languages in the EF fragment of
the CTL temporal logic, based on certain axiomatic properties of these
monoids, such as commutativity and aperiodicity.

The regular expressions for graphs of tree-width at most~$2$,
introduced by Doumane~\cite{DBLP:conf/icalp/Doumane22}, capture the
\cmso-definable sets, just as the regular grammars we propose in this
paper. However, establishing a formal connection between these regular
expressions and the regular grammars in this paper (i.e., a Kleene
theorem for graphs of tree-width at most~$2$) requires further
study. A major difficulty is that the definition of the guard
conditions on iteration variables from a regular
expression~\cite[Definition 39]{DBLP:conf/icalp/Doumane22} is not
syntactic, but relies on a decidable semantic check~\cite[Proposition
  42]{DBLP:conf/icalp/Doumane22}. In contrast, our regular grammars
have an easily-checkable syntactic definition, where iteration
variables are confined to rules of the form
(\ref{it1:syntactic-regular}) in Definition
\ref{def:stratified-grammar}. Moreover, the proof of recognizability
of a regular language is quite different from ours:
\cite{DBLP:conf/icalp/Doumane22} establishes recognizability via
\cmso-definability, whereas we directly construct a finite recognizer
algebra from a regular grammar. In principle, it is possible to derive
an algebra of $k$-exponential size from the constructions of a \cmso{}
formula, where $k$ is the quantifier rank of the formula. However,
this analysis is stated as future work
in~\cite{DBLP:conf/icalp/Doumane22}. In contrast, each regular grammar
is guaranteed to have a recognizer algebra of size at most
$2$-exponential in the size of the grammar.  We state as an open
problem the question whether this upper bound can be further improved
and conjecture that this is the case at least for grammars in which
the rules of the form (\ref{it1:syntactic-regular}) are aperiodic.


The work~\cite{Lpar24} introduced tree-verifiable graph grammars, a
strict generalization of the regular graph grammars of
Courcelle~\cite{CourcelleV} that capture the \cmso-definable sets of
graphs of bounded embeddable tree-width. This parameter is an
over-approximation of the tree-width that uses tree decompositions
whose backbone is a spanning tree of the considered graph. The authors
of~\cite{Lpar24} already proved the equivalence with
\cmso-definability of the tree grammars considered in this paper. We
reprove their result using an explicit construction of a recognizer
algebra, which provides complexity upper bounds for the language
inclusion problem. However, our results on regular grammars for
series-parallel graphs and graphs of tree-width $2$ are orthogonal to
the results of~\cite{Lpar24}.


\section{Definitions}
The set of natural numbers is denoted by $\nat$. Given numbers
$i,j\in\nat$, we write $\interv{i}{j} \isdef \set{i, i+1, \ldots, j}$,
assumed to be empty if $i>j$. The cardinality of a finite set $A$ is
denoted by $\cardof{A}$. By writing $A \finsubseteq B$ we mean that
$A$ is a finite subset of $B$. The disjoint union $A \uplus B$ is
defined as the union of $A$ and $B$, if $A \cap B=\emptyset$, and
undefined, otherwise. For a set $A$, we denote by $\pow{A}$ its
powerset and by $\finpow{A}$ the set of finite subsets of $A$. For a
set $A$, we denote by $A^+$ the set of non-empty ordered sequences of
elements from $A$ and by $\lenof{s}$ the length of a sequence $s \in
A^+$.

We denote by $\mpow{A}$ the power-multiset of $A$, i.e., the set of
multisets $m : A \rightarrow \nat$ and by $m_1 \mcup m_2$ we denote
the union of the multisets, i.e., $(m_1 \mcup m_2)(a) \isdef m_1(a) +
m_2(a)$, for all $a \in A$. 

For a relation $R \subseteq A \times B$, we denote by $\dom{R}$ and
$\img{R}$ the sets consisting of the first and second components of
the pairs in $R$, respectively.  We write $R^{-1}$ for the inverse
relation, $R(S)$ for the image of a set $S$ via $R$ and $R(a)$ for
$R(\set{a})$, for some $a \in A$.

A bijective function $f : \nat \rightarrow \nat$ is a \emph{finite
permutation} if the set $\set{n \in \nat \mid f(n)\neq n}$ is
finite. The finite permutation $(n,m)$ swaps $n$ and $m$ mapping each
$i \in \nat \setminus \set{n,m}$ onto itself and $[i_1,\ldots,i_n]$
maps each $j\in\interv{1}{n}$ to $i_j$ and any other $j > n$ to
itself.

We denote by $\poly{x}$ the set of functions $p : \nat \rightarrow
\nat$ for which there exist $a,b,k\in\nat$ such that $p(n) \leq a\cdot
n^k+b$, for all $n \in \nat$. This notation is used in exponentiation,
e.g., $f \in 2^\poly{x}$ stands for $f(x) = 2^{p(x)}$, for some $p(x)
\in \poly{x}$.

\ifLongVersion\else The proofs of the technical results are in
Appendix \ref{app:proofs}. \fi

\subsection{Algebras and Recognizability}

Let $\sorts=\set{\sort_1,\sort_2,\ldots}$ be a set of sorts and
$\fsignature = \set{f_1, f_2, \ldots}$ be a signature of function
symbols. Each function symbol $f$ has an associated tuple of argument
sorts $\rankof{f} = \tuple{\sort_1, \ldots, \sort_n}$ and a value sort
$\sortof{f}$. The arity of $f$ is denoted $\arityof{f}\isdef n$. If
$\arityof{f}=0$ we call $f$ a constant. A variable is a sorted symbol
of arity zero, not part of the signature. We denote by
$\vars=\set{x,y,\ldots}$ the set of variables. The sort of a variable
$x \in \vars$ is denoted $\sortof{x}$. A term $t[x_1,\ldots,x_n]$ is
built as usual from function symbols and the variables
$x_1,\ldots,x_n$ of matching sorts. A term without variables is
ground. A term consisting of a single variable is trivial. A term
$t[x]$ having a single variable is a context. By $t[u_1,\ldots,u_n]$
we denote the term obtained by replacing each occurrence of $x_i$ by
$u_i$, for all $i \in \interv{1}{n}$. A subterm of $t$ is a term
$u[x_1,\ldots,x_n]$ such that $t=w\left[u[t_1,\ldots,t_n]\right]$ for
a context $w[x]$ and terms $t_1,\ldots,t_n$. Note that this is not the
usual definition of subterm, usually defined as a full subtree rooted
at some position within $t$.

An $(\fsignature,\sorts)$-\emph{algebra} is a tuple $\algof{A} =
(\set{\universeOf{A}^\sort}_{\sort\in\sorts},
\set{f^\algof{A}}_{f\in\fsignature})$, where each
$\universeOf{A}^\sort$ is the domain of sort $\sort$ and $f^\algof{A}
: \universeOf{A}^{\sort_1} \times \ldots \times
\universeOf{A}^{\sort_n} \rightarrow \universeOf{A}^{\sort}$ is the
interpretation of the function symbol $f\in\fsignature$, where
$\rankof{f} = \tuple{\sort_1, \ldots, \sort_n}$ and
$\sortof{f}=\sort$. We sometimes omit writing $\sorts$ and assume that
$\sorts=\{\sort \mid f\in\fsignature, \sort = \sortof{f} \mbox{ or } \sort \in \rankof{f} \}$. The domains of
different sorts are assumed to be disjoint and $\universeOf{A} \isdef
\biguplus_{\sort\in\sorts} \universeOf{A}^\sort$ denotes the domain of
$\algof{A}$. The algebra $\algof{A}$ is locally finite if
$\universeOf{A}_\sort$ is finite, for each $\sort\in\sorts$ and finite
if $\universeOf{A}$ is finite.

A term $t[x_1,\ldots,x_n]$ defines a function $t^\algof{A} :
\universeOf{A}^{\sortof{x_1}} \times \ldots \times
\universeOf{A}^{\sortof{x_n}} \rightarrow \universeOf{A}^{\sortof{t}}$
obtained by interpreting each function symbol from $t$ in $\algof{A}$.
An $\fsignature$-algebra is representable if its universe is the set
of values of the ground $\fsignature$-terms. Each algebra considered
in this paper will be representable.

Given $\fsignature$-algebras $\algof{A}$ and $\algof{B}$, a
\emph{homomorphism} between $\algof{A}$ and $\algof{B}$ is a function
$h : \universeOf{A} \rightarrow \universeOf{B}$ such that
$\sortof{h(a)}=\sortof{a}$, for all $a \in \universeOf{A}$ and
$h(f^\algof{A}(a_1,\ldots,a_n))=f^\algof{B}(h(a_1),\ldots,h(a_n))$,
for all $f \in \fsignature$ and elements $a_1,\ldots,a_n \in
\universeOf{A}$ of matching sorts. Note that a homomorphism between
multi-sorted algebras is sort-preserving. Assuming that $\algof{A}$
and $\algof{B}$ are representable algebras, it can be shown that, if a
homomorphism between $\algof{A}$ and $\algof{B}$ exists\footnote{For
instance, there is no homomorphism between the $\fsignature$-algebras
$\algof{A}$ and $\algof{B}$ if $c_1^\algof{A}=c_2^\algof{A}$ and
$c_1^\algof{B} \neq c_2^\algof{B}$, for distinct constants $c_1,c_2
\in \fsignature$.}, then it is unique and surjective. In this case,
the unique homomorphism between $\algof{A}$ and $\algof{B}$ is denoted
$\homof{\algof{A}}{\algof{B}}$. The following notions are central to
this paper:

\begin{definition}\label{def:homo-rec}
  Let $\algof{A}$ be an $\fsignature$-algebra.  An
  $\fsignature$-\emph{recognizer} for $\algof{A}$ is a pair
  $(\algof{B},C)$, where $\algof{B}$ is a locally finite
  $\fsignature$-algebra, $C \subseteq \universeOf{B}$ and the
  homomorphism $\homof{\algof{A}}{\algof{B}}$ exists. A set
  $\langu\subseteq\universeOf{A}$ is recognized by $(\algof{B},C)$ iff
  $\langu=\homof{\algof{A}}{\algof{B}}^{-1}(C)$. A set is
  \emph{recognizable} in $\algof{A}$ if
  $\langu=\homof{\algof{A}}{\algof{B}}^{-1}(C)$, for an
  $\fsignature$-recognizer $(\algof{B},C)$.
\end{definition}

The recognizable sets of an algebra are closed under union,
intersection and complement with respect to the domain of that
algebra. These operations are effectively computable when the
recognizer algebra is effectively computable.

We say that a $(\fsignature_\algof{B},\sorts_\algof{B})$-algebra
$\algof{B}$ is derived from the
$(\fsignature_\algof{A},\sorts_\algof{A})$-algebra $\algof{A}$ if
$\sorts_\algof{B} \subseteq \sorts_\algof{A}$ and, for every function
symbol $f \in \fsignature_\algof{B}$, there exists some
$\fsignature_\algof{A}$-term $t_f$ such that $f^\algof{B} =
t_f^\algof{A}$. It is easy to prove that recognizability in an algebra
implies recognizability in each of its derived algebras, but not
viceversa. For instance, the word language $\set{a^nb^n \mid
  n\in\nat}$ is recognizable in the algebra of words with operation $x
\mapsto axb$, but not in the classical monoid of words with
concatenation.

\ifLongVersion\else Additional material concerning recognizability is
given in Appendix \ref{app:rec}.  \fi

\begin{recTextEnd}
\begin{lemma}\label{lemma:derived-sub-rec}
  Let $\algof{A}$ be an algebra and $\algof{B}$ be an algebra derived
  from $\algof{A}$. Then, $\langu\subseteq\universeOf{B}$ is
  recognizable in $\algof{B}$, if it is recognizable in~$\algof{A}$.
\end{lemma}
\begin{proof}\noindent\emph{Proof.}
  Let $\fsignature_\algof{A}$ and $\fsignature_\algof{B}$ be the
  signatures of $\algof{A}$ and $\algof{B}$, respectively. By
  assumption we have that for every $f \in \fsignature_\algof{B}$
  there is some first-order $\fsignature_\algof{A}$-term $t_f$ such
  that $f^\algof{B} = t_f^\algof{A}$.  Let $(\algof{D},E)$ be a
  $\fsignature_\algof{A}$-recognizer, such that
  $\langu=\homof{\algof{A}}{\algof{D}}^{-1}(E)$. We define the
  $\fsignature_\algof{B}$-recognizer $(\algof{I},J)$, where $\algof{I} \isdef (\set{\universeOf{I}^\sort}_{\sort\in\sorts_\algof{B}},
    \set{f^\algof{I}}_{f \in \fsignature_\algof{B}})$ and
    \begin{align*}
      \universeOf{I}^\sort \isdef & ~\set{t^\algof{D} \mid t \text{ ground $\fsignature_\algof{B}$-term, } \sortof{t}=\sort} \\
      f^\algof{I} \isdef & ~t_f^\algof{D}
    \end{align*}
    for all $\sort\in\sorts_\algof{B}$, where
    \begin{align*}
    \sorts_\algof{B} = & ~\set{\sort_1,\ldots,\sort_n,\sortof{f} \mid f \in \fsignature_\algof{B},~
      \rankof{f}=\tuple{\sort_1,\ldots,\sort_n}}
    \end{align*}
    and $J \isdef \universeOf{I} \cap E$. We prove that
    $\langu=\homof{\algof{B}}{\algof{I}}^{-1}(J)$.  We consider some
    $x \in \universeOf{B}$.  By assumption that we consider only
    representable algebras, there exists a ground
    $\fsignature_\algof{B}$-term $t$ such that $x=t^\algof{B}$.  Let
    $u$ be the $\fsignature_\algof{A}$-term obtained by expanding the
    terms from $\fsignature_\algof{B}$, i.e., $x=u^\algof{A}$.  Then,
    $\homof{\algof{B}}{\algof{I}}(x) =
    \homof{\algof{B}}{\algof{I}}(t^\algof{B}) = t^\algof{I} =
    u^\algof{D} = \homof{\algof{A}}{\algof{D}}(u^\algof{A}) =
    \homof{\algof{A}}{\algof{D}}(x)$.  Hence, $x \in
    \homof{\algof{B}}{\algof{I}}^{-1}(J)$ iff $x \in
    \homof{\algof{A}}{\algof{D}}^{-1}(E)$.  With
    $\langu\subseteq\universeOf{B}$, we obtain that $\langu =
    \homof{\algof{A}}{\algof{D}}^{-1}(E) =
    \homof{\algof{B}}{\algof{I}}^{-1}(J)$.
\end{proof}

An equivalent definition of recognizability in a
$(\fsignature_\algof{A},\sorts_\algof{A})$-algebra $\algof{A}$ uses
congruences. For a given set $\langu\subseteq\universeOf{A}$, we
denote by $\requiv{\langu} \subseteq \universeOf{A} \times
\universeOf{A}$ the equivalence relation $a_1 \requiv{\langu} a_2$ if
and only if $\sortof{a_1} = \sortof{a_2}$ and $t^\algof{A}[a_1] \in
\langu \Leftrightarrow t^\algof{A}[a_2] \in \langu$, for each
$\fsignature_\algof{A}$-context $t[x]$. We say that $\requiv{\langu}$
is (\emph{locally}) \emph{finite} if it has finitely many equivalence
classes (of each sort). It is routine to show that $\requiv{\langu}$
is a \emph{congruence} for each set $\langu$, i.e., for each function
symbol $f \in \fsignature$ and all elements
$a_1,b_1,\ldots,a_{\arityof{f}},b_{\arityof{f}} \in \universeOf{A}$,
if $a_i \requiv{\langu} b_i$, for all $i\in\interv{1}{\arityof{f}}$
then $f^\algof{A}(a_1,\ldots,a_\arityof{f}) \requiv{\langu}
f^\algof{A}(b_1,\ldots,b_\arityof{f})$. Thus, $\requiv{\langu}$ is
called the \emph{syntactic congruence} of $\langu$. For each element
$a\in\universeOf{A}$, we denote by $[a]_\requiv{\langu}$ the
equivalence class of $a$ and by $\algof{A}_\requiv{\langu}$ the
$(\fsignature_\algof{A}, \sorts_\algof{A})$-algebra whose elements are
the equivalence classes, of sorts $\sortof{[a]_\requiv{\langu}} =
\sortof{a}$. The function symbols $f \in \fsignature_\algof{A}$ are
interpreted in $\algof{A}_\requiv{\langu}$ as follows\footnote{The
function $f^{\algof{A}_\requiv{\langu}}$ is well defined because
$\requiv{\langu}$ is a congruence relation.}:
\begin{align*}
  f^{\algof{A}_\requiv{\langu}}([a_1]_\requiv{\langu}, \ldots,
  [a_\arityof{f}]_\requiv{\langu}) \isdef &
  ~[f^\algof{A}(a_1,\ldots,a_\arityof{f})]_\requiv{\langu}
\end{align*}
Let $h_\requiv{\langu} : \universeOf{A} \rightarrow
\universeOf{A}_\requiv{\langu}$ the function that maps each element $a
\in \universeOf{A}$ to its equivalence class $h_\requiv{\langu} \isdef [a]_\requiv{\langu}$.
Then, $h_\requiv{\langu}$ is a homomorphism, by the definition of
$\algof{A}_\requiv{\langu}$. If $\requiv{\langu}$ is locally finite,
it is routine to show that $(\algof{A}_\requiv{\langu},\{
[a]_\requiv{\langu} \mid a \in \langu\})$ is a recognizer for
$\langu$, called the \emph{syntactic recognizer} of $\langu$. For
self-containment reasons, we prove a well-known fact,
namely that the syntactic recognizer is the minimal
recognizer of that set, which can be stated as follows:

\begin{lemma}\label{lemma:syntactic-congruence-homomorphism}
  Let $\langu\subseteq\universeOf{A}$ be a recognizable set in a
  representable $(\fsignature,\sorts)$-algebra $\algof{A}$ and
  $(\algof{I},J)$ be a recognizer for $\langu$, for a representable
  $(\fsignature,\sorts)$-algebra $\algof{I}$, witnessed by the
  homomorphism $\homof{A}{I}$. Then, $h_\requiv{\langu}$ factors
  through $\homof{A}{I}$ via some homomorphism $g$ between $\algof{I}$
  and $\algof{A}_\requiv{\langu}$, i.e., $h_\requiv{\langu} = g \circ
  \homof{A}{I}$.
\end{lemma}
\begin{proof}\noindent\emph{Proof.}
  Because the algebras $\algof{A}$ and $\algof{I}$ are representable,
  the homomorphism $\homof{A}{I} : \universeOf{A} \rightarrow
  \universeOf{I}$ is surjective: for each element $i \in
  \universeOf{I}$ there exists a $\fsignature_\algof{A}$-term $t$ such
  that $t^\algof{I} = i$, hence $\homof{A}{I}(t^\algof{A}) =
  t^\algof{I} = i$ and $t^\algof{A} \in \universeOf{A}$. Let $\kappa
  \subseteq \homof{A}{I}^{-1}$ be a function (such a function exists
  because $\homof{A}{I}$ is surjective) and define $g \isdef
  h_\requiv{\langu} \circ \kappa$. First, we prove the following fact:

  \begin{fact}\label{fact:syntactic-congruence-homomorphism}
    $\homof{A}{I}(a_1)=\homof{A}{I}(a_2)$ only if $a_1 \requiv{\langu}
    a_2$, for all $a_1,a_2\in\universeOf{A}$.
  \end{fact}
  \begin{proof} \noindent\emph{Proof.} Assume that $\homof{A}{I}(a_1)=\homof{A}{I}(a_2)$. Because
    $\langu=\homof{A}{I}^{-1}(J)$, then either $a_1,a_2 \in \langu$ or
    $a_1,a_2\not\in\langu$, thus $a_1 \requiv{\langu} a_2$.
  \end{proof}

  \vspace*{.5\baselineskip}
  \noindent We prove first that $h_\requiv{\langu}=g \circ
  \homof{A}{I}$.  Let $a \in \universeOf{A}$ be an element. Then,
  $g(\homof{A}{I}(a))=h_\requiv{\langu}(\kappa(\homof{A}{I}(a)))$, by
  the definition of $g$. Since $\kappa\subseteq\homof{A}{I}^{-1}$, we
  obtain $\homof{A}{I}(\kappa(\homof{A}{I}(a)))=\homof{A}{I}(a)$,
  hence $\kappa(\homof{A}{I}(a)) \requiv{\langu} a$, by Fact
  \ref{fact:syntactic-congruence-homomorphism}. Thus, we obtain
  $g(\homof{A}{I}(a)) = h_\requiv{\langu}(\kappa(\homof{A}{I}(a))) =
  h_\requiv{\langu}(a)$. Since the choice of $a\in \universeOf{A}$ was
  arbitrary, this proves $h_\requiv{\langu} = g \circ \homof{A}{I}$.

  Second, we prove that $g$ is a homomorphism between $\algof{I}$ and
  $\algof{A}_\requiv{\langu}$. To this end, let $i \isdef
  f^\algof{I}(i_1,\ldots,i_\arityof{f}) \in \universeOf{I}$ be an
  element, for some $i_1,\ldots,i_\arityof{f}$. Because $\algof{I}$ is
  a representable algebra, there exists an $\fsignature$-term
  $t=f(t_1,\ldots,t_\arityof{f})$ such that $i=t^\algof{I}$ and
  $t_j^\algof{I} = i_j$, for all $j \in \interv{1}{\arityof{f}}$.  Let
  $\sim_{\homof{A}{I}}$ be the kernel of $\homof{A}{I}$, i.e., the
  equivalence relation $a_1 \sim_{\homof{A}{I}} a_2 \iff
  \homof{A}{I}(a_1)=\homof{A}{I}(a_2)$, for all
  $a_1,a_2\in\universeOf{A}$.

  \begin{fact}\label{fact:inverse-homomorphism}
    There exist elements $a,a_1,\ldots,a_\arityof{f} \in
    \universeOf{A}$ such that $a =
    f^\algof{A}(a_1,\ldots,a_\arityof{f})$, $a \sim_{\homof{A}{I}}
    \kappa(i)$ and $a_j \sim_{\homof{A}{I}} \kappa(i_j)$, for all $i
    \in \interv{1}{\arityof{f}}$.
  \end{fact}
  \begin{proof} \noindent\emph{Proof.}
    By induction on the structure of $t=f(t_1,\ldots,t_\arityof{f})$.
    By the inductive hypothesis, there exist $a_j \sim_{\homof{A}{I}}
    \kappa(i_j)$, for all $j \in \interv{1}{\arityof{f}}$, and let $a
    \isdef f^\algof{A}(a_1,\ldots,a_\arityof{f})$. We prove $a \sim_{\homof{A}{I}} \kappa(i)$:
    \begin{align*}
      \homof{A}{I}(a) = & ~\homof{A}{I}(f^\algof{A}(a_1,\ldots,a_\arityof{f})) \\
      = & ~f^\algof{I}(\homof{A}{I}(a_1), \ldots,\homof{A}{I}(a_\arityof{f})) \\
      & \text{because } \homof{A}{I} \text{ is a homomorphism} \\
      = & ~f^\algof{I}(\homof{A}{I}(\kappa(i_1)), \ldots,\homof{A}{I}(\kappa(i_\arityof{f}))) \\
      & \text{because } a_j \sim_{\homof{A}{I}} \kappa(i_j) \text{, for all } j \in \interv{1}{\arityof{f}} \\
      = & ~f^\algof{I}(i_1, \ldots,i_\arityof{f}) \text{, because } \kappa \subseteq \homof{A}{I}^{-1} \\
      = & ~i = \homof{A}{I}(\kappa(i))
    \end{align*}
  \end{proof}

  Back to the proof, we compute:
  \begin{align*}
    g(f^\algof{I}(i_1,\ldots,i_\arityof{f})) = & ~h_\requiv{\langu}(\kappa(i)) = h_\requiv{\langu}(a) \text{, by Fact \ref{fact:inverse-homomorphism}} \\
    = & ~f^{\algof{A}_\requiv{\langu}}(h_\requiv{\langu}(a_1), \ldots, h_\requiv{\langu}(a_\arityof{f})) \\
    & \text{because $a_j \sim_{\homof{A}{I}} \kappa(i_j)$} \\
    & \text{by Fact \ref{fact:syntactic-congruence-homomorphism} and $\sim_{\homof{A}{I}} \subseteq \requiv{\langu}$, by Fact \ref{fact:inverse-homomorphism}} \\
    = & ~f^{\algof{A}_\requiv{\langu}}(h_\requiv{\langu}(\kappa(i_1)), \ldots, h_\requiv{\langu}(\kappa(i_\arityof{f}))) \\
    = & ~f^{\algof{A}_\requiv{\langu}}(g(i_1), \ldots, g(i_\arityof{f}))
  \end{align*}
  Since the choice of $f \in \fsignature$ and $i_1, \ldots,
  i_\arityof{f} \in \universeOf{I}$ was arbitrary, we obtain that $g$
  is a homomorphism between $\algof{I}$ and
  $\algof{A}_\requiv{\langu}$.
\end{proof}

We further prove the following factorization property of the
$h_\requiv{\langu}$ homomorphism:
\begin{lemma}\label{lemma:factorization-of-related-syntactic-congruences}
  Let $\algof{A}$ and $\algof{B}$ be two representable
  $(\fsignature,\sorts)$-algebras such that the homomorphism
  $\homof{A}{B}$ exists. Let $\langu_B \subseteq \universeOf{B}$ be a
  set and $\langu_A \isdef \homof{A}{B}^{-1}(\langu_B)$. Then, the
  algebras $\algof{A}_\requiv{\langu_A}$ and
  $\algof{B}_\requiv{\langu_B}$ are related by an isomorphism $g$ and
  $h_\requiv{\langu_A}$ factors through $\homof{A}{B}$ via
  $h_\requiv{\langu_B}$ up to the isomorphism of the co-domains, i.e.,
  $g \circ h_\requiv{\langu_A} = h_\requiv{\langu_B} \circ
  \homof{A}{B}$.
\end{lemma}
\begin{proof} \noindent\emph{Proof.}
  We show first that: \begin{align*}
    a_1 \requiv{\langu_A} a_2 \iff & ~\homof{A}{B}(a_1)
    \requiv{\langu_B} \homof{A}{B}(a_2) ~(\dagger)
  \end{align*}
  ``$\Rightarrow$'' Let $a_1,a_2 \in \universeOf{A}$ such that $a_1
  \requiv{\langu_A} a_2$. Then, $\sortof{\homof{A}{B}(a_1)} =
  \sortof{a_1} = \sortof{a_2} = \sortof{\homof{A}{B}(a_2)}$. Let
  $t[x]$ be an $\fsignature$-context. We establish the
  equivalence: \begin{align*}
    t^\algof{B}[\homof{A}{B}(a_1)] = \homof{A}{B}(t^\algof{A}[a_1]) \in & ~\langu_B \\
    \iff & \text{ because } \langu_A = \homof{A}{B}^{-1}(\langu_B) \\
    t^\algof{A}[a_1] \in & ~\langu_A \\
    \iff & \text{ because } a_1 \requiv{\langu_A} a_2 \\
    t^\algof{A}[a_2] \in & ~\langu_A \\
    \iff & \text{ because } ~\langu_A = \homof{A}{B}^{-1}(\langu_B) \\
    t^\algof{B}[\homof{A}{B}(a_2)] =  \homof{A}{B}(t^\algof{A}[a_2]) \in & ~\langu_B
  \end{align*}

  \noindent ``$\Leftarrow$'' Let $a_1,a_2 \in \universeOf{A}$ such
  that $\homof{A}{B}(a_1) \requiv{\langu_B} \homof{A}{B}(a_2)$. Then,
  $\sortof{a_1} = \sortof{\homof{A}{B}(a_1)} =
  \sortof{\homof{A}{B}(a_2)} = \sortof{a_2}$. Let $t[x]$ be an
  $\fsignature$-context. We establish the equivalence: \begin{align*}
    t^\algof{A}[a_1] \in & \langu_A \\
    \iff & \text{ because } \langu_A = \homof{A}{B}^{-1}(\langu_B) \\
    \homof{A}{B}(t^\algof{A}[a_1]) \in & \langu_B \\
    \iff & \text{ because } \homof{A}{B}(a_1) \requiv{\langu_B} \homof{A}{B}(a_2) \\
    \homof{A}{B}(t^\algof{A}[a_2]) \in & \langu_B \\
    \iff & \text{ because } \langu_A = \homof{A}{B}^{-1}(\langu_B) \\
    t^\algof{A}[a_2] \in & \langu_A
  \end{align*}

  \noindent We define the function $g :
  \universeOf{A}_\requiv{\langu_A} \rightarrow
  \universeOf{B}_\requiv{\langu_B}$ as $g([a]_\requiv{\langu_A})
  \isdef [\homof{A}{B}(a)]_\requiv{\langu_B}$, for each $a \in
  \universeOf{A}$. By ($\dagger$), $g$ is bijective. To prove that $g$
  is an isomorphism between $\algof{A}_\requiv{\langu_A}$ and
  $\algof{B}_\requiv{\langu_A}$, we observe that, for each function
  symbol $f \in \fsignature$: \begin{align*}
    & ~g(f^{\algof{A}_\requiv{\langu_A}}([a_1]_\requiv{\langu_A},\ldots,[a_\arityof{f}]_\requiv{\langu_A})) \\
    = & ~g([f^\algof{A}(a_1,\ldots,a_\arityof{f})]_\requiv{\langu_A}) \\
    = & ~[\homof{A}{B}(f^\algof{A}(a_1,\ldots,a_\arityof{f}))]_\requiv{\langu_B} \\
    = & ~f^{\algof{B}_\requiv{\langu_B}}([\homof{A}{B}(a_1)]_\requiv{\langu_B},\ldots,[\homof{A}{B}(a_\arityof{f})]_\requiv{\langu_B}) \\
    = & ~f^{\algof{B}_\requiv{\langu_B}}(g([a_1]_\requiv{\langu_A}), \ldots, g([a_\arityof{f}]_\requiv{\langu_A}))
  \end{align*}
  The equality $g \circ h_\requiv{\langu_A} = h_\requiv{\langu_B} \circ
  \homof{A}{B}$ follows from the definition of $g$.
\end{proof}
\end{recTextEnd}

\subsection{Graphs}
\label{subsec:graphs}

\begin{figure*}[t!]
  \centerline{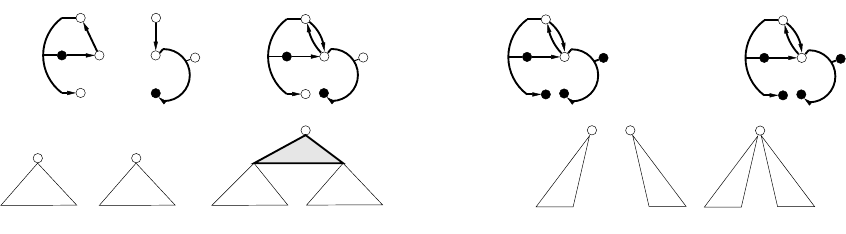}
  \caption{Composition (a), Restriction (b) and Renaming (c) of
    Graphs. Append (d) and Composition (e) of Trees. Sources are
    denoted by hollow circles. Each arrow indicates the order of
    vertices attached to the corresponding edge. }
  \label{fig:graphs}
  \vspace*{-\baselineskip}
\end{figure*}

Let $\alphabet$ be a fixed finite alphabet of edge labels, having
associated arities $\arityof{a} \geq 1$ for each $a\in\alphabet$.  The
following definition introduces hypergraphs, where edges are attached
to ordered sequences of vertices, of length one or more, in which a
vertex can occur multiple times. For simplicity, we call these objects
graphs, in the rest of the paper.

\begin{definition}\label{def:graphs}
  A \emph{graph} of sort $\slabs\finsubseteq \nat$ is a tuple $\graph
  = \tuple{\vertof{}, \edgeof{}, \labof{}, \edgerelof{},
    \sourceof{}}$, where $\vertof{}$ and $\edgeof{}$ are finite sets
  of vertices and edges, respectively, $\labof{} : \edgeof{}
  \rightarrow \alphabet$ and $\edgerelof{} : \edgeof{} \rightarrow
  \vertof{}^+$ give the label and sequence of vertices attached to
  each edge, respectively, such that $\arityof{\labof{}(e)} =
  \lenof{\edgerelof{}(e)}$, for each edge $e \in \edgeof{}$, and
  $\sourceof{} : \slabs \rightarrow \vertof{}$ is a one-to-one mapping
  that designates the sources of $\graph$. A vertex $\sourceof{}(s)$
  is called the $s$-source of $\graph$. We write $\vertof{\graph}$,
  $\edgeof{\graph}$, $\labof{\graph}$, $\edgeof{\graph}$ and
  $\sourceof{\graph}$ for the elements of $\graph$.
\end{definition}

\begin{definition}\label{def:hr}
  The \emph{graph algebra} $\algof{G}$ has the signature:
  \vspace*{-.5\baselineskip}
  \begin{align*}
  \fsignature_\algof{G} \isdef & ~\set{\emptygraph_\slabs, \restrict{\slabs} \mid \slabs\finsubseteq\nat} ~\cup \\
  & ~\set{\sgraph{a}_{s_1, \ldots, s_n} \mid a \in \alphabet,~ \arityof{a} =n,~ s_1, \ldots, s_n \in \nat} ~\cup \\
  & ~\set{\rename{\alpha} \mid \alpha \text{ finite permutation of } \nat} ~\cup \\
  & ~\set{\pop_{\slabs_1,\slabs_2} \mid \slabs_1,\slabs_2 \finsubseteq \nat}
  \end{align*}
  interpreted over the domain $\universeOf{G}$ of graphs as follows:
  \begin{compactitem}[-]
  \item $\emptygraph^{\algof{G}}_\slabs$ is the graph of sort $\slabs$
    consisting of one $s$-source, for each $s\in\slabs$ and no edges,
  \item $\sgraph{a}^{\algof{G}}_{s_1, \ldots, s_n}$ is the graph of
    sort $\set{s_1,\ldots,s_n}$ consisting of a single $a$-labeled
    edge attached to the of the sequence of sources $s_1,\ldots,s_n$ (note that the sources $s_1,\ldots,s_n$ are not required to be pairwise distinct),
  \item $\restrict{\slabs}^{\algof{G}}(x)$ forgets the source labels
    $s\in\sortof{x} \setminus \slabs$ from $x$,
  \item $\rename{\alpha}^{\algof{G}}(x)$ renames the sources of $x$
    according to $\alpha$,
  \item $x_1 \pop^\algof{G}_{\sortof{x_1},\sortof{x_2}} x_2$ is the
    disjoint union of $x_1$ and $x_2$, followed by joinining the
    $s$-sources, respectively, for all $s \in \sortof{x_1} \cap
    \sortof{x_2}$.
  \end{compactitem}
\end{definition}

\begin{example}
For example, the leftmost graph in Figure \ref{fig:graphs} (a) has
four vertices of which tree sources labeled $1$, $2$ and $3$ and three
edges labeled $a$, $b$ and $c$. The $a$-labeled edge is attached to
three vertices, whereas the $b$- and $c$-labeled edges are binary. The
middle graph is of sort $\set{1,2,4}$ and the rightmost one of sort
$\set{1,2,3,4}$. Figure \ref{fig:graphs} (a) shows the result of the
composition of two graphs, whereas (b) and (c) show the result of
applying restriction and renaming to this composition.
\end{example}

Note that the parallel composition of graphs of equal sorts
$\pop^\algof{G}_{\slabs,\slabs}$ is associative and commutative, for
each sort $\slabs\finsubseteq\nat$. For simplicity, we write $\pop$
instead of $\pop_{\slabs,\slabs}$ if $\slabs$ can be inferred from the
context. We use the syntactic shorthands:
\begin{align*}
  \pexp{x}{0} = \emptygraph_{\sortof{x}}
  \hspace*{5mm}
  \pexp{x}{n}\isdef\underbrace{x \pop \ldots \pop x}_{
    \scriptscriptstyle{n\geq1 \text{ times}}
  }
\end{align*}
The following definition is used to introduce trees, series-parallel
graphs and graphs of tree-width at most~$2$:

\begin{definition}\label{def:class}
  A \emph{class} $\class$ of graphs is a derived
  $(\fsignature_\class,\sorts_\class)$-algebra of $\algof{G}$, where
  $\fsignature_\class$ is a finite signature and $\sorts_\class$ is
  finite set of sorts such that $\pop_{\slabs,\slabs} \in
  \fsignature_\class$, for each $\slabs \in \sorts_\class$.
\end{definition}

For each integer $k \geq 1$, we define the subalgebra $\btwclass{k}$
over the signature $\btwsignature{k}$ consisting of the function
symbols from $\fsignature_\algof{G}$ restricted to sorts that are
subsets of $\interv{1}{k+1}$, having the same interpretation as in
$\algof{G}$. Since $\btwsignature{k}$ is a finite signature that
contains $\pop_{\slabs,\slabs}$, for each $\slabs \subseteq
\interv{1}{k+1}$, we obtain that each algebra $\btwclass{k}$ is a
class. Note that $\algof{G}$ is not a class because the signature
$\fsignature_\algof{G}$ is infinite.

We also consider aperiodic algebras over signatures containing
parallel composition(s). Given an algebra $\algof{A}$ having sorts
$\sorts_\algof{A}$ over a signature $\fsignature_\algof{A}$ containing
$\pop_{\slabs,\slabs}$, for each sort $\slabs \in \sorts_\algof{A}$,
an element $a \in \universeOf{A}$ is idempotent if $a \pop^\algof{A} a
= a$. An idempotent power of $a$ is an idempotent element of the form
$\pexp{a}{n}$, for some $n\geq1$. It is easy to see that each element
$a$ of a locally finite algebra has a unique idempotent power, denoted
$\idemof{a}$.

\begin{definition}\label{def:aperiodic-recognizer}
  A locally finite algebra $\algof{A}$ is \emph{aperiodic} if and only if
  $\idemof{a} \pop^\algof{A} a = \idemof{a}$, for each element $a \in
  \universeOf{A}$. A recognizer $(\algof{A},B)$ is aperiodic if the
  algebra $\algof{A}$ is aperiodic.
\end{definition}

\subsection{Trees}
\label{subsec:trees}

Trees form a class of graphs having a single sort $\set{1}$, where the
$1$-source of a tree is its root. Following the definitions of
Courcelle~\cite{CourcelleV}, we consider trees with edges of arity two
or more, where the first vertex attached to an edge is the parent of
the latter vertices attached to the same edge. The children of a node
are thus grouped into several sets of same-edge siblings. The set of
trees is generated by operations from Figure \ref{fig:graphs} (d) and
(e), starting from graphs with a single vertex and no edges.

\begin{definition}\label{def:trees}
  The class $\algof{T}$ of trees has signature:
  \begin{align*}
    \fsignature_\algof{T} \isdef \set{\extend{b} \mid b \in \alphabet} \cup \set{\pop,\emptygraph_{\set{1}}}
  \end{align*}
  where $\extend{b}^\algof{T}(x_1,\ldots,x_{\arityof{b}-1})$ joins the
  roots of the trees $x_1,\ldots,x_{\arityof{b}-1}$ to the positions
  $2,\ldots,\arityof{b}$ of an edge labeled $b$, whose first position
  becomes the root of the resulting tree. The interpretations of
  $\pop$ and $\emptygraph_{\set{1}}$ are the same in the tree algebra
  $\algof{T}$ as in the graph algebra $\algof{G}$ (Definition
  \ref{def:hr}).
\end{definition}

It is easy to see that $\algof{T}$ is indeed a derived algebra of
$\algof{G}$\footnote{
$\extend{b}^\algof{T}(x_1,\ldots,x_{\arityof{b}-1}) \isdef
\restrict{\set{1}}^\algof{G}(\sgraph{b}^\algof{G}_{1,2,\ldots,\arityof{b}}
\pop^\algof{G} \rename{(1,2)}^\algof{G}(x_1) \pop^\algof{G} \ldots
\rename{(1,\arityof{b})}^\algof{G}(x_{\arityof{b}-1}))$}. The standard
terminology is immediately retrieved from the above definition. The
vertices of a tree are called \emph{nodes}. For each edge, the first
vertex attached to it is the \emph{parent} of the other vertices
called \emph{children}. Note that a node may be attached to several
edges on the first position and hence may have children corresponding
to different edges. The \emph{rank} of a tree is the maximum number of
children of a node. A set of trees is \emph{ranked} if the
corresponding set of ranks is finite and \emph{unranked},
otherwise. In particular, the domain of $\algof{T}$ is unranked. A
binary tree is a tree with binary edges only. 

Binary trees are used in the definition of the tree-width parameter of
a graph. A set of nodes $N$ in a tree is connected iff between any two
nodes in $N$ there exists an undirected path of edges that traverses
only nodes from $N$.

\begin{definition}\label{def:tree-decomposition}
  A tree decomposition of a graph $\graph$ is a pair $(\tree,\beta)$,
  where $\tree$ is a binary tree and $\beta : \vertof{\tree}
  \rightarrow \pow{\vertof{\graph}}$ is a mapping of the nodes of $\trees$ with bags
  of vertices from $\graph$, such
  that:~\begin{inparaenum}[(1)]
  \item\label{it0:tree-decomposition} the sources of $\graph$ are contained
    in some bag,
    %
  \item\label{it1:tree-decomposition} for each edge of $\graph$, the vertices
    attached to it are contained in some bag,
    %
  \item\label{it2:tree-decomposition} for each vertex of $\graph$, the
    nodes of $\tree$ whose bags contain it form a non-empty connected
    component of $\tree$.
    %
  \end{inparaenum}
  The width of the tree decomposition is $\width{\tree,\beta} \isdef
  \max\{\cardof{\beta(n)} \mid n \in \vertof{\tree}\}-1$ and the
  \emph{tree-width} of $\graph$ is $\twd{\graph} \isdef \min
  \{\width{\tree,\beta} \mid (\tree,\beta) \text{ tree decomposition
    of } \graph\}$.
\end{definition}

In this paper we consider graphs of tree-width at most~$2$, and some proper subclasses such as (disoriented) series-parallel graphs.

\section{Grammars}
\label{sec:grammars}

We aim at using grammars to define the recognizable sets of graphs
pertaining to certain classes, e.g., trees, series-parallel and
tree-width $2$ graphs. Given a signature $\fsignature$, a
$\fsignature$-\emph{grammar} $\grammar=(\nonterm,\rules)$ consists of
a finite set $\nonterm$ of sorted nonterminals and a finite set $\rules$ of
rules of the form (i) $X \rightarrow t[X_1,\ldots,X_n]$, where
$X,X_1,\ldots,X_n \in \nonterm$ and $t$ is an $\fsignature$-term of
sort $\sortof{X}$, or (ii) $\rightarrow X$, for some $X \in \nonterm$,
called axioms. We denote by $\sizeof{\grammar}$ the total number of
occurrences of a nonterminal or function symbol in some rule of
$\rules$.

Given terms $u$ and $w$, a derivation step $u \step{\grammar} w$
obtains $w$ by replacing an occurrence of a variable $X$ in $u$ by a
term $t$, where $X\rightarrow t$ is a rule of $\grammar$. A
$X$-derivation of $t$ is a finite sequence of derivation steps $X
\step{\grammar}^* t$. The derivation is complete iff $t$ is a ground
term. Given an $\fsignature$-algebra $\algof{A}$, the language of
$\grammar$ in $\algof{A}$ is $\alangof{}{A}{\grammar} \isdef
\bigcup_{\rightarrow X \in \rules} \alangof{X}{\algof{A}}{\grammar}$,
where $\alangof{X}{A}{\grammar} \isdef \set{t^\algof{A} \mid X
  \step{\grammar}^* t \text{ complete derivation}}$. A set is
context-free iff it is the language of a grammar.

Alternatively, a context-free set is the set of evaluations of a
recognizable set of terms in a given algebra. Note that the rules of a
grammar are an equivalent representation of a set of terms recognized
by a tree automaton, see~\cite{comon:hal-03367725}.

The filtering of an $\fsignature$-grammar $\grammar=(\nonterm,\rules)$
by an $\fsignature$-recognizer $(\algof{B},C)$ is the
$\fsignature$-grammar $\grammar^{(\algof{B},C)}$ having a nonterminal
$X^b$ for each $X \in \nonterm$ and $b \in \universeOf{B}$ and rules
either (i) $X^b \rightarrow t[X_1^{b_1}, \ldots, X_n^{b_n}]$, for each
rule $X \rightarrow t[X_1,\ldots,X_n] \in \rules$ and each sequence of
elements $b, b_1, \ldots, b_n \in \universeOf{B}$, such that $b =
t^\algof{B}(b_1,\ldots,b_n)$, or (ii) $\rightarrow X^c$, for each
axiom $\rightarrow X \in \rules$ and each element $c \in C$. The
following result is also known as the \emph{Filtering Theorem}:

\begin{theorem}[Theorem 3.88 in~\cite{courcelle_engelfriet_2012}]
\label{thm:filtering}
  Let $\algof{A}$ be an $\fsignature$-algebra.  For each
  $\fsignature$-grammar~$\grammar=(\nonterm,\rules)$ and
  $\fsignature$-recognizer $(\algof{B},C)$ for $\algof{A}$, we have
  $\alangof{}{\algof{A}}{\grammar^{(\algof{B},C)}}=\alangof{}{\algof{A}}{\grammar}
  \cap \homof{\algof{A}}{\algof{B}}^{-1}(C)$. Moreover,
  $\grammar^{(\algof{B},C)}$ can be built in time
  $\bigO(\cardof{\rules} \cdot
  \cardof{\universeOf{B}}^\cardof{\grammar}\cdot
  \Omega^{\sizeof{\grammar}})$, where $\Omega$ is an upper bound on
  the time needed to compute $f^\algof{A}(b_1,\ldots,b_\arityof{f})$,
  for each function symbol $f \in \fsignature$ and elements $b_1,
  \ldots, b_\arityof{f} \in \universeOf{B}$.
\end{theorem}
\noindent The estimation on the time needed to build a filtered
grammar follows immediately from its definition\footnote{The upper
bound $\Omega$ is needed to evaluate the side condition $b =
t^\algof{B}(b_1,\ldots,b_n)$ for a term $t$ in a rule $X \rightarrow
t[X_1,\ldots,X_n] \in \rules$, where the size of $t$ is bounded by
$\sizeof{\grammar}$.}.

%

\newcommand{\RefinementTheorem}{Refinement Theorem}

\subsection{Stratified Grammars}
\label{subsec:stratified-grammars}

Stratified grammars have been introduced by Courcelle (under the name
AC grammars, for associative-commutative) as a simple syntactic
restriction that ensures the recognizability of an unranked
context-free set of trees~\cite[Definition 3.7]{CourcelleV}. However,
the initial definition of AC grammars does not capture all
recognizable unranked sets of trees (e.g., trees with an even number
of leaves cannot be defined by AC grammars). In contrast, we prove
later that our definition of stratified grammars completely
characterize the recognizable sets of unranked trees and beyond (i.e.,
series-parallel and graphs of tree-width $2$). Moreover, a syntactic
fragment of stratified grammars defines exactly those sets recognized
by aperiodic algebras (Definition \ref{def:aperiodic-recognizer}).

\begin{definition}\label{def:stratified-grammar}
  A \emph{stratified} $\fsignature$-grammar
  $\grammar=(\nonterm,\rules)$ has nonterminals partitioned as
  $\nonterm=\mathcal{X}\uplus\mathcal{Y}$ and rules of the following
  forms, for some nonterminals $X \in \mathcal{X}$ and $Y,Y_1\ldots
  Y_k \in \mathcal{Y}$ pairwise distinct: \begin{compactenum}[A.]
  \item\label{it1:syntactic-regular} $X \rightarrow X \pop
    \pexp{Y}{q}$, for an integer $q \geq 1$,
  \item\label{it2:syntactic-regular} $X \rightarrow \pexp{Y_1}{q_1}
    \pop \ldots \pop \pexp{Y_k}{q_k}$, for integers
    $q_1,\ldots,q_k \geq 1$,
  \item\label{it31:syntactic-regular} $Y \rightarrow
    t[Z_1,\ldots,Z_m]$, for an $\fsignature$-term $t$,
    $Z_1,\ldots,Z_m \in \nonterm$,
  \item\label{it32:syntactic-regular} $X \rightarrow t$, for a ground
    $\fsignature$-term $t$,
  \item\label{it4:syntactic-regular} $\rightarrow Z$, $Z \in \nonterm$.
  \end{compactenum}
  We say that $\grammar$ is \emph{aperiodic} if and only if, for every
  pair of non-terminals $X\in\mathcal{X}$ and $Y\in\mathcal{Y}$, the
  set of rules $\set{X \rightarrow
    X \pop \pexp{Y}{q_i} \in \rules}_{i\in\interv{1}{k}}$  (\ref{it1:syntactic-regular})
  is either (i) empty, (ii) a singleton and $q_1=1$, or (iii) $k\geq2$ and $q_1, \ldots q_k$ are coprime.
\end{definition}


We define the footprint of stratified grammars in terms of templates
describing the shape of the rules of the form other than
(\ref{it1:syntactic-regular}). The templates will be defined as
grammar rules over a special set of nonterminals consisting of exactly
two nonterminals $\X_\slabs$ and $\Y_\slabs$ of sort $\slabs$, for
every sort $\slabs$ of the considered class.  Intuitively, $\X_\slabs$
resp. $\Y_\slabs$ describe the possible occurrences of the
nonterminals $X \in \mathcal{X}$ (resp. $Y \in \mathcal{Y}$) with
$\sortof{X} = \slabs$ (resp. $\sortof{Y} = \slabs$). We introduce the
following notation that relates the nonterminals $Z \in
\mathcal{X}\uplus\mathcal{Y}$ of a stratified grammar to the
corresponding template nonterminals: $\nontermof{Z} \isdef
\X_{\sortof{Z}}$, for $Z \in \mathcal{X}$, and $\nontermof{Z} \isdef
\Y_{\sortof{Z}}$, for $Z \in \mathcal{Y}$.

\begin{definition}\label{def:footprint}
  The \emph{footprint} of a stratified grammar $\grammar =
  (\mathcal{X}\uplus\mathcal{Y},\rules)$ is the set $\fpof{\grammar}$
  consisting of the templates:
  \begin{compactenum}[1.]
  \item\label{it1:footprint} $\X_\slabs \rightarrow
    \pexp{\Y_\slabs}{m}$, where $m \isdef ~ \min\{\sum_{i=1}^k q_i
    \mid X \rightarrow \pexp{Y_1}{q_1} \pop \ldots \pop
    \pexp{Y_k}{q_k} \in \rules,~ \sortof{X} = \sortof{Y_i} =
    \slabs\}$, where $\min \emptyset = \infty$ if there is no such
    rule,
  \item\label{it2:footprint} $\nontermof{Z} \rightarrow
    t[\nontermof{Z_1}, \ldots, \nontermof{Z_m}]$, for each rule $Z
    \rightarrow t[Z_1,\ldots,Z_m] \in \rules$ of the form
    (\ref{it31:syntactic-regular}) or (\ref{it32:syntactic-regular}),
  \item\label{it3:footprint} $\rightarrow \nontermof{Z}$, for each
    axiom $\rightarrow Z$ of $\grammar$.
  \end{compactenum}
  Given sets of templates $\afp_1$ and $\afp_2$, we write $\afp_1
  \subseteq \afp_2$ if and only if $m_1 \geq m_2$, for all $\X_\slabs
  \rightarrow \pexp{\Y_\slabs}{m_1} \in \afp_1$ and $\X_\slabs
  \rightarrow \pexp{\Y_\slabs}{m_2} \in \afp_2$, and each template of
  types \ref{it2:footprint} or~\ref{it3:footprint} from $\afp_1$
  belongs to $\afp_2$.
\end{definition}

Footprints will be used in the definition of regular grammars.
For each class $\class$, we will fix a footprint $\afp_\class$ and
define the regular grammars for the class as the set of grammars
$\grammar$ having $\fpof{\grammar} \subseteq \afp_\class$.


Certain uses of stratified grammars require a normal form:

\begin{definition}\label{def:normal-form}
  A stratified grammar $\grammar =
  (\mathcal{X}\uplus\mathcal{Y},\rules)$ is \emph{normalized} iff, for
  each pair of nonterminals $(X,Y) \in \mathcal{X} \times
  \mathcal{Y}$, there exists at most one rule $X \rightarrow X \pop
  \pexp{Y}{q} \in \rules$ of form (\ref{it1:syntactic-regular}).
\end{definition}

The following lemma shows that the normalization of a stratified
grammar is effective, at the cost of an exponential blowup in size,
and incurs a restriction of the footprint:

\begin{lemmaE}[][category=proofs]\label{lemma:normal-form}
  Let $\grammar = (\mathcal{X}\uplus\mathcal{Y},\rules)$ be a
  stratified grammar for a class $\class$. Then, one can build a
  normalized stratified grammar $\grammar' =
  (\mathcal{X}\uplus\mathcal{Y},\rules')$ such
  that: \begin{compactenum}[1.]
  \item $\alangof{}{\class}{\grammar'} =
    \alangof{}{\class}{\grammar}$,
  \item $\fpof{\grammar'} \subseteq \fpof{\grammar}$,
  \item $\sizeof{\grammar'} \in
    2^{\poly{\sizeof{\grammar} \cdot
        \log(\sizeof{\grammar})}}$,
  \end{compactenum}
  Moreover, $\grammar$ is aperiodic only if $q=1$, for all rules $X
  \rightarrow X \pop \pexp{Y}{q}$ of the form
  (\ref{it1:syntactic-regular}) from $\grammar'$.
\end{lemmaE}
\begin{proofE}
  We perform the following transformation of $\grammar$, for each pair
  $(X,Y) \in \mathcal{X} \times \mathcal{Y}$.  Let $X \rightarrow X
  \pop \pexp{Y}{q_1}, \ldots, X \rightarrow X \pop \pexp{Y}{q_h}$ be
  the rules of the form (\ref{it1:syntactic-regular}) from $\grammar$,
  for the given pair $(X,Y)$. We assume w.l.o.g. that each rule above
  occurs on at least one complete derivation starting with an axiom
  nonterminal; rules that do not comply with this assumption can be
  identified in $\poly{\sizeof{\grammar}}$ time and removed without
  changing the language of $\grammar$. Moreover, we assume
  w.l.o.g. that each nonterminal from $\mathcal{X}\uplus\mathcal{Y}$
  occurs in at least one rule from $\grammar$ and that $q_1 < \ldots <
  q_h$.  We distinguish two cases: \begin{compactitem}[-]
  \item if $h\geq1$, let $d\geq1$ be the gcd of the integers $q_1,
    \ldots, q_h$. By Schur's theorem, there exists a computable
    integer $n\geq1$ such that for each integer $x > n$ there exist
    $x_1,\ldots,x_h \in \nat$ such that $d x = q_1 x_1 + \ldots + q_h
    x_h$. More precisely, $n$ is the Frobenius number of the coprimes
    $\frac{q_1}{d}, \ldots, \frac{q_h}{d}$. We further observe that
    the set $M \isdef \set{q_1 x_1 + \ldots + q_h x_h \mid
      x_1,\ldots,x_h \in \nat} \cap \interv{1}{dn}$ is finite and
    computable.
  \item else $h=0$, in which case we set $d=0$ and
    $M = \emptyset$.
  \end{compactitem}
  We now obtain the set of rules $\rules'$ from $\rules$ as
  follows:
  \begin{compactenum}[(a)]
  \item\label{it1:singleB} we keep the rules $Z \rightarrow t[Z_1,
    \ldots Z_m]$ of the form (\ref{it31:syntactic-regular}) or
    (\ref{it32:syntactic-regular}),
  \item\label{it2:singleB} we remove the rules $X \rightarrow X \pop
    \pexp{Y}{q_1}$, $\ldots$, $X \rightarrow X \pop \pexp{Y}{q_h}$
    (\ref{it1:syntactic-regular}) and we add the rule $X \rightarrow X
    \pop \pexp{Y}{d}$ if $d>0$ (there are no such rules if $d=0$),
  \item\label{it3:singleB} for every rule $X \rightarrow
    \pexp{Y_1}{q_1} \pop \ldots \pop \pexp{Y_k}{q_k}$
    (\ref{it2:syntactic-regular}) such that $Y = Y_\ell$, for some
    $\ell \in \interv{1}{k}$, we add a rule \(X \rightarrow
    \pexp{Y_1}{q_1} \pop \ldots \pop \pexp{Y_{\ell-1}}{q_{\ell-1}}
    \pop \pexp{Y}{q_\ell+m} \pop \pexp{Y_{\ell+1}}{q_{\ell+1}} \pop
    \ldots \pop \pexp{Y_k}{q_k}\), for each $m \in M$.
  \end{compactenum}
  It is easy to show that each derivation of $\grammar$ can be
  simulated by a derivation of $\grammar'$ with the same outcome, and
  viceversa.  Since we perform iteratively the above transformation
  for every pair $(X,Y) \in \mathcal{X} \times \mathcal{Y}$, we obtain
  the grammar $\grammar'$ such that $\alangof{}{\algof{C}}{\grammar'}
  = \alangof{}{\algof{C}}{\Gamma}$ and $\fpof{\Gamma'} \subseteq
  \fpof{\Gamma}$, having the stated property regarding the rules of
  the form (\ref{it1:syntactic-regular}).  We note that in general we
  only have the inclusion $\fpof{\Gamma'} \subseteq \fpof{\Gamma}$
  because step (\ref{it3:singleB}) of the transformation potentially
  increases the size of the right-hand side of the rules of the form
  (\ref{it2:syntactic-regular}), raising the lower bound in the
  footprint, for the case (\ref{it1:footprint}).

  Moreover, $\grammar$ is aperiodic if and only if $d=1$ for all rules
  introduced at step (\ref{it2:singleB}), by
  Definition~\ref{def:stratified-grammar}. We note that these are the
  only rules of the form (\ref{it1:syntactic-regular}) from
  $\grammar'$, by construction.

  It remains to prove that $\sizeof{\grammar'} \in
  2^{\poly{(\cardof{\mathcal{X}}+\cardof{\mathcal{Y}}) \cdot
      \log(\sizeof{\grammar})}}$. Consider the changes described in
  the previous. Clearly, steps (\ref{it1:singleB}) and
  (\ref{it2:singleB}) do not increase the size of the grammar.  Step
  (\ref{it3:singleB}) adds, for each $m \in M$ and each rule $\rho$ of
  form (\ref{it2:syntactic-regular}), a rule of size at most
  $\sizeof{\rho} + m$.  We have $m \leq n$ and note that $n \in
  \poly{\sizeof{\grammar}}$ follows from the upper bound on the
  Frobenius number $n \leq (\frac{q_1}{d}-1)(\frac{q_h}{d}-1)-1 \leq
  \sizeof{\grammar}^2 \in \poly{\sizeof{\grammar}}$, see, e.g.,
  \cite{Erdos1972}. Moreover, this transformation is applied once for
  each pair $(X,Y) \in \mathcal{X} \times \mathcal{Y}$. 
  Hence, we
  obtain for every rule $\rho \in \rules$ of form
  (\ref{it31:syntactic-regular}) at most
  $\poly{\sizeof{\grammar}}^{\cardof{\mathcal{X}}\cdot \cardof{\mathcal{Y}}}$
  rules of form (\ref{it31:syntactic-regular}) in $\rules'$, each of
  size at most $\sizeof{\rho} +
  \poly{\sizeof{\grammar}}^{\cardof{\mathcal{Y}}}$.
  We note that for every $X$ we apply the transformation at most once
  for every $Y$, i.e., at most
  $\cardof{\mathcal{Y}}$ many times. Thus, we
  obtain $\sizeof{\grammar'} \leq
  \poly{\sizeof{\grammar}}^{\poly{(\cardof{\mathcal{X}}+\cardof{\mathcal{Y}})}}
  \in 2^{\poly{\sizeof{\grammar} \cdot \log(\sizeof{\grammar})}}$, by
  the assumption that each nonterminal from $\mathcal{X} \uplus
  \mathcal{Y}$ occurs in some rule of $\grammar$.
\end{proofE}
It is worth pointing out that the above lemma is true in any algebra
that interprets the parallel composition as an associative and
commutative operation.

The result of applying the Filtering Theorem
(Theorem~\ref{thm:refinement}) directly to a stratified grammar and a
recognizer $(\algof{B},C)$ is not necessarily a stratified grammar,
e.g., because a rule $X \rightarrow X \pop Y$
(\ref{it1:syntactic-regular}) yields zero or more rules $X^{a}
\rightarrow X^{b} \pop Y^{c}$, for each $a,b,c\in\universeOf{B}$ such
that $a \neq b$ and $a = b \pop^{\algof{B}} c$, that are not of the
form (\ref{it1:syntactic-regular}). The next result strengthens the
Filtering Theorem, showing that the intersection of the language of a
stratified grammar with a recognizable set can be represented by an
effectively constructible stratified grammar, whose footprint is
included in the original one:

\begin{theoremE}[\RefinementTheorem][category=proofs]\label{thm:refinement}
  Let $\class$ be a class. For each stratified
  $\fsignature_\class$-grammar $\grammar$ and each
  $\fsignature_\class$-recognizer $(\algof{A},B)$ one can build a
  stratified $\fsignature_\class$-grammar $\grammar'$ such that
  $\fpof{\Gamma'} \subseteq \fpof{\Gamma}$ and
  $\alangof{}{\algof{C}}{\grammar'} =
  \alangof{}{\algof{C}}{\Gamma^{(\algof{A},B)}}$. Moreover,
  $\grammar'$ is aperiodic whenever $\grammar$ is an aperiodic
  stratified grammar and $\algof{A}$ is an aperiodic
  $\fsignature_\class$-algebra.
\end{theoremE}
\begin{proofE}
\newcommand{\recgrammarab}{\grammar^{(\algof{A},B)}}

Let $N \isdef \cardof{\universeOf{A}}$ and $\grammar \isdef
(\mathcal{X}\uplus\mathcal{Y},\rules)$. By
Lemma~\ref{lemma:normal-form}, we can assume that for each pair of
nonterminals $(X,Y) \in \mathcal{X} \times \mathcal{Y}$ there is at
most one rule $X \rightarrow X \pop \pexp{Y}{q}$
(\ref{it1:syntactic-regular}) in $\rules$. We define the grammar
$\Gamma' \isdef (\mathcal{X}'\uplus\mathcal{Y}', \rules')$, where
$\mathcal{X}' \isdef \set{ X_a \mid X \in \mathcal{X},~ a \in
  \universeOf{A}}$ and $\mathcal{Y}' \isdef \set{ Y_a \mid Y \in
  \mathcal{Y},~ a \in \universeOf{A}}$ and $\sort(Z_a) \isdef
\sort(Z)$ for all $Z \in \mathcal{X}\uplus\mathcal{Y}$ and $a \in
\universeOf{A}$. The set of rules $\rules'$ of $\grammar'$ is the
following: \begin{enumerate}[A.]\addtocounter{enumi}{4}
\item\label{it1:syntactic-regular-refinement} $X_b \rightarrow X_b
  ~\pop~ \pexp{Y_a}{qr}$, for every rule $X \rightarrow X ~\pop~
  \pexp{Y}{q} \in \rules$ and each $a, b \in \universeOf{A}$ and $r
  \in \interv{1}{N}$ such that $b = b \pop^\algof{A}
  \pexpalg{A}{a}{qr}$,
\item\label{it2:syntactic-regular-refinement} $X_b \rightarrow \pop_{Y
  \in \mathcal{Y}} \pop_{a \in \universeOf{A}} \pexp{Y_a}{s(Y,a) +
  q(Y,a)}$, for each nonterminal $X \in \mathcal{X}$, element $b \in
  \universeOf{A}$ and integers $s(Y,a), q(Y,a) \in \nat$ such that $\sum_{a \in \universeOf{A}}
  s(Y,a) = s(Y)$ and $0 \leq q(Y,a) \leq N \cdot q(Y)$ and: \begin{itemize}[-]
  \item there exists a rule $X \rightarrow ~\pop_{Y \in \mathcal{Y}}
    \pexp{Y}{s(Y)} \in \rules$, and
  \item for all $Y \in \mathcal{Y}$, either: \begin{itemize}[*]
  \item $\sum_{a \in \universeOf{A}} q(Y,a) = k \cdot q(Y)$ if there
    exists a rule $X \rightarrow X \pop \pexp{Y}{q(Y)} \in \rules$ and
    an integer $k\in \nat$, or
  \item $q(Y,a) = 0$, for each $a \in \universeOf{A}$, if no rule of
    the form $X \rightarrow X \pop \pexp{Y}{q} \in \rules$ exists,
  \end{itemize}
  \item $b = ~\pop_{Y \in \mathcal{Y}} \pop_{a \in \universeOf{A}}
    \pexp{a}{s(Y,a) + q(Y,a)}$.
  \end{itemize}
\item\label{it3:syntactic-regular-refinement} $Z_b \rightarrow
  t[(Z_1)_{a_1},\ldots,(Z_m)_{a_m}]$ for every rule $Z \rightarrow
  t[Z_1,\ldots,Z_m] \in \rules$ of the form
  (\ref{it31:syntactic-regular}) or (\ref{it32:syntactic-regular}) and
  $b, a_1, \ldots a_m \in \universeOf{A}$ such that $b =
  t^\algof{A}(a_1, \ldots, a_m)$,
\item\label{it4:syntactic-regular-refinement} $\rightarrow Z_b$ for
  each axiom $\rightarrow Z$ in $\mathcal{R}$ and $b \in B$.
\end{enumerate}

First, it is an easy check that
$\grammar'=(\mathcal{X}'\uplus\mathcal{Y}',\rules')$ is indeed a
stratified grammar and moreover satisfies $\fpof{\Gamma'} \subseteq
\fpof{\Gamma}$, by construction. In particular, the templates for
rules of the forms (\ref{it31:syntactic-regular}),
(\ref{it32:syntactic-regular}) and (\ref{it4:syntactic-regular}) are
the same for $\grammar$ and $\grammar'$, whereas the lower bound $m$
of template $\X \rightarrow \pexp{\Y}{m}$ for the rules of the
form (\ref{it2:syntactic-regular}) is increased in $\grammar'$
w.r.t. $\grammar$, by the definition of $\rules'$ in the case
(\ref{it2:syntactic-regular-refinement}). More precisely, this is
because, for each rule $X_b \rightarrow \pop_{Y \in \mathcal{Y}}
\pop_{a \in \universeOf{A}} \pexp{Y_a}{s(Y,a) + q(Y,a)} \in \rules'$,
there exists a rule $X \rightarrow ~\pop_{Y \in \mathcal{Y}}
\pexp{Y}{s(Y)} \in \rules$ such that $\sum_{a\in\universeOf{A}} s(Y,a)
+ q(Y,a) \geq s(Y)$.

Moreover, if $\grammar$ is aperiodic, we have $q=1$, for each rule $X
\rightarrow X \pop \pexp{Y}{q}$ of the
form~(\ref{it1:syntactic-regular}), by
Lemma~\ref{lemma:normal-form}. By the fact below, the only rules of
the form~(\ref{it1:syntactic-regular-refinement}) from $\grammar'$ are
such that $qr=1$, hence $\grammar'$ is an aperiodic stratified
grammar.

\begin{fact}
  If $\algof{A}$ is an aperiodic algebra, we have $b = b
  \pop^\algof{A} a$, for all $a,b\in\universeOf{A}$ and $r \ge 1$ such that
  \[
  b = b \pop^\algof{A} \underbrace{a \pop^\algof{A} \ldots \pop^\algof{A} a}_{\isdef ~\pexpalg{A}{a}{r} \text{($r$ times)}}.
  \]
\end{fact}
\begin{proof}
  Assume that $\idemof{a}=\pexpalg{A}{a}{n}$, for some $n\geq1$. By
  composing the left and right-hand sides of the above equality $n$
  times with $\pexpalg{A}{a}{r}$, we obtain $b = b \pop^\algof{A}
  \pexpalg{A}{a}{rn} = b \pop^\algof{A}
  \pexpalg{A}{\left(\idemof{a}\right)}{r}$. By the aperiodicity of
  $\algof{A}$, we have $\idemof{a} \pop^\algof{A} a = \idemof{a}$,
  hence \(b = b \pop^\algof{A} \pexpalg{A}{\left(\idemof{a}\right)}{r}
  = \left(b \pop^\algof{A}
  \pexpalg{A}{\left(\idemof{a}\right)}{r}\right) \pop^\algof{A} a = b
  \pop^\algof{A} \left( \pexpalg{A}{\left(\idemof{a}\right)}{r}
  \pop^\algof{A} a\right) = b \pop^\algof{A} \left(\idemof{a} \pop
  a\right) = b \pop^\algof{A} a.\)
\end{proof}

Second, we argue that $\alangof{}{\algof{C}}{\grammar'} =
\alangof{}{\algof{C}}{\Gamma^{(\algof{A},B)}}$:

\vspace*{.5\baselineskip}
\noindent''$\subseteq$''
We will argue that the same $\fsignature_\class$-terms can be
derived from $\recgrammarab$ and $\grammar'$.
More precisely, we prove the following equality, for each nonterminal $Z_a \in \mathcal{X}' \uplus \mathcal{Y}'$:
\begin{align}
  \set{t \mid Z_a \step{\recgrammarab}^* t} = & \set{t \mid Z_a \step{\grammar'}^* t}
\end{align}
We first observe that $\recgrammarab$ and $\grammar'$ have the same
set of axioms of the form (\ref{it4:syntactic-regular-refinement}).
Further, the rules of the form
(\ref{it3:syntactic-regular-refinement}) of $\grammar'$ are exactly
the rules in $\recgrammarab$ obtained by refinement of the rules of
form either (\ref{it31:syntactic-regular}) or
(\ref{it32:syntactic-regular}) of $\grammar$, and hence will generate
the same set of terms. 
It remains to prove that the rules of the forms
(\ref{it1:syntactic-regular-refinement}),
(\ref{it2:syntactic-regular-refinement}) and
(\ref{it1:syntactic-regular}), (\ref{it2:syntactic-regular}) produce
the same set of terms starting from a nonterminal $X_b$ in the
grammars $\grammar'$ and $\recgrammarab$, respectivelly.

We now consider an arbitrary $\pop$-term $t=\pop_{Y
\in \mathcal{Y}} \pop_{a \in \universeOf{A}} \pexp{Y_a}{s(Y,a) +
  q(Y,a)}$ that can be generated by $\recgrammarab$ starting with some
nonterminal $X_b$ using the refined rules of the forms
(\ref{it1:syntactic-regular}) and (\ref{it2:syntactic-regular}), where
$s(Y,a)$ and $q(Y,a)$ are positive integers satisfying the following
constraints: \begin{enumerate}[1.]
\item\label{it1:refinement-left-right} there exists a rule $X
  \rightarrow \pop_{Y \in \mathcal{Y}} \pexp{Y}{s(Y)}$ in $\grammar$,
\item\label{it2:refinement-left-right} there exists a rule $X
  \rightarrow X \pop \pexp{Y}{q(Y)}$ in $\grammar$,
\item\label{it3:refinement-left-right} $b = \pop^{\algof{A}}_{Y \in
  \mathcal{Y}} \pop^{\algof{A}}_{a \in \universeOf{A}} \pexp{a}{s(Y,a)
  + q(Y,a)}$.
\end{enumerate}
To prove that $X_b \step{\grammar'} t$, we distinguish two
cases. First, assume that there is a rule $X \rightarrow X \pop
\pexp{Y}{q(Y)} \in \rules$; we recall that this rule is unique, and
hence $q(Y)$ is unique. By the pigeonhole principle, for each $a \in
\universeOf{A}$ there exist integers $0 \le p(Y,a) < N$ and $1 \le
r(Y,a) \le N$ such that $p(Y,a) + r(Y,a) \le N$ and
$\pexpalg{A}{a}{p(Y,a) + r(Y,a)} = \pexpalg{A}{a}{p(Y,a)}$. By
composing both sides of this equality $q(Y)$ times, we obtain
$\pexpalg{A}{a}{q(Y) \cdot p(Y,a) + q(Y) \cdot r(Y,a)} =
\pexpalg{A}{a}{q(Y) \cdot p(Y,a)}$. Next, we define integers $0 \le
q'(Y,a) < N \cdot q(Y)$, for each $a \in \universeOf{A}$ and $Y \in
\mathcal{Y}$, as follows:
\begin{align*} q'(Y,a) \isdef
  & ~\left\{\begin{array}{l} p(Y,a) \cdot q(Y) ~+ \\
  (q(Y,a)-p(Y,a) \cdot q(Y) ) \bmod (r(Y,a) \cdot q(Y)), \\
  \hspace*{5mm} \text{if } q(Y,a)\ge N \cdot q(Y) \\
  q(Y,a) \text{, otherwise}
  \end{array}\right.
\end{align*}
For each $a \in \universeOf{A}$ and $Y \in \mathcal{Y}$, we have
$q(Y,a) = q'(Y,a) + k(Y,a) \cdot r(Y,a) \cdot q(Y)$, for some
coefficients $k(Y,a) \ge 0$. Moreover, by point~\ref{it3:refinement-left-right}., we obtain that $b = b \pop
a^{q(Y)\cdot r(Y,a)}$, if $q(Y,a)\ge N \cdot q(Y)$, for each $a
\in \universeOf{A}$ and $Y \in \mathcal{Y}$ ($\dagger$).

Second, if there is a no rule $X \rightarrow X \pop \pexp{Y}{q(Y)} \in
\rules$, we have $q(Y,a) = 0$ for all $a \in \universeOf{A}$, and we
set $q'(Y,a) \isdef 0$ for all $a \in \universeOf{A}$.

We now observe that $X_b \step{\grammar'} \pop_{Y \in \mathcal{Y}}
\pop_{a \in \universeOf{A}} \pexp{Y_a}{s(Y,a) + q'(Y,a)}$, because a
rule $X_b \rightarrow \pop_{Y \in \mathcal{Y}} \pop_{a \in
  \universeOf{A}} \pexp{Y_a}{s(Y,a) + q'(Y,a)}$ of the form
(\ref{it2:syntactic-regular-refinement}) exists in
$\grammar'$. According to the definition of $\grammar'$, this is the
case provided that $b = \pop^{\algof{A}}_{Y \in \mathcal{Y}}
\pop^{\algof{A}}_{a \in \universeOf{A}} \pexp{a}{s(Y,a) +
  q'(Y,a)}$. The latter condition can be obtained from ($\dagger$), by
subtracting $k(Y,a)$ times $r(Y,a)\cdot q(Y)$ from $q(Y,a)$ in the
exponent of $a$ in equation (\ref{it3:refinement-left-right}). Then,
by applying the rules of the form
(\ref{it1:syntactic-regular-refinement}), $k(Y,a)$-times respectively,
we obtain the term $\pop_{Y \in \mathcal{Y}} \pop_{a \in
  \universeOf{A}} \pexp{Y_a}{s(Y,a) + q(Y,a)}$ using that $q(Y,a) =
q'(Y,a) + k(Y,a) \cdot r(Y,a) \cdot q(Y)$, for all $Y\in\mathcal{Y}$
and $a\in\universeOf{A}$.

\noindent''$\supseteq$'':
We note that it is straightforward to prove by induction that
$\graph \in \alangof{Z_a}{\class}{\grammar'}$ only if $\homof{G}{A}(\graph) = a$, for all $a \in \universeOf{A}$.
Hence, $\alangof{}{\class}{\grammar'} \subseteq \homof{\class}{\algof{A}}^{-1}(B)$.
Moreover, it is easy to verify that every complete derivation of $\grammar'$ induces a complete derivation $\grammar$ (in particular, every rule $X_b \rightarrow X_b ~\pop~ \pexp{Y_a}{qr} \in \rules'$ induces an $r$-fold application of $X \rightarrow X ~\pop~ \pexp{Y}{q} \in \rules$, and every rule $X_b \rightarrow \pop_{Y \in \mathcal{Y}} \pop_{a \in \universeOf{A}} \pexp{Y_a}{s(Y,a) + q(Y,a)} \in \rules'$ induces an application of $X \rightarrow ~\pop_{Y \in \mathcal{Y}} \pexp{Y}{s(Y)} \in \rules$ plus a $k_Y$-fold application of $X \rightarrow X \pop \pexp{Y}{q(Y)} \in \rules$ for each $Y$, where $\sum_{a \in \universeOf{A}} q(Y,a) = k_Y \cdot q(Y)$).
Hence, $\alangof{}{\class}{\grammar'} \subseteq \alangof{}{\class}{\grammar}$.
Thus, we get  $\alangof{}{\class}{\grammar'} \subseteq \alangof{}{\class}{\grammar} \cap \homof{\class}{\algof{A}}^{-1}(B) = \alangof{}{\class}{\recgrammarab}$.
\end{proofE}

The refinement theorem will be used as follows:
We first establish the existence of a \emph{universally stratified} grammar, i.e., a grammar $\grammar_\class$ with $\fpof{\grammar_\class} = \afp_\class$ and $\alangof{}{\class}{\grammar_\class}=\universeOf{C}$.
The upcoming definitions of regular grammars via the condition $\fpof{\grammar} \subseteq \afp_\class$ will then ensure that these families of grammars are closed under refinement
(Theorem~\ref{thm:refinement}),
see Corollary~\ref{cor:completeness}:

\begin{definition}\label{def:universal-stratified}
  A class $\class$ is (aperiodic) \emph{universally stratified} if there exists a (aperiodic) stratified grammar
  $\grammar_\class$ such that $\fpof{\grammar_\class} = \afp_\class$ and
  $\alangof{}{\class}{\grammar_\class}=\universeOf{C}$, for the domain
  $\universeOf{C}$ of $\class$.
\end{definition}

\begin{corollaryE}[][category=proofs]\label{cor:completeness}
  Let $\class$ be a (aperiodic) universally stratified class and
  $\langu\subseteq\universeOf{C}$ be a set recognizable in $\class$
  (by an aperiodic recognizer). Then, there exists a (aperiodic)
  stratified grammar $\grammar$ for $\class$ such that $\fpof{\grammar}
  \subseteq \fpof{\grammar_\class}$ and
  $\langu=\alangof{}{\class}{\grammar}$.
\end{corollaryE}
\begin{proofE}
  Let $(\algof{A},B)$ be a $\fsignature_\class$-recognizer for
  $\langu$, i.e., we have $\langu=\homof{C}{A}^{-1}(B)$. Since
  $\class$ is universally stratified, there exists a stratified
  grammar $\grammar_\class$ for $\class$ such that
  $\alangof{}{\class}{\grammar_\class}=\universeOf{C}$. Then, $\langu
  = \universeOf{C} \cap \homof{C}{A}^{-1}(B) =
  \alangof{}{\class}{\grammar^{(\algof{A},B)}_\class}$, by Theorem
  \ref{thm:filtering}. By Theorem~\ref{thm:refinement}, there exists a
  stratified grammar $\grammar'$ such that
  $\fpof{\grammar'}\subseteq\fpof{\grammar_\class}$ and
  $\langu=\alangof{}{\class}{\grammar^{(\algof{A},B)}_\class}=\alangof{}{\class}{\grammar'}$.
  Moreover, if $\grammar_\class$ is an aperiodic stratified grammar
  for $\class$ and $\langu$ is recognized by an aperiodic recognizer,
  then $\grammar'$ is an aperiodic stratified grammar for $\class$, by
  Theorem~\ref{thm:refinement}.
\end{proofE}

In the rest, we introduce several notions needed to prove the
contrapositive of Corollary \ref{cor:completeness}, namely that the
language of a stratified grammar
$\grammar=(\mathcal{X}\uplus\mathcal{Y},\rules)$ in a given class
$\class$ is recognizable in $\class$. This is achieved by building,
from the grammar $\grammar$ a recognizer $(\algof{A}_\grammar,
B_\grammar)$ for its language.  The elements of the recognizer algebra
$\algof{A}_\grammar$ will include multisets of nonterminals from the
set $\mathcal{Y}$ of the given stratified grammar. These multisets are
finite representations of graphs consisting of arbitrarily large
parallel compositions of subgraphs. The purpose of the notions below
is bounding the multiplicity of these multisets, in order to define
the finite algebra $\algof{A}_\grammar$.

Given a multiset $m \in \mpow{\mathcal{Y}}$, we define the reduced
multiset $\trunk{m}_\grammar$, in which each nonterminal $Y$ has
multiplicity $\trunk{m}_\grammar(Y)$ bounded by a constant depending on $\grammar$.
A multiset $m$ is said to be \emph{reduced} iff $m = \trunk{m}_\grammar$. \ifLongVersion\else The
definition of the reduction operation relies on the exponents
occurring in the rules of the form (\ref{it1:syntactic-regular}) and
(\ref{it2:syntactic-regular}) from $\grammar$.
\fi

Formally, let $\grammar=(\mathcal{X}\uplus\mathcal{Y},\rules)$ be a stratified grammar. For each nonterminal $Y \in \mathcal{Y}$, we
define the following integers, depending on $\grammar$:
\begin{align*}
  b(Y) \isdef & ~\max\set{q_i \mid X \rightarrow \pexp{Y_1}{q_1} \pop \ldots \pop \pexp{Y_n}{q_n} \in \rules,~ Y_i = Y} \\
  p(Y) \isdef & ~\max\set{1,\lcm\set{q \mid X \rightarrow X \pop \pexp{Y}{q} \in \rules}} \\
  q(Y) \isdef & ~b(Y) + p(Y)
\end{align*}
\begin{align*}
  \trunk{m}_\grammar(Y) \isdef \begin{cases}
    m(Y), \text{ if } m(Y) < q(Y) \\
    q(Y) + (m(Y) - q(Y)) \mod p(Y), \text{ otherwise }
  \end{cases}
\end{align*}


\begin{lemma}\label{fact:trunk-two}
  Given a stratified grammar
  $\grammar=(\mathcal{X}\uplus\mathcal{Y},\rules)$, for all multi-sets
  $m_1, m_2 \in \mpow{\mathcal{Y}}$, we have $\trunk{m_1 \mcup
    m_2}_\grammar = \trunk{\trunk{m_1}_\grammar \mcup
    \trunk{m_2}_\grammar}_\grammar$.
\end{lemma}
\begin{proof}\emph{Proof.}
  Immediate from the definition of the operation $\trunk{\cdot}_\grammar$.
\end{proof}

Given a class $\class$, a \emph{$\pop$-term} is a
$\fsignature_\class$-term consisting of occurrences of some function
symbol $\pop_{\slabs,\slabs}$ and variables of sort
$\slabs\in\sorts_\class$ only. For a stratified
$\fsignature_\class$-grammar
$\grammar=(\mathcal{X}\uplus\mathcal{Y},\rules)$ and a nonterminal
$X\in\mathcal{X}$, we denote by $X \leadsto_\grammar m$ the existence
of a derivation $X \step{\grammar}^* t$, where $t$ is a $\pop$-term
$t$ and $m$ is the multiset of variables from $t$. Intuitively, $X
\leadsto_\grammar m$ means that $X$ can produce the multiset $m$ of
$\mathcal{Y}$ nonterminals using only the rules of the form
(\ref{it1:syntactic-regular}) and (\ref{it2:syntactic-regular}) from
$\grammar$. This matching relation will be used to define the
interpretation of the signature $\fsignature_\class$ in the recognizer
algebra $\algof{A}_\grammar$.

\begin{proofsTextEnd}
\begin{lemma}\label{fact:trunk}
  Let $\grammar=(\mathcal{X}\uplus\mathcal{Y},\rules)$ be a normalized
  stratified grammar.  For each $X \in \mathcal{X}$ and multiset $m
  \in \mpow{\mathcal{Y}}$, we have $X \leadsto_\grammar m \iff X
  \leadsto_\grammar \trunk{m}_\grammar$.
\end{lemma}
\begin{proof}\emph{Proof.}
  ``$\Rightarrow$'' Let $Y$ be a nonterminal.  If
  $m(Y)=\trunk{m}_\grammar(Y)$ there is nothing to prove, hence we
  assume that $m(Y) > \trunk{m}_\grammar(Y)$, i.e.,
  $\trunk{m}_\grammar(Y) = q(Y) + (m(Y) - q(Y)) \mod p(Y) = b(Y) +
  p(Y)+ (m(Y) - q(Y)) \mod p(Y) < m(Y)$.  Suppose, for a
  contradiction, that $\grammar$ has no rule of the form $X
  \rightarrow \pexp{Y}{q}$, for the chosen nonterminals $X$ and $Y$.
  In this case, $m(Y) = \trunk{m}_\grammar(Y)$, contradiction.  Hence,
  there exists a rule $X \rightarrow X \pop \pexp{Y}{q}$ ($\dagger$)
  and an integer $n \in \nat$ such that $m(Y) = n\cdot p(Y) +
  \trunk{m}_\grammar(Y)$. Because $\grammar$ is normalized, the
  rule ($\dagger$) is unique. Let $X \step{\grammar}^* t$ be the
  derivation that witnesses $X \leadsto_\grammar m$.  This derivation
  consists of one or more applications of the rule ($\dagger$) and one
  or more rules of the form $X \rightarrow \pexp{Y_1}{q_1} \pop \ldots
  \pop \pexp{Y_n}{q_n}$. By removing $nk$ occurrences of the rule
  ($\dagger$), we obtain a derivation $X \step{\grammar}^* u$, where
  $\trunk{m}_\grammar(Y)$ is the number of occurrences of $Y$ in $u$
  and $k\geq1$ is an integer such that $p(Y)=kq$. This reduction is
  repeated for each $Y \in \mathcal{Y}$, until we obtain a derivation
  that witnesses $X \leadsto_\grammar \trunk{m}_\grammar$.

  ``$\Leftarrow$'' Let $Y$ be a nonterminal. If
  $m(Y)=\trunk{m}_\grammar(Y)$ there is nothing to prove, hence we
  assume that $m(Y) > \trunk{m}_\grammar(Y)$, i.e.,
  $\trunk{m}_\grammar(Y) = q(Y) + (m(Y) - q(Y)) \mod p(Y)$.  Suppose,
  for a contradiction, that no rule $X \rightarrow X \pop \pexp{Y}{q}$
  exists in $\grammar$. In this case, $m(Y) = \trunk{m}_\grammar(Y)$,
  contradiction.  Then, let $X \rightarrow X \pop \pexp{Y}{q}$ be such
  a rule and $k \in \nat$ be such that $p(Y)=kq$; by the definition of
  $p(Y)$, such a $k\in\nat$ exists. We extend the derivation $X
  \step{\grammar}^* u$ witnessing $X \leadsto_\grammar
  \trunk{m}_\grammar$ by adding $nk$ steps that apply the ($\dagger$)
  rule to it, where $n\in\nat$ is such that $m(Y)=n \cdot q(Y) +
  \trunk{m}_\grammar(Y)$. By repeating this step for each $Y \in
  \supp{m}$, we obtain a derivation that witnesses $X
  \leadsto_\grammar m$.
\end{proof}
\end{proofsTextEnd}

\begin{example}
  For a stratified grammar $\grammar$ having $\mathcal{X}=\set{X}$,
  $\mathcal{Y}=\set{Y_1,Y_2}$ and the following rules: \begin{align*}
    X \rightarrow X \pop \pexp{Y_1}{3}
    \hspace*{4mm} X \rightarrow X \pop \pexp{Y_2}{5}
    \hspace*{4mm} X \rightarrow \pexp{Y_1}{2} \pop \pexp{Y_2}{3}
  \end{align*}
  $X \leadsto_\grammar m$ iff $m(Y_1) = 2 \mod 3$ and $m(Y_2)
  = 3 \mod 5$.
\end{example}

\section{Regular Grammars}
\label{sec:regular}

We define formally the classes of trees (Definition~\ref{def:trees}),
series-parallel graphs (Definition~\ref{def:sp}) and graphs of
tree-width at most~$2$ (Definition~\ref{def:tw2}). The regular
grammars for each class are stratified grammars whose footprint is
subsumed by a footprint depending on the class. Each of the considered
classes is shown to have a universal aperiodic regular grammar,
ensuring that every recognizable subset of the class is the language
of some regular grammar (Corollary \ref{cor:completeness}). The bulk
of this section is devoted to establishing the reverse direction,
i.e., that the language of each regular grammar is recognizable. This
is achieved by an effective construction of a recognizer from the
given regular grammar. This construction yields a uniform
doubly-exponential upper bound on the size of the recognizer.

In the rest of this section, we prove the following result:

\begin{theorem}\label{thm:main}
  Let $\class$ be one of the classes of trees, (disoriented)
  series-parallel graphs and graphs of tree-width at most~$2$. A set
  $\langu\subseteq\universeOf{C}$ is recognizable (by an aperiodic
  recognizer) if and only if there exists a (aperiodic) regular
  grammar $\grammar$ such that $\langu=\alangof{}{C}{\grammar}$.
\end{theorem}

For the class of trees (Definition~\ref{def:trees}), two terms that
represent the same graph differ only by a reordering of their
$\pop$-subterms, modulo associativity and commutativity. This leads to
an almost canonical decomposition of trees, that is instrumental in
proving the two directions of Theorem~\ref{thm:main}, i.e., the
existence of a universal aperiodic stratified grammar (``only if'')
and of a recognizer for the language of a regular tree grammar
(``if''). In particular, the elements of the recognizer are sets of
multisets of nonterminals used to parse the $\pop$-subterms of trees,
where the outer sets are necessary to deal with the nondeterminism in
the grammar. Hence, the doubly-exponential upper bound on the size of
the recognizer algebra.

For graphs of tree-width at most~$2$, the proof of Theorem~\ref{thm:main} uses a standard decomposition of connected graphs into
blocks (i.e., maximally 2-connected components)~\cite[Theorem
  III.23]{TutteBook} in conjunction with the fact that each block of
tree-width at most~$2$ is a disoriented series-parallel
graph~\cite[Lemma 6.15]{CourcelleV}. For technical purposes, we prove
a stronger version of this result, that allows to pick any vertex for
the $1$-source of a block. The decomposition of a graph of tree-width
$2$ as a tree of disoriented series-parallel blocks leads to the
definition of an algebra (i.e., a class) and of a family of regular
grammars, having sorts $\set{1,2}$ and $\set{1}$ for blocks and block
trees, respectively. Accordingly, the recognizer algebra is
$2$-sorted, multisets of nonterminals sort $\set{1}$ represent block
trees, whereas multisets and pairs of nonterminals of sort $\set{1,2}$
represent disoriented series-parallel blocks.

The upper bound on the size of the recognizer algebra for a regular
grammar generalizes from trees to graphs of tree-width at most
$2$. Despite the best of our efforts, we could not find a matching
lower bound, namely a regular grammar $\grammar$ whose language could
not be recognized by an algebra with less than
$2^{2^{\sizeof{\grammar}}}$ elements.

An application of the effective construction of recognizer algebras is
bounding the time complexity for the problem of inclusion between an
arbitrary context-free language and a regular language:

\begin{theoremE}[][category=proofs]\label{thm:inclusion}
  Let $\class$ be one of the classes of trees, (disoriented)
  series-parallel graphs and graphs of tree-width at most~$2$. The
  problem \emph{``given $\fsignature_\class$-grammars $\grammar$ and
    $\grammar'$ such that $\grammar'$ is regular, does
    $\alangof{}{C}{\grammar} \subseteq \alangof{}{C}{\grammar'}$
    ?''}  belongs to \twoexptime\ and is \exptime-hard, even if
  $\grammar'$ is aperiodic.
\end{theoremE}
\begin{proofE}
  Let $(\algof{A}, \universeOf{B})$ be a recognizer for
  $\alangof{}{C}{\grammar'}$. The existence of such a recognizer
  follows from the fact that $\grammar'$ is regular for $\class$,
  where $\class$ is the class of either trees (Theorem
  \ref{thm:tree-reg}), series-parallel graphs (Theorem
  \ref{thm:sp-reg}), or graphs of tree-width $2$ at most (Theorem
  \ref{thm:tw2-reg}).  Moreover, we have $\cardof{\universeOf{A}} \in
  2^{2^\poly{\sizeof{\grammar'}}}$ and that each function
  $f^\algof{A}$ can be computed in time
  $2^\poly{\sizeof{\grammar'}}$, for each of the three classes
  considered. Then, $(\algof{A},\universeOf{A}\setminus
  \universeOf{B})$ is a recognizer for the complement language
  $\overline{ \alangof{}{C}{\grammar'}} \isdef \universeOf{C}
  \setminus \alangof{}{C}{\grammar'}$.  By Theorem
  \ref{thm:filtering}, we obtain
  $\alangof{}{C}{\grammar^{(\algof{A},\universeOf{A}\setminus
      \universeOf{B})}} = \alangof{}{C}{\grammar} \cap \overline{
    \alangof{}{C}{\grammar'}} = \alangof{}{C}{\grammar} \setminus
  \alangof{}{C}{\grammar'}$.  Hence
  $\alangof{}{C}{\grammar^{(\algof{A},\universeOf{A}\setminus
      \universeOf{B})}} = \emptyset$ if and only if
  $\alangof{}{C}{\grammar} \subseteq \alangof{}{C}{\grammar'}$.
  Since $\grammar^{(\algof{A}, \universeOf{A}\setminus
    \universeOf{B})}$ can be computed in time:
  \begin{align*}
  \bigO\left(\sizeof{\grammar} \cdot
  \left(2^{2^\poly{\sizeof{\grammar'}}}\right)^{\sizeof{\grammar}} \cdot
  \left(2^\poly{\sizeof{\grammar'}}\right)^{\sizeof{\grammar}}\right) \\
  \in 2^{2^{\poly{\sizeof{\grammar}+\sizeof{\grammar'}}}}
  \end{align*}
  and the emptiness of a grammar can be checked in linear time, we
  obtain that the inclusion problem $\alangof{}{C}{\grammar}
  \subseteq \alangof{}{C}{\grammar'}$ belongs to \twoexptime.

  For the \exptime-hardness lower bound, we deal separately with the case of trees and series-parallel graphs.
  Since the class of graphs of tree-width $2$ at most subsumes both trees and series-parallel graphs, either result implies \exptime-hardness of the inclusion problem for the class of graphs of tree-width $2$ at most.

  We next prove the \exptime-hardness of the inclusion problem for trees (over a binary alphabet).
  We reduce from the inclusion problem between languages of
  nondeterministic bottom-up tree automata
  $A_i=(Q_i,I_i,\delta_i,F_i)$ over an arbitrary ranked alphabet
  $\alphabet$, i.e., whose symbols $f \in \alphabet$ have associated
  ranks $\arityof{f}\geq0$, where $Q_i$ is a finite set of states,
  $I_i,F_i \subseteq Q_i$ are initial and final states and $\delta_i$
  is a set of transitions of the form $f(q_1,\ldots,q_{\arityof{f}})
  \arrow{}{} q$, for $i=1,2$.
  The language $\langof{}{A_i}$ is the set of ranked trees $t$ over $\alphabet$ that can be labeled with states by the automaton $A_i$ such that the following hold:
  \begin{compactitem}[-]
  \item each subtree $f(t_1,\ldots,t_{\arityof{f}})$ of $t$ is
    recursively labeled by $q$ if $t_j$ is labeled by $q_j$, for $j
    \in \interv{1}{\arityof{f}}$ and $f(q_1,\ldots,q_{\arityof{f}})
    \arrow{a}{} q \in \delta_j$,
  \item the root of $t$ is labeled by a final state.
  \end{compactitem}
  The problem ``given nondeterministic bottom-up tree automata $A_1$ and $A_2$, does $\langof{}{A_1} \subseteq \langof{}{A_2}$?''  is  known to be \exptime-complete \cite[Corollary 1.7.9]{comon:hal-03367725}.

  We now state how tree automata can be encoded by regular tree
  grammars.  For this we first fix some ranked alphabet $\alphabet$.
  Let $m = \max \{ \arityof{a} \mid a \in \alphabet \}$.  We define
  the alphabet $\alphabetTwo_\alphabet = \{e_f \mid f \in \alphabet \}
  \cup \{ p_i \mid 1 \le i \le m \}$ for some fresh symbols $e_f$ and
  $p_i$, with $\arityof{e_f} = 2$ and $\arityof{p_i} = 2$ for all $f
  \in \alphabet$ and $1 \le i \le m$.  Let $\class$ be the class of
  trees over the alphabet $\alphabetTwo$.  We define the encoding
  $\encof{\tree}$ of a tree over the alphabet $\alphabet$ by a tree
  over the alphabet $\alphabetTwo_\alphabet$; the encoding is defined
  via the unique term $t$ such that $t^\algof{T} = \tree$:
  \begin{align*}
  \encof{\tree} \isdef \left\{\begin{array}{l}
    ~ \emptygraph_{\set{1}}^\algof{T} \text{, if } t=  \emptygraph_{\set{1}} \\
    ~ \extend{e_f}^\algof{T}(\pop^\algof{T}_{i\in\interv{1}{_\arityof{f}}} \extend{p_i}^\algof{T}(\encof{t_i})),  \\
    \hspace*{5mm} \text{if } t = f(t_1,\ldots,t_\arityof{f}),~ \arityof{f}=n
    \end{array}\right.
  \end{align*}

  We now show how to encode some tree automaton  $A=(Q,I,\delta,F)$ over the alphabet $\alphabet$ by a regular tree grammar $\grammar_A =(\nonterm,\rules)$ over the alphabet $\alphabetTwo_\alphabet$.
  We choose $\nonterm$ as the union of the non-terminals $Y_{i,q}$, for all $1 \le i \le m$ and $q \in \states$, the non-terminals $X_{i,q,f}$, for all $1 \le i \le m$, $q \in \states$ and $f \in \alphabet$, and the non-terminal $X$.
  For every transition  $f(q_1,\ldots,q_{\arityof{f}})
  \arrow{}{} q$ of $A$ and $1 \le i \le m$, we add to $\rules$ the rules
  \begin{align*}
  Y_{q} \rightarrow & ~ \extend{e_f}(X_{f,q}) &  X_{f,q} \rightarrow & ~ Y_{1,q_1} \pop \cdots \pop Y_{\arityof{f},q_{\arityof{f}}} \\
  Y_{i,q} \rightarrow & ~ \extend{p_i}(X_{q}) &
  X_{q} \rightarrow & ~ Y_q
  \end{align*}
  The axioms of $\grammar$ consist of the rules $\rightarrow Y_q$ such that $q \in \initstates$.
  It is easy to check that
  $\alangof{}{T}{\grammar_A} = \encof{\langof{}{A}}$ and that $\grammar_A$ is aperiodic (actually, it has no rules of the form (\ref{it1:syntactic-regular})).
  Then, given
  nondeterministic bottom-up automata $A_1$ and $A_2$, we have
  $\langof{}{A_1} \subseteq \langof{}{A_2}$ if and only if $\encof{\langof{}{A_1}} \subseteq \encof{\langof{}{A_2}}$ if and only if $\alangof{}{T}{\grammar_{A_1}} \subseteq
  \alangof{}{T}{\grammar_{A_2}}$.
  Since the time needed to compute $\grammar_A$ is polynomial in the size of $A$, for any bottom-up
  tree automaton $A$, the inclusion problem for the class of trees over an alphabet containing only binary symbols, where the
  right-hand side is an aperiodic regular grammar, is \exptime-hard.

  Now, let $\class$ be the class of SP graphs over an alphabet
  $\alphabet$ of binary symbols.  We reduce from the inclusion problem
  for trees, where the right-hand side is an aperiodic regular
  grammar.  To this end, we encode a tree $\tree$ over an alphabet
  $\alphabet$ of binary symbols of binary edge labels as an SP graph
  $\spof{\tree}$ over the edge alphabet $\alphabetTwo = \alphabet \cup
  \set{e}$, for some fresh edge label $e$, where $e$ is used to link
  the leaves of the tree to the $2$-source of the SP graph; the
  $1$-source is the root of the tree. It is clear that this encoding
  is injective, hence $\alangof{}{T}{\grammar} \subseteq
  \alangof{}{T}{\grammar'}$ if and only if
  $\spof{\alangof{}{T}{\grammar}} \subseteq
  \spof{\alangof{}{T}{\grammar'}}$, for any two
  $\fsignature_\algof{T}$-grammars $\grammar$ and $\grammar'$.

  First, given an arbitrary tree grammar $\grammar$ over $\alphabet$,
  we define the $\fsignature_\algof{SP}$-grammar $\overline{\grammar}$
  having the same nonterminals as $\grammar$ and the rules:
  \begin{compactitem}[-]
  \item $Z \rightarrow \sgraph{a}_{1,2} \sop Z'$, for each rule
    $Z \rightarrow \extend{\alpha}(Z')$ of $\grammar$, where $a
    \in \alphabet$,
  \item $Z \rightarrow \sgraph{e}_{1,2}$, for each rule $Z \rightarrow
    \emptygraph_{\set{1}}$ of $\grammar$.
  \item $\rightarrow Z$, for each axiom $\rightarrow Z$ of $\grammar$.
  \end{compactitem}
  It is easy to check that $\spof{\alangof{}{T}{\grammar}} =
  \alangof{}{SP}{\overline{\grammar}}$: each term generated by
  $\grammar$ corresponds to exactly one term generated by
  $\overline{\grammar}$, starting from the same nonterminal, in which each subterm $\extend{\alpha}(x)$ is replaced by
  $\sgraph{a}_{1,2} \sop x$ and each subterm
  $\emptygraph_{\set{1}}$ by $\sgraph{e}_{1,2}$.

  Second, given a regular tree grammar
  $\grammar=(\mathcal{X}\uplus\mathcal{Y},\rules)$ over the alphabet $\alphabet$, we define the regular series-parallel
  grammar $\widetilde{\grammar}=(\set{P_X \mid X \in \mathcal{X} \cup \alphabetTwo} \uplus \set{S_Y \mid Y \in \mathcal{Y}},
  \widetilde{\rules})$ with rules:
  \begin{align*}
    S_Y \rightarrow & ~P_a \sop P_X, P_X \rightarrow ~S_{Y_1} \pop \ldots \pop S_{Y_n}, \\
    \text{ if } & Y \rightarrow \extend{a}(X),~ X \rightarrow Y_1 \pop \ldots \pop Y_n \in \rules \text{ for } n \geq 2 \\
    S_{Y_1} \rightarrow & ~P_a \sop S_{Y_2},
    \text{ if } Y_1 \rightarrow \extend{\alpha}(X),~ X \rightarrow Y_2 \in \rules \\
    S_Y \rightarrow & ~P_a \sop P_e,
    \text{ if } Y \rightarrow \extend{a}(X),~ X \rightarrow \emptygraph_{\set{1}} \in \rules \\
    P_b \rightarrow  & ~\sgraph{b}_{1,2},
    \text{ if } b \in \alphabetTwo \\
    \rightarrow & ~P_X,
    \text{ if } \rightarrow X, X \rightarrow Y_1 \pop \ldots \pop Y_n \in \rules, \text{, for } n \geq 2 \\
    \rightarrow & ~S_Y,
    \text{ if } \rightarrow X, X \rightarrow Y \in \rules \\
    \rightarrow & ~P_e,
    \text{ if } \rightarrow X, X \rightarrow \emptyset_{\set{1}} \in \rules
  \end{align*}
  It is easy to see that $\widetilde{\grammar}$ is an aperiodic
  regular series-parallel grammar if $\grammar$ is an aperiodic tree grammar and, moreover, $\spof{\alangof{}{T}{\grammar}} =
  \alangof{}{SP}{\overline{\grammar}}$.

  Then, $\spof{\alangof{}{T}{\grammar}} \subseteq
  \spof{\alangof{}{T}{\grammar'}}$ if and only if
  $\alangof{}{SP}{\overline{\grammar}_1} \subseteq
  \alangof{}{SP}{\widetilde{\grammar}_2}$, for any two
  $\spsignature$-grammars $\grammar$ and $\grammar'$, such that
  $\grammar'$ is regular.  Since the time needed to compute both
  $\overline{\grammar}_1$ and $\widetilde{\grammar}_2$ is polynomial
  in the sum of the sizes of $\grammar$ and $\grammar'$, the
  inclusion problem for the series-parallel class, where the
  right-hand side is an aperiodic regular grammar, is \exptime-hard.
\end{proofE}

We believe that the above upper bound can be improved for aperiodic
regular grammars and conjecture that the size of the minimal
recognizer is bounded by a simple exponential in the size of the
grammar. If this conjecture is proved, a consequence is that the
inclusion problem between a context-free language and an aperiodic
regular language is \exptime-complete, at least for the classes of
graphs considered in this paper, thus matching the complexity of the
inclusion between ranked tree languages~\cite[Corollary
  1.7.9]{comon:hal-03367725}.

Another application of effectively constructible recognizers is the
computation of the union, intersection and complement of recognizable
sets, within their respectively classes. \ifLongVersion
\else As a side remark, the generalized
inclusion problem, of a context-free language into a boolean
combination of recognizable languages has the same \twoexptime\ upper
bound as the problem from Theorem~\ref{thm:inclusion}.  \fi

\subsection{Trees}
\label{subsec:trees-reg}

Let $\alphabet$ be a finite alphabet of edge labels, of arities
$\arityof{a} \geq 1$ for each $a \in \alphabet$.  We recall that the
class $\algof{T}$ of unordered trees has signature
$\fsignature_\algof{T} = \set{\extend{b} \mid b \in \alphabet} \cup
\set{\pop,\emptygraph_{\set{1}}}$ and sort $\set{1}$, where the
$1$-source of a tree is its root (Definition~\ref{def:trees}). For
simplicity, the footprint $\afp_\algof{T}$ of the class of trees is
written using $\X$ and $\Y$ instead of $\X_{\set{1}}$ and
$\Y_{\set{1}}$, respectively:

\begin{definition}\label{def:tree-reg}
  A \emph{regular tree grammar} is a stratified
  $\fsignature_\algof{T}$-grammar
  $\grammar=(\mathcal{X}\uplus\mathcal{Y},\rules)$ such that
  $\fpof{\grammar} \subseteq \afp_\algof{T}$, where:
  \begin{align*}
    \afp_\algof{T}
    \isdef & ~\set{\X \rightarrow \pexp{\Y}{\geq 0}} \cup \set{\Y
      \rightarrow \extend{b}(\X,\ldots,\X) \mid b \in \alphabet} \\
    & \cup \set{\rightarrow \X}
  \end{align*}
\end{definition}

To establish the ``only if'' direction of Theorem~\ref{thm:main} for
the class $\algof{T}$ it is sufficient to prove the existence of an
aperiodic universal regular tree grammar (Corollary
\ref{cor:completeness}). Lemma \ref{lemma:universal-synt-reg-tree}
below shows that such a grammar exists. To improve clarity, we
annotate each rule with its type, according to
Definition~\ref{def:stratified-grammar}:

\begin{lemmaE}[][category=proofs]\label{lemma:universal-synt-reg-tree}
  Let $\grammar_\algof{T} \isdef(\set{X,Y},\rules)$ be the regular
  tree grammar having the following rules:
  \begin{align*}
    Y \rightarrow & ~\extend{a}(X, \ldots, X) \text{, for all } a \in \alphabet \text{ (\ref{it31:syntactic-regular})} \\
    X \rightarrow & ~X \pop Y \text{ (\ref{it1:syntactic-regular})} \hspace*{5mm}
    X \rightarrow \emptygraph_{\set{1}} \text{ (\ref{it2:syntactic-regular})} \hspace*{5mm}
    \rightarrow X \text{ (\ref{it4:syntactic-regular})}
  \end{align*}
  Then, we have $\alangof{}{\algof{T}}{\grammar_\algof{T}} = \universeOf{T}$.
\end{lemmaE}
\begin{proofE}
  ``$\subseteq$'' Obvious, as $\grammar_\algof{T}$ uses only
  operations from $\fsignature_\algof{T}$. ``$\supseteq$''.  Let $T
  \in \universeOf{T}$ be a tree.  Let $t$ be an
  $\fsignature_\algof{T}$-term with $t^\algof{T} = T$.  By a
  straight-forward induction on the structure of $t$ we can establish
  that $T \in \alangof{}{\algof{T}}{\grammar_\algof{T}}$.
\end{proofE}

Note that the rule $X \rightarrow \emptygraph_{\set{1}}$ is of the
form (\ref{it2:syntactic-regular}) because the right-hand side is an
empty parallel composition. It is easy to see that
$\fpof{\grammar_\algof{T}}$ is a subset of the footprint
$\afp_\algof{T}$ of the class of trees (Definition~\ref{def:tree-reg})
and the only rule of the form (\ref{it1:syntactic-regular}) has
exponent $1$, hence $\grammar_\algof{T}$ is an aperiodic regular tree
grammar.

\begin{figure}[t!]
  \centerline{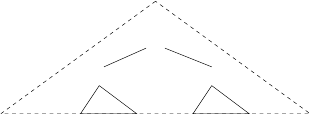}
  \caption{A view $\mset{Y_1,\ldots,Y_n}$ of a tree
    $(v_1\pop\ldots\pop v_n)^\algof{T}$.}
  \label{fig:tree-view}
  \vspace*{-\baselineskip}
\end{figure}

To establish the ``if'' direction of Theorem~\ref{thm:refinement} for
the class $\algof{T}$, let us fix a normalized regular tree grammar
$\grammar=(\mathcal{X}\uplus\mathcal{Y},\rules)$. Assuming that $\grammar$ is normalized
loses no generality (Lemma~\ref{lemma:normal-form}). Intuitively, a
view of a tree $\tree$ is a multiset of $\mathcal{Y}$ nonterminals
that gives the complete information on how $\tree$ is parsed by a
subsequence of some derivation of $\grammar$. Reduced views provide a
finitary abstraction of this information, that loses no precision. A
profile (i.e., the set of all possible reduced views) describes all
the possible ways in which $\tree$ can be parsed by the
grammar. Figure \ref{fig:tree-view} depicts this intuition, formalized
below:

\newcommand{\recgrammar}{\grammar_\universeOf{T}^{(\algof{A},B)}}

\begin{definition}\label{def:tree-view}
  Let $\tree$ be a tree and
  $\grammar=(\mathcal{X}\uplus\mathcal{Y},\rules)$ be a normalized
  regular tree grammar. A \emph{view of $\tree$ for $\grammar$} is a
  multiset $\mset{Y_1,\ldots,Y_n} \in \mpow{\mathcal{Y}}$ such that
  there exist complete derivations $Y_i \step{\grammar}^* v_i$, for
  all $i \in \interv{1}{n}$ and $\tree = (v_1 \pop \ldots \pop
  v_n)^\algof{T}$.  A \emph{reduced view} of $\tree$ is a multiset
  $\trunk{m}_\grammar$, where $m$ is a view $m$ of $\tree$, for
  $\grammar$. The \emph{profile of $\tree$ for $\grammar$} is the set
  $h_\grammar(\tree) \isdef \{\trunk{m}_\grammar \mid m \text{ is a view of
  } \tree \text{ for } \grammar\}$.
\end{definition}


We shall build a $\fsignature_\algof{T}$-recognizer
$(\algof{A}_\grammar,B_\grammar)$ such that
$\alangof{}{T}{\grammar}=h_\grammar^{-1}(B_\grammar)$, where
$h_\grammar$ is the function that maps a tree to its profile. The
$\fsignature_\algof{T}$-algebra $\algof{A}_\grammar$ has domain
$\universeOf{A}_\grammar \isdef \set{h_\grammar(\tree) \mid \tree \in
  \universeOf{T}}$. The interpretations of the function symbols from
$\fsignature_\algof{T}$ are the least sets built according to the
inference rules from Figure \ref{fig:tree-alg}.

Note that the algebra $\algof{A}_\grammar$ is finite, because there
are finitely many reduced views, i.e., a reduced view counts each
nonterminal $Y \in \mathcal{Y}$ up to a bound depending on the
grammar.  Moreover, the elements $a \in \universeOf{A}$ are sets of
reduced views, hence there are finitely many such sets. We define
$\sizeof{a}\isdef\sum_{m\in a}\cardof{m}$, for each element $a \in
\universeOf{A}$. The following lemma gives an upper bound on the size
of the elements of $\algof{A}_\grammar$, the cardinality of
$\universeOf{A}_\grammar$ and the time needed to compute the
interpretation of a function symbol:

\begin{lemmaE}[][category=proofs]\label{lemma:size-of-tree-rec}
  Let $\grammar = (\mathcal{X}\uplus\mathcal{Y}, \rules)$ be a
  normalized regular tree grammar.  Then,
  $\cardof{\universeOf{A}_\grammar} \in
  2^{{\sizeof{\grammar}}^{\poly{\cardof{\mathcal{Y}}}}}$, $\sizeof{a}
  \in {\sizeof{\grammar}}^{\poly{\cardof{\mathcal{Y}}}}$, for each $a
  \in \universeOf{A}_\grammar$ and
  $f^{\algof{A}_\grammar}(a_1,\ldots,a_{\arityof{f}})$ is computable
  in time ${\sizeof{\grammar}}^{\poly{\cardof{\mathcal{Y}}}}$, for
  each function symbol $f \in \fsignature_\algof{T}$ and all
  $a_1,\ldots,a_{\arityof{f}} \in \universeOf{A}_\grammar$.
\end{lemmaE}
\begin{proofE}
  We observe that:
  \begin{align*}
    b(Y) = & ~\max\set{q_i \mid X \rightarrow \pexp{Y_1}{q_1} \pop \ldots \pop \pexp{Y_n}{q_n} \in \rules,~ Y_i = Y} \\
    \le & ~\sizeof{\grammar}, \\
    p(Y) = & ~\max\set{1,\lcm\set{q \mid X \rightarrow X \pop \pexp{Y}{q} \in \rules}} \\
    \le & ~\sizeof{\grammar}^{\cardof{\mathcal{Y}}},\\
    q(Y) = & ~b(Y) + p(Y) \le \sizeof{\grammar}^{\cardof{\mathcal{Y}}} + \sizeof{\grammar} \\
    \in & ~{\sizeof{\grammar}}^{\poly{\cardof{\mathcal{Y}}}}
  \end{align*}
  We then set $\mathcal{Q} \isdef \max\set{q(Y) \mid Y \in
    \mathcal{Y}} \in
  {\sizeof{\grammar}}^{\poly{\cardof{\mathcal{Y}}}}$.  There are at
  most $\mathcal{Q}^{\cardof{\mathcal{Y}}} \in
  {\sizeof{\grammar}}^{\poly{\cardof{\mathcal{Y}}}}$ different reduced
  multisets and hence at most
  $2^{{\sizeof{\grammar}}^{\poly{\cardof{\mathcal{Y}}}}}$ different
  sets of reduced multisets, thus $\cardof{\universeOf{A}} \in
  2^{{\sizeof{\grammar}}^{\poly{\cardof{\mathcal{Y}}}}}$.

  We note that an element $a \in \universeOf{A}$ is a set of reduced
  multisets and there are at most $\mathcal{Q}^{\cardof{\mathcal{Y}}}
  \in {\sizeof{\grammar}}^{\poly{\cardof{\mathcal{Y}}}}$ different
  reduced multisets $m$.  We then observe that $\cardof{m} \le
  \cardof{\mathcal{Y}} \cdot \mathcal{Q} \in {\sizeof{\grammar}}^{\poly{\cardof{\mathcal{Y}}}}$.
  Hence, we obtain that $\sizeof{a} \in {\sizeof{\grammar}}^{\poly{\cardof{\mathcal{Y}}}}$.

  We now compute an upper bound on the time needed to evaluate the
  function $f^\algof{A}(a_1,\ldots,a_n)$, for each function symbol $f
  \in \fsignature_\algof{T}$ of arity $\arityof{f}=n \geq 0$ and
  elements $a_1,\ldots,a_n \in
  \universeOf{A}_\grammar$: \begin{itemize}[-]
  \item $\emptygraph_{\set{1}}^\algof{A} = \set{ \emptymset }$
    requires constant time.
  \item $a_1 \pop^\algof{A} a_2$ requires computing $\trunk{m_1 \mcup
    m_2}$ for all $m_1 \in a_1$ and $m_2 \in a_2$. We first note that
    computing $\trunk{m_1 \mcup m_2}$ requires a constant number of
    arithmetic operations for every $Y \in \mathcal{Y}$. Hence,
    $\trunk{m_1 \mcup m_2}$ can be computed in
    $\poly{\cardof{\mathcal{Y}} \cdot \mathcal{Q}} \in
    {\sizeof{\grammar}}^{\poly{\cardof{\mathcal{Y}}}}$.  As we have
    to do this computation for at most
    ${\sizeof{\grammar}}^{\poly{\cardof{\mathcal{Y}}}}$ multisets
    $m_1 \in a_1$ and $m_2 \in a_2$, we get that $a_1 \pop^\algof{A}
    a_2$ can be computed in
    ${\sizeof{\grammar}}^{\poly{\cardof{\mathcal{Y}}}}$.
  \item $\extend{b}^\algof{A}(a_1,\ldots,a_{\arityof{b}-1})$ requires
    an iteration over all rules $Y \rightarrow
    \extend{b}(X_1,\ldots,X_{\arityof{b}-1})$ of which there are at
    most $\sizeof{\grammar}$ many and over all $m_i \in a_i$ of which
    there are at most
    ${\sizeof{\grammar}}^{\poly{\cardof{\mathcal{Y}}}}$ many.  We then
    need to check whether $X_i \leadsto_\grammar m_i$, for all $i \in
    \interv{1}{\arityof{b}-1}$. This check can be done by going over
    all rules of the form (\ref{it2:syntactic-regular}) with some
    left-hand side $X_i$ and then performing a modulo computation
    w.r.t. the rules of the form (\ref{it1:syntactic-regular}) with
    left-hand size $X_i$ (we recall that in $\grammar$ there is at
    most one such rule for every $Y \in \mathcal{Y}$. Most precisely,
    for each rule $X_j \rightarrow \pexp{Y_1}{q_1} \pop \ldots
    \pexp{Y_k}{q_k}$ and each $Y_\ell$, $\ell\in \interv{1}{k}$, one
    must check that $r(j,\ell)$ divides $m_i(Y_\ell) - q_\ell$, where
    $X_j \rightarrow X_j \pop \pexp{Y_\ell}{r(j,\ell)}$ is the unique
    rule of the form (\ref{it1:syntactic-regular}) for $X_j$ and
    $Y_\ell$. Since all numbers are encoded in unary, this check takes
    time at most $(m_i(Y_\ell)+q_\ell) \cdot r(j,\ell) \leq
    (\mathcal{Q}+\sizeof{\grammar}) \cdot \sizeof{\grammar} \in
    {\sizeof{\grammar}}^{\poly{\cardof{\mathcal{Y}}}}$.
    Then, $\extend{b}^\algof{A}(a_1,\ldots,a_{\arityof{b}-1})$ can be
    computed in ${\sizeof{\grammar}}^{\poly{\cardof{\mathcal{Y}}}}$.
  \end{itemize}
\end{proofE}

The following two lemmas establish that the language of the given
regular tree grammar $\grammar$ is recognized by
$(\algof{A}_\grammar,B_\grammar)$, where $B_\grammar \isdef
\set{h_\grammar(\tree) \mid \tree \in \alangof{}{T}{\grammar}}$, via
the homomorphism $h_\grammar$ between $\algof{T}$ and
$\algof{A}_\grammar$. A consequence of Lemma
\ref{lemma:universal-synt-reg-tree} is that the $\algof{T}$ algebra is
representable, meaning that $\algof{A}_\grammar$ is also representable
and $h_\grammar$ is the only such homomorphism.

\begin{figure}[t!]
{\small\begin{center}
  \begin{minipage}{2cm}
    \begin{prooftree}
      \AxiomC{$\begin{array}{c} \\\\ \end{array}$}
      \RightLabel{$\emptygraph$}
      \UnaryInfC{$\emptymset \in \emptygraph_{\set{1}}^{\algof{A}_\grammar}$}
    \end{prooftree}
  \end{minipage}
    \begin{minipage}{5cm}
    \begin{prooftree}
      \AxiomC{$\begin{array}{c} \\ \set{m_i \in a_i}_{i=1,2} \end{array}$}
      \RightLabel{$\pop$}
      \UnaryInfC{$\trunk{m_1+m_2}_\grammar \in a_1 \pop^{\algof{A}_\grammar} a_2$}
    \end{prooftree}
    \end{minipage}

    \vspace*{.5\baselineskip}
    \begin{minipage}{5cm}
      \begin{prooftree}
        \AxiomC{$\begin{array}{c}
            \set{X_i \leadsto_\grammar m_i \in a_i}_{i=1,\ldots,n} \\
            Y \rightarrow \extend{b}(X_1,\ldots,X_{n}) \in \rules
          \end{array}$}
        \RightLabel{$\extend{}$}
        \UnaryInfC{$\mset{Y} \in \extend{b}^{\algof{A}_\grammar} (a_1,\ldots,a_{n})$}
      \end{prooftree}
    \end{minipage}
\end{center}}
\caption{The interpretation of the signature $\fsignature_\algof{T}$ in $\algof{A}_\grammar$}
\vspace*{-\baselineskip}
\label{fig:tree-alg}
\end{figure}

\begin{lemmaE}[][category=proofs]\label{lemma:tree-view-membership}
  For all $\tree_1, \tree_2 \in \universeOf{T}$, if
  $h_\grammar(\tree_1) = h_\grammar(\tree_2)$ then $\tree_1 \in
  \alangof{}{\algof{T}}{\grammar} \iff \tree_2 \in
  \alangof{}{\algof{T}}{\grammar}$.
\end{lemmaE}
\begin{proofE}
  We assume $\tree_1 \in \alangof{}{\algof{T}}{\grammar}$. Then, there
  is a $X$-derivation of $\tree_1$, for some axiom $\rightarrow X$ of
  $\grammar$. Without loss of generality, each such derivation can be
  reordered such that it starts with zero or more rules of type
  (\ref{it1:syntactic-regular}), followed by a single rule of type
  (\ref{it2:syntactic-regular}). Let $m_1$ be the multiset of
  nonterminals from $\mathcal{Y}$ resulting from this prefix of the
  derivation, namely:
  \begin{align*}
    X \step{\grammar}^* \pop_{Y \in \mathcal{Y}} \pexp{Y}{m_1(Y)}
  \end{align*}
  By Definition~\ref{def:tree-view}, $m_1$ is a view of $\tree_1$,
  hence $\trunk{m_1} \in h_\grammar(\tree_1)$ and $\trunk{m_1} \in
  h_\grammar(\tree_2)$ follows from the equality $h_\grammar(\tree_1)
  = h_\grammar(\tree_2)$. Then, there exists a view $m_2$ of $\tree_2$
  such that $\trunk{m_1} = \trunk{m_2}$. Since $X \leadsto_\grammar
  m_1$, we obtain $X \leadsto_\grammar \trunk{m_1}_\grammar =
  \trunk{m_2}_\grammar$ and $X \leadsto_\grammar m_2$, by Lemma
  \ref{fact:trunk}. Because $m_2$ is a view of $\tree_2$ and $X
  \leadsto_\grammar m_2$, by Defintion \ref{def:tree-view}, there
  exists a complete $X$-derivation:
  \begin{align*}
    X \step{\grammar}^* \pop_{Y \in \mathcal{Y}} \pexp{Y}{m_2(Y)} \step{\grammar}^* \pop_{Y \in \mathcal{Y},~ i \in \interv{1}{m_2(Y)}} v(Y,i)
  \end{align*}
  where $v(Y,i)$ is a ground term such that $Y_i \step{\grammar}^*
  v(Y,i)$, for each $Y \in \mathcal{Y}$ and $i \in \interv{1}{m_2(Y)}$
  and $\tree_2 = (\pop_{Y \in \mathcal{Y},~ i \in \interv{1}{m_2(Y)}}
  v(Y,i))^\algof{T}$. Because $\rightarrow X$ is an axiom of
  $\grammar$, we obtain $\tree_2 \in
  \alangof{}{\algof{T}}{\grammar}$. The other direction is symmetric.
\end{proofE}

\begin{lemmaE}[][category=proofs]\label{lemma:tree-view-homomorphism}
  $h_\grammar$ is a homomorphism between $\algof{T}$ and
  $\algof{A}_\grammar$.
\end{lemmaE}
\begin{proofE}
  We prove the following points:

  \vspace*{.5\baselineskip}
  \noindent \underline{$h_\grammar(\emptygraph_{\set{1}}^\algof{T}) =
    \emptygraph_{\set{1}}^{\algof{A}_\grammar}$}: the tree having one
  vertex and no edges is an empty $\pop$-composition of trees derived
  starting with nonterminals from $\mathcal{Y}$. Hence,
  $h_\grammar(\emptygraph_{\set{1}}^\algof{T}) = \set{\emptymset} =
  \emptygraph_{\set{1}}^{\algof{A}_\grammar}$, by the definition of
  $\emptygraph_{\set{1}}^{\algof{A}_\grammar}$.

  \vspace*{.5\baselineskip}
  \noindent \underline{$h_\grammar(\tree_1 \pop^\algof{T} \tree_2) =
    h_\grammar(\tree_1) \pop^{\algof{A}_\grammar}
    h_\grammar(\tree_2)$}: ``$\subseteq$'' Let $m \in
  h_\grammar(\tree_1 \pop^\algof{T} \tree_2)$ be a reduced view.
  Then, by Definition~\ref{def:tree-view}, there is a view $m'$ with
  $m' = \trunk{m}_\grammar$ such that for each $Y \in \mathcal{Y}$ and
  $i \in \interv{1}{m(Y)}$, there exists a complete derivation $Y
  \step{\grammar}^* v(Y,i)$, where $v(Y,i)$ is a ground term, such
  that $\tree_1 \pop^\algof{T} \tree_2 = \pop^\algof{T}_{Y \in
    \mathcal{Y},~ i \in \interv{1}{m'(Y)}} v(Y,i)^\algof{T}$.  By
  splitting the set of terms $\set{v(Y,i) \mid Y \in \mathcal{Y}, i
    \in \interv{1}{m'(Y)}}$ according to whether $v(Y,i)^\algof{T}$ is
  a subtree of $\tree_1$ or $\tree_2$, we obtain multisets $m_1', m_2'
  \in \mpow{\mathcal{Y}}$ such that $m' = m_1' + m_2'$ and $\tree_j =
  (\pop_{Y \in \mathcal{Y}, i \in \interv{1}{m_j'(Y)}}
  v(Y,i))^\algof{T}$, for $j = 1,2$. Hence, $m'_i$ is a view of
  $\tree_i$, for $i=1,2$. Let $m_i \isdef \trunk{m_i'}_\grammar$ be
  reduced views of $\tree_i$, for $i=1,2$. By Lemma
  \ref{fact:trunk-two}, we have $m = \trunk{m'}_\grammar =
  \trunk{m'_1+m'_2}_\grammar = \trunk{\trunk{m'_1}_\grammar +
    \trunk{m'_2}_\grammar}_\grammar = \trunk{m_1+m_2}_\grammar$.
  Then, $m = m_1 + m_2 \in h_\grammar(\tree_1)
  \pop^{\algof{A}_\grammar} h_\grammar(\tree_2)$, by the definition of
  $\pop^{\algof{A}_\grammar}$.

  \vspace*{.5\baselineskip}
  \noindent ``$\supseteq$'' Let $m \in h_\grammar(\tree_1)
  \pop^{\algof{A}_\grammar} h_\grammar(\tree_2)$ be a reduced view.
  By the definition of $\pop^{\algof{A}_\grammar}$, there exist
  reduced views $m_i \in h_\grammar(\tree_i)$, for $i=1,2$, such that
  $m = \trunk{m_1 + m_2}_\grammar$.  Then, there exist views $m'_i$ of
  $\tree_i$, such that $m_i = \trunk{m'_i}_\grammar$, for
  $i=1,2$. Hence, for each $Y \in \supp{m'_j}$ and $i \in
  \interv{1}{m_j(Y)}$, there exists a complete derivation $Y
  \step{\grammar}^* v(Y,i,j)$ such that $\tree_j =
  (\pop_{Y\in\mathcal{Y},i\in\interv{1}{m'_j(Y)}}
  v(Y,i,j))^\algof{T}$, for $j = 1,2$.  We now define the ground terms
  $v(Y,i)$ as $v(Y,i,1)$ for $1 \le i \le m'_1(Y)$, and as $v(Y,i -
  m'_1(Y),2)$ for $m'_1(Y) + 1 \le i \le m'_1(Y) + m'_2(Y)$.  With
  these definitions, we have complete derivations $Y \step{\grammar}^*
  v(Y,i)$ for all $1 \le i \le m'_1(Y) + m'_2(Y)$ and $Y \in
  \supp{m'_1+m'_2}$.  Thus, $m'_1+m'_2$ is a view of
  $\tree_1\pop^\algof{T}\tree_2$. Since $\trunk{m'_1+m'_2}_\grammar =
  \trunk{\trunk{m'_1}_\grammar + \trunk{m'_2}_\grammar}_\grammar =
  \trunk{m_1+m_2}_\grammar =m$, by Lemma \ref{fact:trunk-two}, we
  obtain $m \in h_\grammar(\tree_1\pop^\algof{T}\tree_2)$.

  \vspace*{.5\baselineskip}
  \noindent
  \underline{$h_\grammar(\extend{b}^\algof{T}(\tree_1,\ldots,\tree_{\arityof{b}-1})) =$}
  \underline{$\extend{b}^{\algof{A}_\grammar}(h_\grammar(\tree_1),\ldots,h_\grammar(\tree_{\arityof{b}-1}))$}:
  ``$\subseteq$'' Let $m \in h_\grammar(\tree)$ be a reduced view of
  $\tree$, where $\tree \isdef
  \extend{b}^\algof{T}(\tree_1,\ldots,\tree_{\arityof{b}-1})$. Then,
  there exists a view $m'$ of $\tree$ such that $m =
  \trunk{m'}_\grammar$. Since the only factor in a $\pop$-composition
  that yields $\tree$ is $\tree$ itself, there exists exactly one
  nonterminal $Y \in \mathcal{Y}$ such that $m'=\mset{Y}$. Without
  loss of generality, each derivation that yields a ground term $t$
  such that $\tree=t^\algof{T}$ can be reorganized such that it starts
  with a rule $Y \rightarrow \extend{b}(X_1,\ldots,X_{\arityof{b}-1})$
  followed by complete $X_i$-derivations that yield ground terms $t_i$
  such that $\tree_i = t_i^\algof{T}$, for each $i \in
  \interv{1}{\arityof{b}-1}$. Moreover, we consider w.l.o.g. that each
  such $X_i$-derivation starts with the following prefix:
  \begin{align*}
    X_i \step{\grammar}^* \pop_{Y\in\mathcal{Y}} \pexp{Y}{m_i(Y)}
  \end{align*}
  for some view $m_i$ of $\tree_i$, for each $i \in
  \interv{1}{\arityof{b}-1}$. Then $\trunk{m_i}_\grammar \in
  h_\grammar(\tree_i)$, for each $i \in
  \interv{1}{\arityof{b}-1}$. Since $X_i \leadsto_\grammar m_i$, we
  obtain $X_i \leadsto_\grammar \trunk{m_i}_\grammar$, by Lemma
  \ref{fact:trunk}, hence $m = \trunk{m'}_\grammar = \mset{Y} \in
  \extend{b}^{\algof{A}_\grammar}(h_\grammar(\tree_1),\ldots,h_\grammar(\tree_{\arityof{b}-1}))$,
  by the definition of
  $\extend{b}^{\algof{A}_\grammar}$.

  \vspace*{.5\baselineskip}
  \noindent``$\supseteq$'' Let $m \in
  \extend{b}^{\algof{A}_\grammar}(h_\grammar(\tree_1),$ $\ldots,
  h_\grammar(\tree_{\arityof{b}-1}))$ be a multiset. By the definition
  of $\extend{b}^{\algof{A}_\grammar}$, there exist reduced views $m_i
  \in h_\grammar(\tree_i)$ and a rule $Y \rightarrow
  \extend{b}(X_1,\ldots,X_{\arityof{b}-1})$ in $\grammar$ such that
  $X_i \leadsto_\grammar m_i$, for all $i \in
  \interv{1}{\arityof{b}-1}$, and $m = \mset{Y}$. Then, for each $i
  \in \interv{1}{\arityof{b}-1}$, there exists a view $m'_i$ of
  $\tree_i$ for $\grammar$ such that $m_i=\trunk{m'_i}_\grammar$.  By
  Lemma \ref{fact:trunk}, we obtain $X_i \leadsto_\grammar m'_i$,
  hence there exists complete derivations:
  \begin{align*}
    X_i \step{\grammar}^* \pop_{Y \in \mathcal{Y}} \pexp{Y}{m'_i(Y)} \step{\grammar}^* \pop_{Y \in \mathcal{Y},~ j \in \interv{1}{m'_i(Y)}} v(Y,i,j)
  \end{align*}
  where $v(Y,i,j)$ are ground terms such that $Y \step{\grammar}^*
  v(Y,i,j)$ and $\tree_i = (\pop_{Y \in \mathcal{Y},~ j \in
    \interv{1}{\arityof{b}-1}}) v(Y,i,j))^\algof{T}$, for each $i \in
  \interv{1}{\arityof{b}-1}$.
  We combine these $X_i$-derivations with the rule $Y \rightarrow \extend{b}(X_1,\ldots,X_{\arityof{b}-1})$ to obtain
  a complete $Y$-derivation of a ground term that evaluates to $\tree
  \isdef \extend{b}^\algof{T}(\tree_1,\ldots,\tree_{\arityof{b}-1})$,
  hence $\mset{Y}$ is a view of $\tree$ and $m = \trunk{\mset{Y}} =
  \mset{Y} \in h_\grammar(\tree)$.
\end{proofE}

The next lemma is needed to deal with aperiodic regular tree
grammars. Its proof relies on the fact that the powerset of a
commutative aperiodic semigroup, where the semigroup operation is
lifted to sets, is aperiodic~\cite[Lemma
  5.11]{journals/corr/abs-2008-11635}. Since, for an aperiodic
stratified grammar, the reduced multisets of nonterminals with
(reduced) union forms an aperiodic commutative semigroup, we obtain
the result:

\begin{lemmaE}[][category=proofs]\label{lemma:aperiodic-tree-rec}
  The algebra $\algof{A}_\grammar$ is aperiodic if the regular tree grammar
  $\grammar$ is aperiodic.
\end{lemmaE}
\begin{proofE}
  The regular grammar $\grammar$ is aperiodic iff the modified grammar
  $\grammar'$ obtained from Lemma~\ref{lemma:normal-form} has $q=1$ for
  each rule $X \rightarrow X \pop \pexp{Y}{q}$ of the form
  (\ref{it1:syntactic-regular}). If $\grammar$ is aperiodic, we must have
  $p(Y)=1$, for each nonterminal $Y \in \mathcal{Y}$, because $p$ is
  computed from $\grammar'$. Then, for each multiset $m \in
  \mpow{\mathcal{Y}}$, we have:
  \[\trunk{m}(Y) = \left\{\begin{array}{ll}
  m(Y) \text{, if } m(Y) < q(Y) \\
  q(Y) \text{, otherwise}
  \end{array}\right.\]
  By \cite[Lemma 5.11]{journals/corr/abs-2008-11635}, it is sufficient
  to prove that the commutative semigroup $(\universeOf{M},\oplus)$ is
  aperiodic, where $\universeOf{M} \subseteq \mpow{\mathcal{Y}}$ is the set of all reduced multisets and $m_1 \oplus m_2 \isdef
  \trunk{m_1+m_2}$.
  Now consider an element $m \in \universeOf{M}$ and
  its idempotent power $\idemof{m}$.
  By the definition of $\trunk{.}$
  and Lemma~\ref{fact:trunk-two}, we obtain that $\idemof{m}(Y)$ is
  either $0$ or $q(Y)$, for each nonterminal $Y \in \mathcal{Y}$.
  Hence, $\idemof{m} \oplus m = \idemof{m}$. Since the choice of
  $m\in\universeOf{M}$ was arbitrary, we obtain that
  $(\universeOf{M},\oplus)$ is a aperiodic.
\end{proofE}

The main result of this subsection is the ``if'' direction of Theorem~\ref{thm:main} for the class of trees\footnote{Note that we do not require the grammar to be normalized.}:

\begin{theoremE}[][category=proofs]\label{thm:tree-reg}
  The language of each (aperiodic) regular tree grammar $\grammar$ is
  recognizable by an (aperiodic) recognizer $(\algof{A},B)$, where
  $\cardof{\universeOf{A}} \in 2^{2^{\poly{\size{\grammar}}}}$ and
  $f^{\algof{A}}(a_1,\ldots,a_{\arityof{f}})$ can be computed in time
  $2^{\poly{\size{\grammar}}}$, for each function symbol $f \in
  \fsignature_\algof{T}$ and all elements $a_1,\ldots,a_\arityof{f}
  \in \universeOf{A}$.
\end{theoremE}
\begin{proofE}
  Let $\grammar$ be a regular (aperiodic) tree grammar.  Using
  Lemma~\ref{lemma:normal-form}, we obtain a normalized (aperiodic) grammar
  $\grammar'=(\mathcal{X}\uplus\mathcal{Y},\rules')$.
  In particular, we have $\alangof{}{\algof{T}}{\grammar'} =
  \alangof{}{\algof{T}}{\grammar}$ and $\sizeof{\grammar'} \in
  2^{\poly{\cardof{\mathcal{X}}+\cardof{\mathcal{Y}}} \cdot
    \log(\sizeof{\grammar})}$.  Note that $\grammar$ and $\grammar'$
  have the same set of nonterminals. Then, we build the recognizer
  $\algof{A}_{\grammar'}$ based on the grammar $\grammar'$. By
  Lemma~\ref{lemma:tree-view-membership}, we obtain
  $\alangof{}{\algof{T}}{\grammar'} = h_\grammar^{-1}(B_{\grammar'})$,
  where $B_{\grammar'} \isdef \set{h_{\grammar'}(\tree) \mid \tree \in
    \alangof{}{T}{\grammar'}}$.  By
  Lemma~\ref{lemma:tree-view-homomorphism},
  $(\algof{A}_{\grammar'},B_{\grammar'})$ is a recognizer for
  $\alangof{}{T}{\grammar'}$.
  Because of
  $\alangof{}{\algof{T}}{\grammar'} =
  \alangof{}{\algof{T}}{\grammar}$, $\algof{A}_{\grammar'}$ is also a
  recognizer for the language of $\grammar$.
  Moreover, $\algof{A}_{\grammar'}$ is
  aperiodic if ${\grammar'}$ is aperiodic, by Lemma
  \ref{lemma:aperiodic-tree-rec}.
  By Lemma~\ref{lemma:size-of-tree-rec}, we further have that:
  \begin{align*}
    \cardof{\universeOf{A}_{\grammar'}} \in & ~2^{{\sizeof{\grammar'}}^{\poly{\cardof{\mathcal{Y}}}}} \\
    \in & ~2^{2^{\poly{\cardof{\mathcal{X}}+\cardof{\mathcal{Y}}} \cdot \log(\sizeof{\grammar}) \cdot \poly{\cardof{\mathcal{Y}}}}} \\
    \in & ~2^{2^{\poly{\size{\grammar}}}}
  \end{align*}
  Moreover,
  $f^{\algof{A}_{\grammar'}}(a_1,\ldots,a_n)$ can be computed in time:
  \begin{align*}
  {\sizeof{\grammar'}}^{\poly{\cardof{\mathcal{Y}}}}
  \in &
  ~2^{\poly{\cardof{\mathcal{X}}+\cardof{\mathcal{Y}}} \cdot \log(\sizeof{\grammar})\cdot \poly{\cardof{\mathcal{Y}}}} \\
  \in & ~2^{\poly{\size{\grammar}}}
  \end{align*}
  for each $f \in \fsignature_\algof{T}$ and $a_1,\ldots,a_n \in \universeOf{A}_{\grammar'}$.
 \end{proofE}

\subsection{Graphs of Tree-Width $2$}
\label{subsec:tw2-reg}

In principle, one can define the class of graphs of tree-width $2$ by
the $\btwclass{2}$ subalgebra of $\algof{G}$, defined by the
restriction of sorts to subsets of $\set{1,2,3}$. However, such a
definition does not suit the development of (stratified) regular
grammars that match the recognizable subsets of the class. Instead, we
introduce a derived algebra of $\algof{G}$ that relies on the previous
algebras of trees (Definition~\ref{def:trees}) and disoriented
series-parallel (Definition~\ref{def:sp}). This is made possible by a
graph-theoretic decomposition property, which states that every graph
of tree-width at most~$2$ is, roughly speaking, a tree of disoriented
series parallel graphs.

For reasons of presentation, we consider the following simplifying
restrictions: \begin{compactenum}[1.]
\item We do not consider self-loops. This restriction loses no
  generality, because loops can be encoded as unary edges.
\item We consider an alphabet $\alphabetTwo$ of binary edge
  labels. 
  The extension of the results in this paper to edges of arity more than $3$ will be considered for an extended version.
\item We only consider connected graphs. This is in line with trees
  and series-parallel graphs, which are also connected. The
  generalization of each of these classes to disconnected graphs will
  be considered for an extended version.
\end{compactenum}
The rest of this section is organized as follows. The definitions of
(disoriented) series parallel graphs are recalled in subsection (1).
The class of tree-width $2$ graphs and the proof of Theorem~\ref{thm:main} for this class are presented in subsection (2).

\subsubsection{Series-Parallel Graphs}
\label{subsubsec:sp}

We give the formal definition of (disoriented) series-parallel graphs
and recall a well-known decomposition result. \ifLongVersion\else For
reasons of presentation, we defer the proof of Theorem~\ref{thm:main}
for series-parallel graphs to Appendix \ref{app:sp}. \fi

\begin{definition}\label{def:sp}
  The class $\algof{SP}$ of \emph{series-parallel} graphs has sort
  $\set{1,2}$ and signature: \begin{align*} \spsignature \isdef
    \set{\sgraph{b}_{1,2} \mid b \in \alphabetTwo} \cup
    \set{\sop,\pop} \end{align*} where $x \sop^\algof{SP} y$ joins the
  $2$-source of $x$ with the $1$-source of $y$ and erases the source
  label of the joined vertex.

  The class $\algof{DSP}$ of \emph{disoriented} series-parallel graphs
  has sort $\set{1,2}$ and signature: \begin{align*} \dspsignature
    \isdef \set{\sgraph{b}_{i,3-i} \mid b \in \alphabetTwo,~i=1,2} \cup
    \set{\sop,\pop}
  \end{align*}
  The interpretation of $\pop$ and $\sgraph{b}_{1,2}$
  (resp. $\sgraph{b}_{2,1}$) is the same in $\algof{SP}$
  (resp. $\algof{DSP}$) as in the graph algebra $\algof{G}$
  (Definition~\ref{def:hr}).
\end{definition}

Figure \ref{fig:sp} (a) shows the interpretation of the function
symbols from $\spsignature$. We remark that $\algof{SP}$ is a derived
algebra of $\algof{G}$\footnote{$x \pop^\algof{SP} y \isdef x
  \pop^\algof{G} y$ and $x \sop^\algof{SP} y \isdef
  \rename{(2,3)}^\algof{G}(\restrict{\set{1,3}}^\algof{G}(x
  \pop^\algof{G} \rename{[2,3,1]}^\algof{G}(y)))$.}. Figure
\ref{fig:sp} (b) shows a disoriented series-parallel graph.

\begin{figure}[t!]
  \vspace*{-\baselineskip}
  \centerline{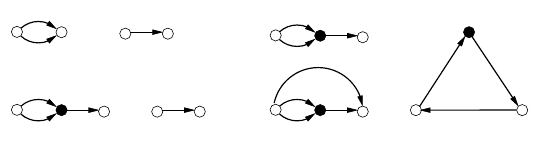}
  \vspace*{-.5\baselineskip}
  \caption{Series-parallel operations (a). A disoriented
    series-parallel graph (b).}
  \label{fig:sp}
  \vspace*{-\baselineskip}
\end{figure}

A graph $\graph \in \universeOf{SP}$ is said to be
\emph{$\sop$-atomic} if there are no graphs $\graph_1,\graph_2 \in
\universeOf{SP}$, such that $\graph = \graph_1 \sop^\algof{SP}
\graph_2$. The set $\universeOf{SP}$ is partitioned into
\emph{P-graphs}, that are $\sop$-atomic, and \emph{S-graphs}, that are
not $\sop$-atomic. For example, Figure \ref{fig:sp}~(a) shows an
S-graph on the top-right and a P-graph on the bottom-right.

\begin{lemmaE}[Lemma 6.3 in \cite{CourcelleV}][]\label{lemma:sp-decomp}
  For each $\graph \in \universeOf{SP}$, we have: \begin{compactenum}[1.]
  \item\label{it1:lemma:sp-decomp} If $\graph$ is an S-graph then
    there exists a unique sequence $\graph_1, \ldots, \graph_k$ of
    P-graphs such that $\graph = \graph_1
    \sop^\algof{SP} \ldots \sop^\algof{SP} \graph_k$, for some
    $k\geq2$.
  \item\label{it2:lemma:sp-decomp} If $\graph$ is a P-graph then
    either $\graph$ is a single edge, or there exists a unique
    nonempty set $\set{\graph_1, \ldots, \graph_k}$ of single edges
    and S-graphs such that $\graph = \graph_1 \pop^\algof{SP} \ldots
    \pop^\algof{SP} \graph_k$, for some $k\geq2$.
  \end{compactenum}
\end{lemmaE}

To apply this canonical decomposition to a disoriented series-parallel
graph of sort $\set{1,2}$, one notices that each disoriented
series-parallel graph can be transformed into an oriented one by a
unique reversal of zero or more edges. For instance, the triangle in
Figure \ref{fig:sp}~(b) has a unique decomposition as
$(\sgraph{c}_{1,2} \sop \sgraph{a}_{1,2}) \pop \sgraph{b}_{2,1}$. We
say that a disoriented series-parallel graph is a P-graph
(resp. S-graph) iff its corresponding series-parallel orientation is a
P-graph (resp. S-graph).

\begin{spTextEnd}
We define regular series-parallel grammars in terms of disjoint sets
of nonterminals $\mathcal{P}$ and $\mathcal{S}$, with the intuition
that P-graphs (resp. S-graphs) are parsed starting with a nonterminal
from $\mathcal{P}$ (resp. $\mathcal{S}$). In the rest of this section,
we write $P$, $S$, $\mathcal{P}$ and $\mathcal{S}$ instead of $X$,
$Y$, $\mathcal{X}$ and $\mathcal{Y}$, respectively. Since $\set{1,2}$
is the only sort from $\spsignature$ (resp. $\dspsignature$), we write
$\P$ and $\S$ instead of $\X_{\set{1,2}}$ and $\Y_{\set{1,2}}$,
respectively, in the footprint of a regular grammar:

\begin{definition}\label{def:regular-sp-grammar}
  A \emph{regular series-parallel grammar} is a stratified
  $\spsignature$-grammar
  $\grammar=(\mathcal{P}\uplus\mathcal{S},\rules)$ such that $\fpof{\grammar} \subseteq \afp_\algof{SP}$, where:
  \begin{align*}
    \afp_\algof{SP} \isdef & ~\set{\P \rightarrow \pexp{\S}{\geq 2}, \S \rightarrow \P \sop \S, \S \rightarrow \P \sop \P, \rightarrow \P, \rightarrow \S} ~\cup \\
    & ~\set{\P \rightarrow \sgraph{b}_{1,2}, \S \rightarrow \sgraph{b}_{1,2} \mid b \in \alphabetTwo}
  \end{align*}
\end{definition}

The definition of regular grammars for the $\algof{DSP}$ class has, in
addition, the footprint constraints $\P \rightarrow \sgraph{b}_{2,1}$
and $\S \rightarrow \sgraph{b}_{2,1}$, for all $b \in
\alphabetTwo$. Note that the use of the series composition is
restricted to rules of the form $S_1 \rightarrow P \sop S_2$ and $S
\rightarrow P_1 \sop P_2$. This ensures that the P-components of an
S-graph are parsed left-to-right, thus avoiding the nondeterminism
introduced by the associativity of the series composition. By way of
analogy, a right-recursive regular word grammar reads words left to
right, just as a finite automaton. Fixing a parsing order for
S-components is an ingredient that allows us to build, from the rules
of the grammar, a recognizer for its language.

We establish first the existence of a universal aperiodic regular
series-parallel grammar, in order to prove the ``only if'' direction
of Theorem~\ref{thm:main} for the $\algof{SP}$ class (Corollary
\ref{cor:completeness}). As before, we annotate the rules of the
grammar with their type (Definition~\ref{def:stratified-grammar}).

\begin{lemma}\label{lemma:universal-synt-reg-sp}
  Let $\grammar_\algof{SP} \isdef(\set{S,P},\rules)$ be the regular
  series-parallel grammar having the following rules, for all $b \in \alphabetTwo$:
  \vspace*{-1.5\baselineskip}
  \begin{center}
    \begin{minipage}{1.5cm}
      \begin{align*}
        \rightarrow & ~P \text{ (\ref{it4:syntactic-regular})} \\[-1mm]
        \rightarrow & ~S \text{ (\ref{it4:syntactic-regular})}
      \end{align*}
    \end{minipage}
    \begin{minipage}{2cm}
      \begin{align*}
        S \rightarrow & ~\sgraph{b}_{1,2} \text{ (\ref{it31:syntactic-regular})} \\[-1mm]
        P \rightarrow & ~\sgraph{b}_{1,2} \text{ (\ref{it32:syntactic-regular})}
      \end{align*}
    \end{minipage}
    \begin{minipage}{2.5cm}
      \begin{align*}
        S \rightarrow & ~P \sop S \text{ (\ref{it31:syntactic-regular})} \\[-1mm]
        S \rightarrow & ~P \sop P \text{ (\ref{it31:syntactic-regular})}
      \end{align*}
    \end{minipage}
    \begin{minipage}{2.5cm}
      \begin{align*}
        P \rightarrow & ~P \pop S \text{ (\ref{it1:syntactic-regular})} \\[-1mm]
        P \rightarrow & ~S \pop S \text{ (\ref{it2:syntactic-regular})}
      \end{align*}
    \end{minipage}
  \end{center}
  Then, we have $\alangof{}{\algof{SP}}{\grammar_\algof{SP}} =
  \universeOf{SP}$.
\end{lemma}
\begin{proof}\emph{Proof.}
  ``$\subseteq$'' Each graph $\graph \in
  \alangof{}{\algof{SP}}{\grammar}$ is series-parallel, because the
  definition of $\grammar_\algof{SP}$ uses only operations from
  $\spsignature$.

  \vspace*{\baselineskip}\noindent ``$\supseteq$'' Let $\graph \in
  \universeOf{SP}$ be a series-parallel graph. By \emph{$X$-derivation
    of $\graph$} we mean a complete derivation $X \step{\grammar}^* t$
  such that $t^\algof{SP} = \graph$. We prove the following facts
  simultaneously: \begin{compactitem}[-]
  \item If $\graph$ is an S-graph then there exists a $S$-derivation
    of $\graph$,
  \item If $\graph$ is a P-graph then there exists a $P$-derivation of
    $\graph$.
  \end{compactitem}
  Consider the following cases for the
  induction:

  \vspace*{.5\baselineskip}\noindent \underline{$\graph =
    \sgraph{b}_{1,2}^\algof{SP}$, for some $b \in \alphabetTwo$}: The
  $P$-derivation applies the rule $P \rightarrow \sgraph{b}_{1,2}$ of
  $\grammar_\algof{SP}$.

  \vspace*{.5\baselineskip}\noindent \underline{$\graph$ is a P-graph
    and not a single edge}: in this case, we have \(\graph = \graph_1
  \pop^\algof{SP} \ldots \pop^\algof{SP} \graph_k\), for some $k\geq
  2$, where $\graph_1, \ldots, \graph_k$ are either single edges or
  S-graphs, by Lemma \ref{lemma:sp-decomp}. Let $i \in \interv{1}{k}$
  be an index.  If $\graph_i=\sgraph{b}_{1,2}^\algof{SP}$ is a single
  edge, there exists an S-derivation of $\graph_i$ using the rule $S
  \rightarrow \sgraph{b}_{1,2}$ of $\grammar_\algof{SP}$. Otherwise,
  there is a $S$-derivation of $\graph_i$, by the inductive
  hypothesis. We obtain a $P$-derivation of $\graph$ by applying the rule $P \rightarrow S \pop S$ once and $k-2$ times the rule $P
  \rightarrow P \pop S$.

  \vspace*{.5\baselineskip}\noindent \underline{$\graph$ is an
    S-graph}: in this case, we have \(\graph = \graph_1
  \sop^\algof{SP} \ldots \sop^\algof{SP} \graph_k\), for some $k \geq
  2$, where $\graph_1, \ldots, \graph_{k}$ are P-graphs, by Lemma
  \ref{lemma:sp-decomp}. Then, there exist a $P$-derivation of each
  $\graph_i$, for $i \in \interv{1}{k}$, by the inductive
  hypothesis. We obtain an $S$-derivation of $\graph$ by applying once
  the rule $S \rightarrow P \sop P$ and $k-2$ times the rule $S
  \rightarrow P \sop S$.
\end{proof}
\noindent Note that $\fpof{\grammar_\algof{SP}} \subseteq
\afp_\algof{SP}$ (Definition~\ref{def:regular-sp-grammar}) and the
only rule of the form (\ref{it1:syntactic-regular}) in
$\grammar_\algof{SP}$ has exponent~1, hence $\grammar_\algof{SP}$ is
an aperiodic regular series-parallel grammar.

To prove the ``if'' direction of Theorem~\ref{thm:main} for the
$\algof{SP}$ class, let us fix a normalized regular series-parallel grammar
$\grammar=(\mathcal{S}\uplus\mathcal{P},\rules)$. Since
each stratified grammar can be normalized, this assumption
loses no generality, by Lemma~\ref{lemma:normal-form}). In the rest of
this section, we give the definition of a $\spsignature$-recognizer
$(\algof{A}_\grammar,B_\grammar)$ for the language
$\alangof{}{\algof{SP}}{\grammar}$, together with upper bounds on the
size of the algebra $\algof{A}_\grammar$ and the time required to
compute the interpretation of a function symbol from $\spsignature$.

For two $\spsignature$-terms $u$ and $w$, we denote by $u \assoc w$
the fact that $w$ can be obtained from $u$ by zero or more
applications of the associativity axiom for serial composition:
\begin{align}\label{eq:assoc-sp}
(x \sop y) \sop z = x \sop (y \sop z)
\end{align}
In the following, we write $u \astep{\grammar}^* w$ iff $u
\step{\grammar}^* v$ and $v \assoc w$, for some $\spsignature$-term
$v$. The following lemmata are used to (de-)compose partial
derivations modulo associativity of the serial composition:

\begin{lemma}\label{lemma:assoc-sp-start}
  Let $S \in \mathcal{S}$ and $Q \in \mathcal{P} \uplus \mathcal{S}$ be nonterminals.
  Then, for every derivation $S \astep{\grammar}^* u \sop Q$, there are non-terminals $P_1, \ldots, P_n$, derivations $P_i \step{\grammar}^* v_i$ and rules $S_i \rightarrow P_i \sop S_{i+1}$, for $1 \le i < n-1$, with $S=S_1$, and a rule $S_{n-1} \rightarrow P_n \sop Q$ such that
  $v_1 \sop (v_2 \sop \ldots (v_n) \ldots) \assoc u$.
\end{lemma}
\begin{proof}\emph{Proof.}
  We consider the derivation $S \step{\grammar}^* t[Q]$, where $t[Q]$ is a term such that $t[Q] \assoc u \sop Q$.
  We note that the only rules of a regular SP-grammar that generate series compositions are the rules of shape $S_1 \rightarrow P \sop S_2$ and $S \rightarrow P_1 \sop P_2$.
  Further, because $Q$ occurs at the right-most position of a series-composition of terms, we must have (up to a reordering of the rule applications) that  the derivation $S \step{\grammar}^* t[Q]$ has the
  form
  \begin{align*} S_1 \step{\grammar} & ~P_1 \sop
  S_2 \step{\grammar} P_1 \sop (P_2 \sop S_3) \\
  \ldots & \\
  \step{\grammar} & ~P_1 \sop
  (P_2 \sop \ldots (P_n \sop Q) \ldots) \\
  \step{\grammar}^* &
  ~v_1 \sop (P_2 \sop \ldots (P_n \sop Q) \ldots) \\
  \step{\grammar}^* & ~v_1 \sop (v_2 \sop \ldots (P_n \sop
  Q) \ldots) \\ \ldots \\ \step{\grammar}^* & ~v_1 \sop
  (v_2 \sop \ldots (v_n \sop Q) \ldots) \end{align*}
  where $S=S_1$ and
  $P_i \step{\grammar}^* v_i$ are complete derivations.
  Then, $v_1 \sop
  (v_2 \sop \ldots (v_n \sop Q) \ldots) \assoc u \sop Q$, and we get that $v_1 \sop (v_2 \sop \ldots (v_n) \ldots) \assoc u$.
\end{proof}

\begin{lemma}\label{lemma:assoc-sp}
  Let $S \in \mathcal{S}$ and $Q \in \mathcal{P} \uplus \mathcal{S}$
  be nonterminals and $\graph$ be an S-graph. For all graphs
  $\graph_1,\graph_2\in\universeOf{SP}$, such that
  $\graph=\graph_1\sop^\algof{SP}\graph_2$, the following are
  equivalent: \begin{compactenum}[1.]
  \item\label{it1:lemma:assoc-sp} there exists a derivation $S
    \astep{\grammar}^* u \sop Q$, where $u$ is a ground
    $\spsignature$-term such that $u^\algof{SP}=\graph$,
  \item\label{it2:lemma:assoc-sp} there exist derivations $S
    \astep{\grammar}^* u_1 \sop S'$ and $S' \astep{\grammar}^* u_2
    \sop Q$, where $S' \in \mathcal{S}$ and $u_i$ are ground
    $\spsignature$-terms such that $u_i^\algof{SP}=\graph_i$, for
    $i=1,2$.
  \end{compactenum}
\end{lemma}
\begin{proof}\emph{Proof.}
  ``(\ref{it1:lemma:assoc-sp}) $\Rightarrow$
  (\ref{it2:lemma:assoc-sp})''
    We assume that there exists a derivation $S
    \astep{\grammar}^* u \sop Q$, where $u$ is a ground
    $\spsignature$-term such that $u^\algof{SP}=\graph$.
    By Lemma~\ref{lemma:assoc-sp-start}, there are non-terminals $P_1, \ldots, P_n$, derivations $P_i \step{\grammar}^* v_i$ and rules $S_i \rightarrow P_i \sop S_{i+1}$, for $1 \le i < n-1$, with $S=S_1$, and a rule $S_{n-1} \rightarrow P_n \sop Q$ such that   $v_1 \sop (v_2 \sop \ldots (v_n) \ldots) \assoc u$.
    Hence, $v_1^\algof{SP} \sop^\algof{SP} \ldots \sop^\algof{SP} v_n^\algof{SP} = u^\algof{SP}  = \graph_1 \sop^\algof{SP} \graph_2$.
    Then, there exists $i \in \interv{1}{n-1}$ such that $\graph_1 = v_1^\algof{SP} \sop^\algof{SP} \ldots \sop^\algof{SP}  v_i^\algof{SP}$ and $\graph_2 =  v_{i+1}^\algof{SP} \sop^\algof{SP} \ldots \sop^\algof{SP} v_n^\algof{SP}$.
    We choose $S' \isdef S_{i+1}$ and check that indeed $S \astep{\grammar}^* u_1 \sop S'$ and $S' \astep{\grammar}^* u_2 \sop Q$, where $u_1$ and $u_2$ are ground terms such that  $u_1 \assoc v_1 \sop \ldots \sop v_i$ and $u_2 \assoc v_{i+1} \sop \ldots \sop v_n$.

  \vspace*{.5\baselineskip} \noindent``(\ref{it2:lemma:assoc-sp})
  $\Rightarrow$ (\ref{it1:lemma:assoc-sp})''
  We assume that there exist derivations $S    \astep{\grammar}^* u_1 \sop S'$ and $S' \astep{\grammar}^* u_2 \sop Q$, where $S' \in \mathcal{S}$ and $u_i$ are ground   $\spsignature$-terms such that $u_i^\algof{SP}=\graph_i$, for $i=1,2$.
  By Lemma~\ref{lemma:assoc-sp-start},
  there are non-terminals $P_1, \ldots, P_n$, derivations $P_i \step{\grammar}^* v_i$ and rules $S_i \rightarrow P_i \sop S_{i+1}$, for $1 \le i < n-1$, with $S_1 = S$ and $S_{k+1} = S'$   and a rule $S_{n-1} \rightarrow P_n \sop Q$ such that $u_1 \assoc v_1 \sop \ldots \sop
  v_k$ and $u_2 \assoc v_{k+1} \sop \ldots \sop v_n$.
  By composing the
  above derivations, we obtain a derivation $S \astep{\grammar}^* (v_1
  \sop \ldots \sop v_k) \sop (v_{k+1} \sop \ldots \sop v_n) \sop
  Q$. Moreover, we have $(v_1 \sop \ldots \sop v_k)^\algof{SP}
  \sop^\algof{SP} (v_{k+1} \sop \ldots \sop v_n)^\algof{SP} =
  u_1^\algof{SP} \sop^\algof{SP} u_2^\algof{SP} = \graph_1
  \sop^\algof{SP} \graph_2 = \graph$, hence we can choose $u \isdef
  v_1 \sop \ldots v_n$, thus obtaining $S \astep{\grammar}^* u \sop Q$
  and $u^\algof{SP} = \graph$.
\end{proof}

\begin{figure}[t!]
  \centerline{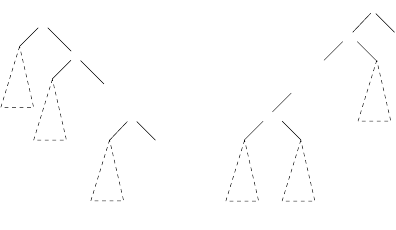} \caption{A derivation $S \step{\grammar}^* (v_1 \sop (v_2 \sop ( \ldots  Q) \ldots)$ (a)
  A view of the graph $\graph = (v_1 \sop \ldots \sop v_n)^\algof{SP}$ corresponding to $S \astep{\grammar}^* (v_1 \sop v_2 \ldots v_n) \sop Q$ (b)   } \label{fig:sp-view} \vspace*{-\baselineskip}
\end{figure}

In order to simplify the development of the recognizer we consider regular series-parallel grammars \emph{in alternative form}:

\begin{definition}\label{def:regular-sp-grammar-alternative}
  A regular series-parallel grammar in \emph{alternative} form is a stratified
  $\spsignature$-grammar
  $\grammar=(\mathcal{P}\uplus\mathcal{S},\rules)$ such that $\fpof{\grammar} \subseteq \afp_\algof{SP}$, where
  \begin{align*}
    \afp_\algof{SP} \isdef & ~\set{\P \rightarrow \pexp{\S}{\geq 1}, \S \rightarrow \P \sop \S, \S \rightarrow \P \sop \P, \rightarrow \P, \rightarrow \S} ~\cup \\
    & ~\set{\S \rightarrow \sgraph{b}_{1,2} \mid b \in \alphabetTwo},
  \end{align*}
  and we require in addition that for every rule $P \rightarrow S \in \rules$ there is some rule
  $S \rightarrow \sgraph{b}_{1,2} \in \rules$ and $S$ does not occur in any other rule of $\rules$.
\end{definition}
Every regular grammar $\grammar$ can be transformed into a regular
grammar in alternative form
$\widehat{\grammar}=(\widehat{\mathcal{S}}\uplus\mathcal{P},\widehat{\rules})$,
by setting $\widehat{\mathcal{S}}\isdef\mathcal{S}\uplus\set{S_b\mid
  b\in\alphabetTwo}$, for some fresh nonterminals $S_b \not\in
\mathcal{S}$, and replacing each rule $P \rightarrow
\sgraph{b}_{1,2}\in \rules$ by the rules $P \rightarrow S_b, S_b
\rightarrow \sgraph{b}_{1,2} \in \widehat{\rules}$.  It is
straightforward to see that $\alangof{}{SP}{\widehat{\grammar}} =
\alangof{}{SP}{\grammar}$ (in every derivation one can replace the
derivation steps $P \rightarrow S_b, S_b \rightarrow \sgraph{b}_{1,2}$
with $P \rightarrow \sgraph{b}_{1,2}$, and vice verse), and that
$\widehat{\grammar}$ is aperiodic iff $\grammar$ is aperiodic.
Clearly, $\size{\grammar} = \Theta(\size{\widehat{\grammar}})$.
Similarly, every regular grammar in alternative form can be turned
into a regular grammar.  Hence, the regular grammars and the regular
grammars in alternative form are in a one-to-one
correspondence\footnote{We note that regular grammars in alternative
form are not regular grammars in general because rules $P \rightarrow
S$ do not satisfy the footprint constraint $\P \rightarrow
\pexp{\S}{\geq2}$ from Definition~\ref{def:regular-sp-grammar}.}.  We
work with regular grammars in alternative form in the following
because we can define the views of P-graphs as multisets of
nonterminals from $\mathcal{S}$:

\begin{definition}\label{def:sp-view}
  Let $\graph$ be a series-parallel graph and
  $\grammar=(\mathcal{S}\uplus\mathcal{P},\rules)$ be a normalized regular series-parallel grammar in alternative form.
  We distinguish the
  following cases:
  \begin{compactenum}
  \item\label{it1:sp-view} $\graph$ is a P-graph: in this case, a
    \emph{view of $\graph$ for $\grammar$} is a multiset
    $\mset{S_1,\ldots,S_n}\in\mpow{\mathcal{S}}$ for which
    there exist complete derivations $S_i \step{{\grammar}}^*
    v_i$, for all $i \in \interv{1}{n}$, such that $\graph = (v_1 \pop
    \ldots \pop v_n)^\algof{SP}$. A \emph{reduced view} of $\graph$ is
    a multiset $\trunk{m}_\grammar$ such that $m$ is a view of
    $\graph$, for $\grammar$.
  \item\label{it2:sp-view} $\graph$ is an S-graph: in this case, a
    view of $\graph$ for $\grammar$ is a pair $(S,Q) \in \mathcal{S}
    \times (\mathcal{S}\uplus\mathcal{P}\uplus\set{\bot})$ for which
    there exists a derivation, either $S \astep{\grammar}^* v \sop Q$,
    if $Q \neq \bot$, or $S \step{\grammar}^* v$, if $Q = \bot$, for a
    ground term $v$, such that $\graph=v^\algof{SP}$. This view is
    reduced, by definition.
  \end{compactenum}
  The \emph{profile of $\graph$ for $\grammar$} is the set
  $h_\grammar(\graph)$ of reduced views of $\graph$ for $\grammar$.
  We call $h_\grammar(\graph)$ a \emph{P-profile} (resp. \emph{S-profile})
  if $\graph$ is a P-graph (resp. an S-graph).
\end{definition}

The distinction between P- and S-graphs plays an important role in the definition of (reduced) views.
Intuitively, a view $\mset{S_1,\ldots,S_n}$ of a P-graph $\graph$ denotes the precise information of how $\graph$ can be obtained as a possible (sub-)derivation of some larger derivation;
here, the alternative form of the regular grammars allows us to also represent single edges as such a multiset $\mset{S}$.
On the other hand, the views of S-graphs are pairs $(S,Q) \in \mathcal{S}\times(\mathcal{S}\uplus\mathcal{P}\uplus\set{\bot})$ that correspond to all possible (sub-)derivations of some larger derivation.
Due to the associativity of the serial composition, the ground subterms in such a derivation can be grouped into a single ground term $v$ with $S \astep{\grammar}^*
v \sop Q$, see Figure \ref{fig:sp-view}.

We define the finite algebra $\algof{A}_\grammar$ having domain
$\universeOf{A}_\grammar \isdef \set{h_\grammar(\graph) \mid \graph \in \universeOf{SP}}$,
where the interpretations of the function symbols from
$\fsignature_\algof{SP}$ are the least sets built according to the
inference rules from Figure \ref{fig:sp-alg}. The rules that define
$\pop^{\algof{A}_\grammar}$ and $\sop^{\algof{A}_\grammar}$ are named
according to the type of their arguments, e.g., $\S\sop^1\P$ and
$\S\sop^2\P$ define the result of $\sop^{\algof{A}_\grammar}$ when the
first argument is an S-profile and the second is a P-profile, etc.

\begin{figure}[t!]
{\small\begin{center}
  \begin{minipage}{7cm}
    \begin{prooftree}
      \AxiomC{$S \rightarrow \sgraph{b}_{1,2} \in \rules$}
      \RightLabel{$\sgraph{b}_{1,2}$}
      \UnaryInfC{$\mset{S} \in \sgraph{b}_{1,2}^{\algof{A}_\grammar}$}
    \end{prooftree}
  \end{minipage}
  \begin{minipage}{7cm}
    \begin{prooftree}
      \AxiomC{$\set{m_i \in a_i}_{i=1,2}$}
      \RightLabel{$\P \pop \P$}
      \UnaryInfC{$\trunk{m_1+m_2}_\grammar \in a_1 \pop^{\algof{A}_\grammar} a_2$}
    \end{prooftree}
  \end{minipage}

  \vspace*{.5\baselineskip}
  \begin{minipage}{7cm}
    \begin{prooftree}
      \AxiomC{$m \in a_1,~ (S,\bot) \in a_2$}
      \RightLabel{$\P \pop \S$}
      \UnaryInfC{$\trunk{m+\mset{S}}_\grammar \in a_1 \pop^{\algof{A}_\grammar} a_2$}
    \end{prooftree}
  \end{minipage}
  \begin{minipage}{7cm}
    \begin{prooftree}
      \AxiomC{$\set{(S_i,\bot) \in a_i}_{i=1,2}$}
      \RightLabel{$\S \pop \S$}
      \UnaryInfC{$\trunk{\mset{S_1}+\mset{S_2}}_\grammar \in a_1 \pop^{\algof{A}_\grammar} a_2$}
    \end{prooftree}
  \end{minipage}

  \vspace*{.5\baselineskip}
  \begin{minipage}{7cm}
    \begin{prooftree}
      \AxiomC{$\begin{array}{c}
          \set{P_i \leadsto_{{\grammar}} m_i \in a_i}_{i=1,2} \\
          S \rightarrow P_1 \sop S_1,~ S_1 \rightarrow P_2 \sop Q \in \rules
          \end{array}$}
      \RightLabel{$\P\sop^1\P$}
      \UnaryInfC{$(S,Q) \in a_1 \sop^{\algof{A}_\grammar} a_2$}
    \end{prooftree}
  \end{minipage}
  \begin{minipage}{7cm}
    \begin{prooftree}
       \AxiomC{$\begin{array}{c}
          \set{P_i \leadsto_{{\grammar}} m_i \in a_i}_{i=1,2} \\
          S \rightarrow P_1 \sop P_2 \in \rules
          \end{array}$}
      \RightLabel{$\P\sop^2\P$}
      \UnaryInfC{$(S,\bot) \in a_1 \sop^{\algof{A}_\grammar} a_2$}
    \end{prooftree}
  \end{minipage}

  \vspace*{.5\baselineskip}
  \begin{minipage}{7cm}
    \begin{prooftree}
       \AxiomC{$\begin{array}{c}
           P \leadsto_{{\grammar}} m \in a_1,~ (S_1,Q) \in a_2 \\
           S \rightarrow P \sop S_1 \in \rules
          \end{array}$}
      \RightLabel{$\P\sop\S$}
      \UnaryInfC{$(S,Q) \in a_1 \sop^{\algof{A}_\grammar} a_2$}
    \end{prooftree}
  \end{minipage}
  \begin{minipage}{7cm}
    \begin{prooftree}
      \AxiomC{$\begin{array}{c}
          \\
          (S,S_1) \in a_1,~ (S_1,Q) \in a_2
        \end{array}$}
      \RightLabel{$\S\sop\S$}
      \UnaryInfC{$(S,Q) \in a_1 \sop^{\algof{A}_\grammar} a_2$}
    \end{prooftree}
  \end{minipage}

  \vspace*{.5\baselineskip}
  \begin{minipage}{7cm}
    \begin{prooftree}
      \AxiomC{$\begin{array}{c}
          (S,S_1) \in a_1,~ P \leadsto_{{\grammar}} m \in a_2 \\
          S_1 \rightarrow P \sop Q \in \rules
          \end{array}$}
      \RightLabel{$\S\sop^1\P$}
      \UnaryInfC{$(S,Q) \in a_1 \sop^{\algof{A}_\grammar} a_2$}
    \end{prooftree}
  \end{minipage}
  \begin{minipage}{7cm}
    \begin{prooftree}
      \AxiomC{$\begin{array}{c}
          \\
          (S,P) \in a_1,~ P \leadsto_{{\grammar}} m \in a_2
        \end{array}$}
      \RightLabel{$\S\sop^2\P$}
      \UnaryInfC{$(S,\bot) \in a_1 \sop^{\algof{A}_\grammar} a_2$}
    \end{prooftree}
  \end{minipage}
\end{center}}
\caption{The interpretation of the signature $\spsignature$ in $\algof{A}_\grammar$}
\label{fig:sp-alg}
\end{figure}

We give upper bounds on the size of the recognizer algebra and the
time needed to evaluate the interpretation of the function
symbols.
The size of an element $a \in \universeOf{A}_\grammar$ is
defined as follows: \begin{align*}
  \sizeof{a} \isdef & ~\left\{\begin{array}{ll}
  \sum_{m \in a} \cardof{m} \text{, if } a \text{ is a P-profile} \\[1mm]
  2\cdot\cardof{a} \text{, if } a \text{ is an S-profile}
  \end{array}\right.
\end{align*}

\begin{lemma}\label{lemma:size-of-sp-rec}
  Let $\grammar = (\mathcal{S}\uplus\mathcal{P}, \rules)$ be a normalized regular series-parallel grammar in alternative form.
  Then, $\cardof{\universeOf{A}_\grammar} \in
  2^{{\sizeof{\grammar}}^{\poly{\sizeof{\grammar}}}}$, $\sizeof{a} \in \sizeof{\grammar}^{\poly{\sizeof{\grammar}}}$, for each $a \in
  \universeOf{A}_\grammar$ and
  $f^{\algof{A}_\grammar}(a_1,\ldots,a_{\arityof{f}})$ is computable in time ${\sizeof{\grammar}}^{\poly{\sizeof{\grammar}}}$, for each function symbol $f \in \fsignature_\algof{SP}$ and all elements  $a_1,\ldots,a_{\arityof{f}} \in \universeOf{A}_\grammar$.
\end{lemma}
\begin{proof}\emph{Proof.}
  As in the proof of Lemma~\ref{lemma:size-of-tree-rec}, we obtain
  that there are at most
  ${\sizeof{\grammar}}^{\poly{\cardof{\mathcal{S}}}}$
  different reduced multisets and hence at most
  $2^{{\sizeof{\grammar}}^{\poly{\cardof{\mathcal{S}}}}}$
  different sets of reduced multisets.
  Moreover, there are
  $\cardof{\mathcal{S}} \cdot (\cardof{\mathcal{S}} +
  \cardof{\mathcal{P}} + 1) \in \poly{\cardof{\mathcal{S}} +
    \cardof{\mathcal{P}}}$ many different pairs of the form $(S,Q) \in
  \mathcal{S}\times(\mathcal{S}\uplus\mathcal{P}\uplus\set{\bot})$. Hence,
  there are at most different $2^ \poly{\cardof{\mathcal{S}} +
    \cardof{\mathcal{P}}}$ sets of such pairs. Thus, we obtain: \begin{align*}
    \cardof{\universeOf{A}_\grammar} \in & ~2^{{\sizeof{\grammar}}^{\poly{\cardof{\mathcal{S}}}}} + 2^ \poly{\cardof{\mathcal{S}} + \cardof{\mathcal{P}}} \\
    \in & ~2^{{\sizeof{\grammar}}^{\poly{\sizeof{\grammar}}}}
  \end{align*}
  We note that an element $a \in \universeOf{A}_\grammar$ is either a
  set of reduced multisets or a set of pairs $(S,Q)\in
  \mathcal{S}\times(\mathcal{S}\uplus\mathcal{P}\uplus\set{\bot})$. As
  in the proof of Lemma~\ref{lemma:size-of-tree-rec}, there are at
  most ${\sizeof{\grammar}}^{\poly{\cardof{\mathcal{S}}}}$
  different reduced multisets $m$, for which we have $\cardof{m} \in
  {\sizeof{\grammar}}^{\poly{\cardof{\mathcal{S}}}}$.
  Hence, for each $a \in \universeOf{A}_\grammar$, we
  obtain: \begin{align*}
    \sizeof{a} \in & ~\max\left(
    {\sizeof{{\grammar}}}^{\poly{\cardof{\mathcal{S}}}},
    \poly{\cardof{\mathcal{S}} + \cardof{\mathcal{P}}}\right) \\
    \in & ~{\sizeof{\grammar}}^{\poly{\sizeof{\grammar}}}
  \end{align*}
  We compute an upper bound on the time needed to evaluate the
  function $f^{\algof{A}_\grammar}$, for each $f \in
  \fsignature_\algof{SP}$: \begin{itemize}[-]
  \item $\sgraph{b}_{1,2}^{\algof{A}_\grammar}$: is computed in linear
    time by going over all the rules of ${\grammar}$, where
    $\sizeof{{\grammar}} = \bigO(\sizeof{\grammar})$.
  \item $a_1 \pop^\algof{A} a_2$: We note that for P-profiles
    $a_1,a_2$, the value of $a_1 \pop^\algof{A} a_2$ can be computed
    in time
    ${\sizeof{\grammar}}^{\poly{\cardof{\mathcal{S}}}} \subseteq
    \sizeof{\grammar}^{\poly{\sizeof{\grammar}}}$, by the same
    arguments as in the proof of
    Lemma~\ref{lemma:size-of-tree-rec}. We note that the cases, where
    at least one of $a_1$ and $a_2$ is an S-profile are similar.
  \item $a_1 \sop^\algof{A} a_2$: First, we assume that $a_1$ and
    $a_2$ are both S-profiles. Then, the computation of $a_1
    \sop^\algof{A} a_2$ requires to check for every pair $(S,Q)$
    whether there is an $S_1 \in \mathcal{S}$ such that $(S,S_1) \in
    a_1$ and $(S_1,Q) \in a_2$.  Clearly, this can be done in
    $\poly{\cardof{\mathcal{S}}+ \cardof{\mathcal{P}}}
    \subseteq \sizeof{\grammar}^{\poly{\sizeof{\grammar}}}$.  In case
    $a_1$ is a P-profile and $a_2$ is an S-profile, for every pair
    $(S,Q)$, we need to iterate over all rules $S \rightarrow P \sop
    S_1 \in \rules$ and $m \in a_1$ and check whether $(S_1,Q) \in
    a_2$ and $P \leadsto_{{\grammar}} m$.  We note that there are
    at most
    ${\sizeof{{\grammar}}}^{\poly{\cardof{\mathcal{S}}}}$
    many multisets $m \in a_1$. As in the proof of Lemma
    \ref{lemma:size-of-tree-rec}, the check $P
    \leadsto_{{\grammar}} m$, for a given nonterminal $P$ and
    multiset $m$, takes time
    ${\sizeof{\grammar}}^{\poly{\cardof{\mathcal{S}}}}$.
    As there are at most $\poly{\cardof{\mathcal{S}}+
      \cardof{\mathcal{P}}}$ many pairs $(S_1,Q) \in a_2$, we overall
    get that $a_1 \sop^\algof{A} a_2$, in this case, can be computed
    in ${\sizeof{\grammar}}^{\poly{\cardof{\mathcal{S}}}} \subseteq
    \sizeof{\grammar}^{\poly{\sizeof{\grammar}}}$. The other cases are
    similar and give the same asymptotic complexity bound.
  \end{itemize}
\end{proof}

We establish now the counterparts of Lemmas
\ref{lemma:tree-view-membership} and
\ref{lemma:tree-view-homomorphism}. First, equality of profiles
amounts to indistinguishability of graphs with respect to membership
in $\alangof{}{SP}{\grammar}$:

\begin{lemma}\label{lemma:sp-view-membership}
  For all $\graph_1, \graph_2 \in \universeOf{SP}$, if
  $h_\grammar(\graph_1)=h_\grammar(\graph_2)$
  then $\graph_1\in\alangof{}{SP}{\grammar} \iff
  \graph_2\in\alangof{}{SP}{\grammar}$.
\end{lemma}
\begin{proof}\emph{Proof.}
  If
  $h_\grammar(\graph_1)=h_\grammar(\graph_2)$
  then $\graph_1$ and $\graph_2$ are either both P-graphs or both
  S-graphs. Assume that $\graph_1\in\alangof{}{SP}{\grammar}$. We prove that
  $\graph_2\in\alangof{}{SP}{\grammar}$ by considering the following
  cases:

  \vspace*{.5\baselineskip}\noindent \underline{$\graph_1,\graph_2$
    are P-graphs}: Since $\graph_1 \in
  \alangof{}{SP}{{\grammar}}$, there exists a complete
  derivation $P \step{{\grammar}}^* v$, for some axiom
  $\rightarrow P$ of ${\grammar}$ such that $v^\algof{SP} =
  \graph_1$. Without loss of generality, the derivation can be
  reorganized as:
    \begin{align*}
      P \step{{\grammar}}^* ~\pop_{S \in \mathcal{S}} ~\pexp{S}{m_1(S)}
      \step{{\grammar}}^* ~\pop_{S \in \mathcal{S}, i \in
        \interv{1}{m_1(S)}} v^1(S,i)
    \end{align*}
    for a multiset $m_1 \in \mpow{\mathcal{S}}$ and ground
    terms $v_1(S,i)$ such that $\graph_1 = (\pop_{S \in
      \mathcal{S},~ i \in \interv{1}{m_1(S)}}
    v_1(S,i))^\algof{SP}$. Then, $m_1$ is a view of $\graph_1$, hence
    $\trunk{m_1}_\grammar \in
    h_\grammar(\graph_1)$. Because
    $h_\grammar(\graph_1)=h_\grammar(\graph_2)$, there exists a view
    $m_2$ of $\graph_2$ such that $\trunk{m_2}_\grammar =
    \trunk{m_1}_\grammar$. Moreover, $P \leadsto_{{\grammar}}
    m_1$, hence $P \leadsto_{{\grammar}}
    \trunk{m_1}_\grammar =
    \trunk{m_2}_\grammar$ and $P
    \leadsto_{{\grammar}} m_2$, by Lemma \ref{fact:trunk}. We
    obtain a derivation:
    \begin{align*}
      P \step{{\grammar}}^* ~\pop_{S \in \mathcal{S}} ~\pexp{S}{m_2(S)}
      \step{{\grammar}}^* ~\pop_{S \in \mathcal{S}, i \in
      \interv{1}{m_2(S)}} v_2(S,i)
    \end{align*}
    for some ground terms $v_2(S,i)$ such that $\graph_1 = (\pop_{S
      \in \mathcal{S},~ i \in \interv{1}{m_2(S)}}
    v_2(S,i))^\algof{SP}$. Since $\rightarrow P$ is an axiom of
    ${\grammar}$, we obtain that $\graph_2 \in
    \alangof{}{SP}{{\grammar}}$.

    \vspace*{.5\baselineskip}\noindent \underline{$\graph_1,\graph_2$
      are S-graphs}: Since $\graph_1
    \in\alangof{}{SP}{{\grammar}}$, there exists a complete
    derivation $S \step{{\grammar}}^* v_1$, such that
    $\graph_1 = v^\algof{SP}$, for some axiom $\rightarrow S$ of
    ${\grammar}$. Then, $(S,\bot) \in
    h_\grammar(\graph_1) =
    h_\grammar(\graph_2)$, hence there exists a complete
    derivation $S \step{{\grammar}}^* v_2$ such that
    $v_2^\algof{SP}=\graph_2$, leading to $\graph_2 \in
    \alangof{}{SP}{{\grammar}}$.

    \vspace*{.5\baselineskip}\noindent The direction
    $\graph_2\in\alangof{}{SP}{\grammar} \Rightarrow
    \graph_1\in\alangof{}{SP}{\grammar}$ is symmetric.
\end{proof}

\begin{lemma}\label{lemma:sp-view-homomorphism}
  $h_\grammar$ is a homomorphism between the
  $\spsignature$-algebras $\algof{SP}$ and $\algof{A}_\grammar$.
\end{lemma}
\begin{proof}\emph{Proof.}
  We prove the following points, for all $b \in \alphabetTwo$ and
  $\graph_1, \graph_2 \in \universeOf{SP}$:

  \vspace*{.5\baselineskip}
  \noindent\underline{$h_\grammar(\sgraph{b}_{1,2}^\algof{SP})
    = \sgraph{b}_{1,2}^{\algof{A}_\grammar}$}: Since
  $\sgraph{b}_{1,2}^\algof{SP}$ is a P-graph, we have
  $h_\grammar(\sgraph{b}_{1,2}^\algof{SP}) = \set{\mset{S}
    \mid S \rightarrow b_{1,2} \in \rules} =
  \sgraph{b}_{1,2}^{\algof{A}_\grammar}$, where
  $\sgraph{b}_{1,2}^{\algof{A}_\grammar}$ is obtained by the inference
  rule ($b_{1,2}$) from Figure \ref{fig:sp-alg}.

  \vspace*{.5\baselineskip}
  \noindent\underline{$h_\grammar(\graph_1
    \pop^{\algof{SP}} \graph_2) = h_\grammar(\graph_1)
    \pop^{\algof{A}_\grammar} h_\grammar(\graph_2)$}:
  Since $\graph_1 \pop^{\algof{SP}} \graph_2$ is a P-graph, each
  element of $h_\grammar(\graph_1 \pop^{\algof{SP}}
  \graph_2)$ is a multiset $m \in \mpow{\mathcal{S}}$ for
  which there exists a view $m'$ of $\graph_1 \pop^{\algof{SP}}
  \graph_2$ such that $m = \trunk{m'}_\grammar$. We
  distinguish the following cases, according to the types of $G_1$ and
  $G_2$: \begin{itemize}[-]
  \item $\graph_1$ and $\graph_2$ are P-graphs: ``$\subseteq$'' Let $m
    \in h_\grammar(\graph_1 \pop^{\algof{SP}} \graph_2)$
    and $m'$ be the view of $\graph_1 \pop^{\algof{SP}} \graph_2$ such
    that $m = \trunk{m'}_\grammar$. Then, for each $S \in
    \mathcal{S}$ and $i \in \interv{1}{m'(S)}$, there exists
    a complete derivation $S \step{{\grammar}}^* v(S,i)$,
    where $v(S,i)$ is a ground term such that $\graph_1
    \pop^{\algof{SP}} \graph_2 = (\pop_{S \in \mathcal{S}, i
      \in \interv{1}{m'(S)}} v(S,i))^\algof{SP}$. We split the set of
    ground terms $v(S,i)$ according to whether $v(S,i)^\algof{SP}$ is
    a subgraph of $\graph_1$ or of $\graph_2$, thus obtaining views
    $m'_i$ of $\graph_i$ such that $m' = m'_1 + m'_2$. Since
    $\trunk{m'_i}_\grammar \in
    h_\grammar(\graph_i)$ and $m = \trunk{m'_1 +
      m'_2}_\grammar =
    \trunk{\trunk{m'_1}_\grammar +
      \trunk{m'_2}_\grammar}_\grammar$, by
    Lemma \ref{fact:trunk-two}, we obtain $m \in
    h_\grammar(\graph_1) \pop^{\algof{A}_\grammar}
    h_\grammar(\graph_2)$, by applying the inference rule
    ($\P \pop \P$) from Figure \ref{fig:sp-alg}.

    \noindent''$\supseteq$'' Let $m \in
    h_\grammar(\graph_1) \pop^{\algof{A}_\grammar}
    h_\grammar(\graph_2)$ be a multiset. Since both
    $h_\grammar(\graph_i)$ are P-profiles, for $i=1,2$, by
    the inference rule ($\P \pop \P$) from Figure \ref{fig:sp-alg},
    there exist $m_i \in h_\grammar(\graph_i)$, for
    $i=1,2$, such that $m = \trunk{m_1 +
      m_2}_\grammar$. Then, there exist views $m'_i$ of
    $\graph_i$ such that $m_i = \trunk{m'_i}_\grammar$. By
    Lemma \ref{fact:trunk-two}, we obtain $m = \trunk{m'_1 +
      m'_2}_\grammar$, hence $m'_1 + m'_2$ is a view of
    $\graph_1 \pop^\algof{SP} \graph_2$ for ${\grammar}$, thus
    $m \in h_\grammar(\graph_1 \pop^\algof{SP} \graph_2)$.
  \item $\graph_i$ is a P-graph and $\graph_{3-i}$ is an S-graph, for
    $i = 1,2$: We assume that $i=1$, the proof being the same for the
    case $i=2$, because of the commutativity of $\pop^\algof{SP}$ and
    $\pop^{\algof{A}_\grammar}$. ``$\subseteq$'' Let $m \in
    h_\grammar(\graph_1 \pop^{\algof{SP}} \graph_2)$ and
    $m'$ be the view of $\graph_1 \pop^{\algof{SP}} \graph_2$ such
    that $m = \trunk{m'}_\grammar$. Then, for each $S \in
    \mathcal{S}$ and $i \in \interv{1}{m'(S)}$, there exists
    a complete derivation $S \step{{\grammar}}^* v(S,i)$,
    where $v(S,i)$ is a ground term such that $\graph_1
    \pop^{\algof{SP}} \graph_2 = (\pop_{S \in \mathcal{S}, i
      \in \interv{1}{m'(S)}} v(S,i))^\algof{SP}$. Since $\graph_2$ is
    an S-graph, the only possibility is that $\graph_2 =
    v(S,i)^\algof{SP}$, for a single $S \in \mathcal{S}$ and $i \in
    \interv{1}{m'(S)}$, such that $S \step{{\grammar}}^*
    v(S,i)$ is a complete derivation.  Then, $(S,\bot) \in
    h_\grammar(\graph_2)$. Moreover, $m'_1 \isdef m' -
    \mset{S}$ (where $-$ here denotes multiset difference) is a view
    of $\graph_1$ for ${\grammar}$ and $m_1 \isdef
    \trunk{m'_1}_\grammar \in
    h_\grammar(\graph_1)$. By Lemma \ref{fact:trunk-two},
    we obtain $m = \trunk{m_1 + \mset{S}}_\grammar$, thus
    $m \in h_\grammar(\graph_1) \sop^{\algof{A}_\grammar}
    h_\grammar(\graph_2)$, by applying the inference rule
    ($\P \pop \S$) from Figure \ref{fig:sp-alg}.

    \noindent``$\supseteq$'' Let $m \in
    h_\grammar(\graph_1) \pop^{\algof{A}_\grammar}
    h_\grammar(\graph_2)$ be a multiset. Since
    $h_\grammar(\graph_1)$ is a P-profile and
    $h_\grammar(\graph_2)$ is an S-profile, by the
    inference rule ($\P \pop \S$) from Figure \ref{fig:sp-alg}, there
    exists $m_1 \in h_\grammar(\graph_1)$ and $(S,\bot)
    \in h_\grammar(\graph_2)$ such that $m = \trunk{m_1 +
      \mset{S}}_\grammar$. Then, there exists a view
    $m'_1$ of $\graph_1$ such that $m_1 =
    \trunk{m'_1}_\grammar$. Moreover, there exists a
    complete derivation $S \step{{\grammar}}^* v$ such that
    $v^\algof{SP} = \graph_2$. Hence, $m'_1 + \mset{S}$ is a view of
    $\graph_1 \pop^\algof{SP} \graph_2$ for ${\grammar}$ and
    $m = \trunk{m_1 + \mset{S}}_\grammar =
    \trunk{\trunk{m'_1}_\grammar +
      \mset{S}}_\grammar$, by Lemma \ref{fact:trunk-two},
    thus $m \in h_\grammar(\graph_1 \pop^\algof{SP}
    \graph_2)$.
  \item $\graph_1$ and $\graph_2$ are S-graphs: ``$\subseteq$'' Let $m
    \in h_\grammar(\graph_1 \pop^{\algof{SP}} \graph_2)$
    and $m'$ be the view of $\graph_1 \pop^{\algof{SP}} \graph_2$ such
    that $m = \trunk{m'}_\grammar$. Then, $m' =
    \mset{S_1,S_2}$ and there are complete derivations $S_i
    \step{{\grammar}}^* v_i$, for $i=1,2$, where $v_i$ is a
    ground term such that $\graph_i = v_i^\algof{SP}$.  Then, we
    obtain $(S_i,\bot) \in h_\grammar(\graph_i)$, for
    $i=1,2$. Moreover, $m'=\mset{S_1}+\mset{S_2}$.  Thus, $m \in
    h_\grammar(\graph_1) \pop^{\algof{A}_\grammar}
    h_\grammar(\graph_2)$, by applying the inference rule
    ($\S\pop\S$) from Figure \ref{fig:sp-alg}.

    \noindent``$\supseteq$'' Let $m \in
    h_\grammar(\graph_1) \pop^{\algof{A}_\grammar}
    h_\grammar(\graph_2)$ be a multiset. Since both
    $h_\grammar(\graph_i)$ are S-profiles, for $i=1,2$, by
    the inference rule ($\S\pop\S$) from Figure \ref{fig:sp-alg},
    there exist views $(S_i,\bot) \in
    h_\grammar(\graph_i)$, hence complete derivations $S_i
    \step{{\grammar}} v_i$ such that $\graph_i =
    v_i^\algof{SP}$, for $i=1,2$ and $m = \trunk{\mset{S_1} +
      \mset{S_2}}_\grammar$. Then, $\mset{S_1}+\mset{S_2}$ is a view
    of $\graph_1 \pop^\algof{SP} \graph_2$ for ${\grammar}$,
    hence $m \in h_\grammar(\graph_1 \pop^\algof{SP}
    \graph_2)$.
  \end{itemize}

  \vspace*{.5\baselineskip}
  \noindent\underline{$h_\grammar(\graph_1
    \sop^{\algof{SP}} \graph_2) = h_\grammar(\graph_1)
    \sop^{\algof{A}_\grammar} h_\grammar(\graph_2)$}:
  Since $\graph_1 \sop^{\algof{SP}} \graph_2$ is an S-graph, any view
  of it for ${\grammar}$ is a pair $(S,Q)$, such that $S \in
  \mathcal{S}$ and $Q \in \mathcal{S} \uplus \mathcal{P}
  \uplus \set{\bot}$. We distinguish the following cases, according to
  the types of $\graph_1$ and $\graph_2$: \begin{itemize}[-]
  \item $\graph_1$ and $\graph_2$ are P-graphs: ``$\subseteq$'' Let
    $(S,Q) \in h_\grammar(\graph_1 \sop^{\algof{SP}}
    \graph_2)$ and assume $Q\neq\bot$ (the case $Q = \bot$ is
    similar). By Definition~\ref{def:sp-view} (\ref{it2:sp-view}),
    there exists a derivation $S \astep{{\grammar}}^* v \sop
    Q$, for a ground term $v$ such that $\graph_1 \sop^{\algof{SP}}
    \graph_2 = v^\algof{SP}$.
    By Lemma \ref{lemma:assoc-sp}, there
    exist derivations $S \astep{{\grammar}}^* v_1 \sop S'$ and
    $S' \astep{{\grammar}}^* v_2 \sop Q$, for some nonterminal
    $S' \in \mathcal{S}$ and ground terms $v_i$ such that $\graph_i =
    v_i^\algof{SP}$, for $i=1,2$.
    Then, by Lemma~\ref{lemma:assoc-sp-start}, and because $\graph_1$ and $\graph_2$ are $\sop$-atomic, there are rules $S \rightarrow P_1 \sop S'$ and $S' \rightarrow P_2 \sop Q$, and complete derivations $P_i \step{{\grammar}}^* v_i$ for some non-terminals $P_1,P_2 \in \mathcal{P}$.
    By a reordering, we can
    assume w.l.o.g. that these complete derivations are of the form:
    \begin{align*}
      P_i \step{{\grammar}}^* \pop_{S \in \mathcal{S}} \pexp{S}{m_i(S)} \step{{\grammar}}^* v_i
    \end{align*}
    for some views $m_i$ of $\graph_i$ for ${\grammar}$, for
    $i=1,2$. Then, $P_i \leadsto_{{\grammar}} m_i$, and we
    obtain $P_i \leadsto_{{\grammar}}
    \trunk{m_i}_\grammar$, by Lemma \ref{fact:trunk}, for
    $i=1,2$. Moreover, $\trunk{m_i}_\grammar \in
    h_\grammar(\graph_i)$, for $i=1,2$.
    Thus, we obtain $(S,Q) \in h_\grammar(\graph_1)
    \sop^{\algof{A}_\grammar} h_\grammar(\graph_2)$, by applying the
    inference rule ($\P\sop^1\P$) from Figure \ref{fig:sp-alg}.

    \noindent ``$\supseteq$'' Let $(S,Q) \in
    h_\grammar(\graph_1) \sop^{\algof{A}_\grammar}
    h_\grammar(\graph_2)$ and assume $Q \neq \bot$ (the
    case $Q = \bot$ is similar). Since both
    $h_\grammar(\graph_i)$ are P-profiles, by the
    inference rule ($\P\sop^1\P$) from Figure \ref{fig:sp-alg}, there
    exists rules $S \rightarrow P_1 \sop S_1$ and $S_1 \rightarrow P_2
    \sop Q$ in ${\grammar}$, such that $P_i
    \leadsto_{{\grammar}} m_i \in
    h_\grammar(\graph_i)$, for $i=1,2$.  Then, there are
    views $m_i'$ of $\graph_i$, with
    $\trunk{m_i'}_\grammar = m_i$, for $i=1,2$, hence also
    derivations:
    \begin{align*}
      P_i \step{{\grammar}}^* \pop_{S \in \mathcal{S}} \pexp{S}{m_i'(S)} \step{{\grammar}}^* v_i,
    \end{align*}
    for some ground terms $v_i$, such that $\graph_i =
    v_i^\algof{SP}$.  Then, we obtain $P_i \leadsto_{{\grammar}} m_i'$, by
    Lemma~\ref{fact:trunk}, for $i=1,2$.  Thus, we can build a
    derivation:
    \begin{align*}
      S \step{{\grammar}} P_1 \sop S_1 \step{{\grammar}} P_1 \sop P_2
      \sop Q \astep{{\grammar}}^* v_1 \sop v_2 \sop Q
    \end{align*}
    such that $\graph_1 \sop^\algof{SP} \graph_2 = (v_1 \sop
    v_2)^\algof{SP}$ and $(S,Q) \in h_\grammar(\graph_1
    \sop^\algof{SP} \graph_2)$, by Definition~\ref{def:sp-view}
    (\ref{it2:sp-view}).
    \item $\graph_1$ is a P-graph and $\graph_2$ is an S-graph:
      ``$\subseteq$'' Let $(S,Q) \in h_\grammar(\graph_1
      \sop^{\algof{SP}} \graph_2)$ be a view of $\graph_1
      \sop^{\algof{SP}} \graph_2$ and assume $Q \neq \bot$ (the case
      $Q = \bot$ is similar). By Definition~\ref{def:sp-view}
      (\ref{it2:sp-view}), there exists a derivation $S
      \astep{{\grammar}}^* v \sop Q$, for a ground term $v$
      such that $\graph_1 \sop^{\algof{SP}} \graph_2 =
      v^\algof{SP}$. By Lemma \ref{lemma:assoc-sp}, there exist
      derivations $S \astep{{\grammar}}^* v_1 \sop S_1$ and
      $S_1 \astep{{\grammar}}^* v_2 \sop Q$, for some $S_1 \in
      \mathcal{S}$, such that $v_i^{\algof{SP}} = \graph_i$, for
      $i=1,2$.
      Then, by Lemma~\ref{lemma:assoc-sp-start}, and because $\graph_1$ is $\sop$-atomic, there is a rule $S \rightarrow P \sop S_1$ and a complete derivations $P \step{{\grammar}}^* v_1$ for some non-terminal $P \in \mathcal{P}$.
      By a reordering, if necessary, we assume the latter
      derivation to have the form: \begin{align*} P
        \step{{\grammar}}^* ~\pop_{S \in
          \mathcal{S}} \pexp{S}{m(S)}
        \step{{\grammar}}^* ~\pop_{S \in
          \mathcal{S},i \in \interv{1}{m(S)}} v_1(S,i)
      \end{align*}
      for some multiset $m \in \mpow{\mathcal{S}}$, where $v_1(S,i)$ are ground terms such that $\graph_1 = (\pop_{S
        \in \mathcal{S},i \in \interv{1}{m(S)}}
      v_1(S,i))^\algof{SP}$. Then, $m$ is a view of $\graph_1$ for
      ${\grammar}$ such that $P \leadsto_{{\grammar}}
      m$. By Lemma \ref{fact:trunk}, we obtain $P
      \leadsto_{{\grammar}} \trunk{m}_\grammar \in
      h_\grammar(\graph_1)$. Moreover, since $\graph_2$ is
      an S-graph, we obtain $(S,Q) \in
      h_\grammar(\graph_2)$, by Definition
      \ref{def:sp-view} (\ref{it2:sp-view}). Thus we obtain $(S,Q) \in
      h_\grammar(\graph_1) \sop^{\algof{A}_\grammar}
      h_\grammar(\graph_2)$, by an application of the
      inference rule ($\P\sop\S$) from Figure \ref{fig:sp-alg}.

      \noindent ``$\supseteq$'' Let $(S,Q) \in
      h_\grammar(\graph_1) \sop^{\algof{A}_\grammar}
      h_\grammar(\graph_2)$ and assume that $Q \neq \bot$
      (the case $Q = \bot$ is similar). Since
      $h_\grammar(\graph_1)$ is a P-profile and
      $h_\grammar(\graph_2)$ is an S-profile, by the
      inference rule ($\P\sop\S$), there exists a rule $S \rightarrow
      P \sop S_1$ in ${\grammar}$ such that $P
      \leadsto_{{\grammar}} m$, for some reduced view $m \in
      h_\grammar(\graph_1)$, and $(S_1,Q) \in
      h_\grammar(\graph_2)$. Then, there exists a view
      $m'$ of $\graph_1$, such that $\trunk{m'}_\grammar =
      m$.  We obtain $P \leadsto_{{\grammar}} m'$, by
      Lemma~\ref{fact:trunk}. Hence, there exists a derivation:
      \begin{align*}
        P \step{{\grammar}}^* \pop_{S \in
          \mathcal{S}} \pexp{S}{m'(S)} \step{{\grammar}}^* v_1,
      \end{align*}
      for some ground term $v_1$, such that $\graph_1 =
      v_1^\algof{SP}$. Moreover, $(S_1,Q) \in h_\grammar(\graph_2)$, hence there exists a derivation $S_1
      \astep{{\grammar}}^* v_2 \sop Q$, where $v_2$ is a ground term such
      that $\graph_2 = v_2^\algof{SP}$. We can build a derivation:
      \begin{align*}
        S \step{{\grammar}} P \sop S_1 \step{{\grammar}}^* (\pop_{S \in
          \mathcal{S}} \pexp{S}{m'(S)}) \sop S_1 \astep{{\grammar}}^* v_1 \sop v_2 \sop Q.
      \end{align*}
      Then, $(v_1 \sop v_2)^\algof{SP} = \graph_1 \sop^\algof{SP}
      \graph_2$ and $(S,Q) \in h_\grammar(\graph_1
      \sop^\algof{SP} \graph_2)$, by Definition~\ref{def:sp-view}
      (\ref{it2:sp-view}).
    \item $\graph_1$ is an S-graph and $\graph_2$ is a P-graph:
      ``$\subseteq$'' Let $(S,Q) \in h_\grammar(\graph_1
      \sop^{\algof{SP}} \graph_2)$ and assume that $Q \neq \bot$ (the
      case $Q = \bot$ is similar). By Definition~\ref{def:sp-view}
      (\ref{it2:sp-view}), there exists a derivation $S
      \astep{\grammar}^* v \sop Q$, where $v$ is a ground term such
      that $\graph_1 \sop^{\algof{SP}} \graph_2 = v^\algof{SP}$. By
      Lemma \ref{lemma:assoc-sp}, there exist derivations $S
      \astep{{\grammar}}^* v_1 \sop S_1$ and $S_1
      \astep{{\grammar}}^* v_2 \sop Q$, for some $S_1 \in
      \mathcal{S}$, such that $v_i^{\algof{SP}} = \graph_i$, for
      $i=1,2$. Since $\graph_1$ is an S-graph, by Definition
      \ref{def:sp-view} (\ref{it2:sp-view}), we obtain $(S,S_1) \in
      h_\grammar(\graph_1)$.
      Then, by Lemma~\ref{lemma:assoc-sp-start}, and because $\graph_2$ is $\sop$-atomic, there is a rule $S_1 \rightarrow P \sop Q$ and a complete derivations $P \step{{\grammar}}^* v_2$ for some non-terminal $P \in \mathcal{P}$.
     By a reordering, if necessary, we assume the latter derivation to have the form: \begin{align*} P
        \step{{\grammar}}^* ~\pop_{S \in
          \mathcal{S}} \pexp{S}{m(S)}
        \step{{\grammar}}^* ~\pop_{S \in
          \mathcal{S},i \in \interv{1}{m(S)}} v_2(S,i)
      \end{align*}
      for some multiset $m \in \mpow{\mathcal{S}}$, where
      $v_2(S,i)$ are ground terms and $\graph_2 = (\pop_{S \in
        \mathcal{S},i \in \interv{1}{m(S)}}
      v_2(S,i))^\algof{SP}$. Then, $m$ is a view of $\graph_2$ for
      ${\grammar}$ such that $P \leadsto_{{\grammar}}
      m$. By Lemma \ref{fact:trunk}, we obtain $P
      \leadsto_{{\grammar}} \trunk{m}_\grammar \in
      h_\grammar(\graph_2)$. Thus, we obtain $(S,Q) \in
      h_\grammar(\graph_1) \sop^{\algof{A}_\grammar}
      h_\grammar(\graph_2)$, by an application of the
      inference rule ($\S\sop^1\P$) from Figure \ref{fig:sp-alg}.

      \noindent ``$\supseteq$'' Let $(S,Q) \in h_\grammar(\graph_1)
      \sop^{\algof{A}_\grammar} h_\grammar(\graph_2)$ and assume that
      $Q \neq \bot$ (the case $Q = \bot$ is similar). Since
      $h_\grammar(\graph_1)$ is an S-profile and
      $h_\grammar(\graph_2)$ is a P-profile, by the inference rule ($\S\sop^1\P$) from Figure \ref{fig:sp-alg}, there exist
      $(S,S_1)\in h_\grammar(\graph_1)$, a rule $S_1 \rightarrow P
      \sop Q$ in ${\grammar}$ and a reduced view $m \in
      h_\grammar(\graph_2)$ such that $P \leadsto_{{\grammar}} m$.
      Then, there exists a view $m'$ of
      $\graph_2$, such that $\trunk{m'}_\grammar = m$, and a derivation
      \begin{align*}
        P \step{{\grammar}}^* \pop_{S \in
          \mathcal{S}} \pexp{S}{m'(S)} \step{{\grammar}}^* v_2,
      \end{align*}
      for some ground term $v_2$, such that $\graph_2 =
      v_2^\algof{SP}$.  Then, we obtain $P
      \leadsto_{{\grammar}} m'$, by
      Lemma~\ref{fact:trunk}. Moreover,there exists a derivation $S
      \astep{\grammar}^* v_1 \sop S_1$, where $v_1$ is a ground term
      such that $\graph_1 = v_1^\algof{SP}$. By Lemma
      \ref{lemma:assoc-sp}, there exists a derivation $S
      \astep{{\grammar}}^* v \sop Q$ such that $v^{\algof{SP}}
      = \graph_1 \sop^{\algof{SP}} \graph_2$. Thus, we obtain $(S,Q)
      \in h_\grammar(\graph_1 \sop^\algof{SP} \graph_2)$,
      by Definition~\ref{def:sp-view} (\ref{it2:sp-view}).
    \item $\graph_1$ and $\graph_2$ are S-graphs: ``$\subseteq$'' Let
      $(S,Q) \in h_\grammar(\graph_1 \sop^{\algof{SP}}
      \graph_2)$ and assume that $Q \neq \bot$ (the case $Q = \bot$ is
      similar). By Definition~\ref{def:sp-view} (\ref{it2:sp-view}),
      there exists a derivation $S \astep{\grammar}^* v \sop Q$, where
      $v$ is a ground term such that $\graph_1 \sop^{\algof{SP}}
      \graph_2 = v^\algof{SP}$.  By Lemma \ref{lemma:assoc-sp}, there
      exist derivations $S \astep{{\grammar}}^* v_1 \sop S_1$
      and $S_1 \astep{{\grammar}}^* v_2 \sop Q$, where $S_1
      \in \mathcal{S}$ and $v_i$ are ground terms such that $\graph_i
      = v_i^{\algof{SP}}$, for $i=1,2$. Hence, we obtain $(S,S_1) \in
      h_\grammar(\graph_1)$ and $(S_1,Q) \in
      h_\grammar(\graph_2)$, by Definition
      \ref{def:sp-view} (\ref{it2:sp-view}). Thus, $(S,Q) \in
      h_\grammar(\graph_1) \sop^{\algof{A}_\grammar}
      h_\grammar(\graph_2)$, by applying the inference
      rule ($\S\sop\S$) from Figure \ref{fig:sp-alg}.

      \noindent''$\supseteq$'' Let $(S,Q) \in
      h_\grammar(\graph_1) \sop^{\algof{A}_\grammar}
      h_\grammar(\graph_2)$ and assume that $Q \neq \bot$
      (the case $Q = \bot$ is similar). Since both
      $h_\grammar(\graph_i)$ are S-profiles, by the
      inference rule ($\S\sop\S$) from Figure \ref{fig:sp-alg}, there
      exist $(S,S_1) \in h_\grammar(\graph_1)$ and
      $(S_1,Q) \in h_\grammar(\graph_2)$. By Definition
      \ref{def:sp-view} (\ref{it2:sp-view}), there exist derivations
      $S \astep{\grammar}^* v_1 \sop S_1$ and $S_1 \astep{\grammar}^*
      v_2 \sop Q$, where $v_i$ are ground terms such that $\graph_i =
      v_i^\algof{SP}$, for $i=1,2$. By Lemma \ref{lemma:assoc-sp}, we
      obtain a derivation $S \astep{\grammar}^* v \sop Q$ such that
      $v^\algof{SP} = \graph_1 \sop^\algof{SP} \graph_2$, hence $(S,Q)
      \in h_\grammar(\graph_1 \sop^\algof{SP} \graph_2)$,
      by Definition~\ref{def:sp-view} (\ref{it2:sp-view}).
  \end{itemize}
\end{proof}

Furthermore, Lemma \ref{lemma:aperiodic-tree-rec} carries
over from the class of trees, because the idempotent power
$\idemof{a}$ of each S-profile $a$ has the aperiodic property
$\idemof{a} \pop^{\algof{A}_\grammar} a = \idemof{a}$, by the
definition of $\pop^{\algof{A}_\grammar}$ for S-profiles.

The main result of this subsection is the ``if'' direction of Theorem~\ref{thm:main} for the class of series-parallel graphs (note that we
do not require the grammar to be normalized).  A similar result holds
for disoriented series-parallel graphs

\begin{theorem}\label{thm:sp-reg}
  The language of each (aperiodic) regular series-parallel grammar
  $\grammar$ is recognizable by an (aperiodic) recognizer
  $(\algof{A},B)$. Moreover, $\cardof{\universeOf{A}} \in
  2^{2^{\poly{\size{\grammar}}}}$ and
  $f^{\algof{A}}(a_1,\ldots,a_{\arityof{f}})$ can be computed in time
  $2^{\poly{\size{\grammar}}}$, for each function symbol $f \in
  \spsignature$ and elements $a_1,\ldots,a_n \in \universeOf{A}$.
\end{theorem}
\begin{proof}\emph{Proof.}
  Let $\grammar=(\mathcal{S}\uplus\mathcal{P},\rules)$ be a regular
  (aperiodic) series-parallel grammar.  Using
  Lemma~\ref{lemma:normal-form}, we obtain a normalized (aperiodic) grammar
  $\grammar'=(\mathcal{S}\uplus\mathcal{P},\rules')$.
  We point out that $\grammar$ and $\grammar'$ have the same set of
  nonterminals.
  In particular, we have
  $\alangof{}{\algof{SP}}{\grammar'} =
  \alangof{}{\algof{SP}}{\grammar}$ and $\sizeof{\grammar'} \in
  2^{\poly{\cardof{\mathcal{S}}+\cardof{\mathcal{P}}} \cdot \log(\sizeof{\grammar})}$.
  We then consider the corresponding normalized regular grammar in alternative form $\widehat{\grammar'}=(\widehat{\mathcal{S}}\uplus\mathcal{P},\widehat{\rules'})$.
  We recall that $\alangof{}{SP}{\widehat{\grammar'}} = \alangof{}{SP}{\grammar'}$ and $\size{\grammar'} = \Theta(\size{\widehat{\grammar'}})$.

  We build the recognizer $(\algof{A}_{\widehat{\grammar'}},B_{\widehat{\grammar'}})$
  based on the grammar $\widehat{\grammar'}$.
  By Lemma
  \ref{lemma:sp-view-membership}, we obtain
  $\alangof{}{\algof{SP}}{\widehat{\grammar'}} =
  h_{\widehat{\grammar'}}^{-1}(B_{\widehat{\grammar'}})$, where $B_{\widehat{\grammar'}} \isdef \set{h_{\widehat{\grammar'}}(\graph) \mid \graph \in \alangof{}{SP}{\widehat{\grammar'}}}$.
  By Lemma
  \ref{lemma:sp-view-homomorphism},
  $(\algof{A}_{\widehat{\grammar'}},B_{\widehat{\grammar'}})$ is a recognizer for
  $\alangof{}{SP}{\widehat{\grammar'}}$.
  Because of $\alangof{}{\algof{SP}}{\widehat{\grammar'}} =
  \alangof{}{\algof{SP}}{\grammar}$,
  $\algof{A}_{\widehat{\grammar'}}$ is also a
  recognizer for the language of $\grammar$.
  Moreover, $\algof{A}_{\widehat{\grammar'}}$ is
  aperiodic, if $\widehat{\grammar'}$ is aperiodic, by the counterpart of Lemma \ref{lemma:aperiodic-tree-rec} for $\algof{SP}$.
  We further recall that $\widehat{\grammar'}$ is aperiodic iff $\grammar'$ is aperiodic.
  The complexity upper
  bounds follow in the same way as in the proof of
  Theorem~\ref{thm:tree-reg}.
\end{proof}

\end{spTextEnd}

\subsubsection{An Algebra of Tree-Width $2$ Graphs}
\label{subsubsec:tw2}

We recall several graph-theoretic notions~\cite{TutteBook}. A
\emph{cutvertex} $w$ of $\graph$ is a vertex for which there exist two
distinct vertices $u,v\in\vertof{G}\setminus\set{w}$ such that there
is a path between $u$ and $v$ in $\graph$ and no path between $u$ and
$v$ in the graph obtained from $\graph$ by removing the vertex $w$ and
all edges incident to it. A \emph{block} of $\graph$ is a maximal
subgraph of $\graph$ without a cutvertex. A block is \emph{nontrivial}
if it has more than one vertex.

\begin{lemmaE}[][category=proofs]\label{lemma:nontrivial-block}
  Each connected graph with at least two vertices has only nontrivial
  blocks.
\end{lemmaE}
\begin{proofE}
  Suppose, for a contradiction, that $\graph$ is a connected graph
  with two or more vertices that has a trivial block $B$, consisting
  of one vertex $u \in \vertof{\graph}$. Then, $\graph$ has another
  vertex $v \in \vertof{\graph}\setminus\set{u}$. Because $\graph$ is
  connected, there is a path between $u$ and $v$, hence there is an
  edge $e$ attached to $u$ and some other vertex $w$ (possibly
  $w=v$). The graph with vertex set $\set{u,w}$ and edge set $\set{e}$
  has no cutvertex and subsumes $B$, contradicting the fact that $B$
  is a maximal subgraph of $\graph$ without a cutvertex.
\end{proofE}

A \emph{block tree} of $\graph$ is an undirected graph\footnote{A
  directed tree is obtained by chosing any vertex as root and
  orienting the edges from the root downwards.} whose vertices are the
blocks and the cutvertices of $\graph$, and there is an edge between a
block $B$ and a cutvertex $c$ of $\graph$ if and only if $c \in
\vertof{B}$. The proof that the block tree of a graph is indeed a tree
can be found in, e.g., \cite[Theorem III.23]{TutteBook}. The following
lemma relates connected graphs of tree-width at most~$2$ with
disoriented series-parallel graphs:

\begin{lemma}[Lemma 6.15 in \cite{CourcelleV}]\label{lemma:tw2-block}
  Let $\graph$ be a connected graph without sources. Then $\graph$ has
  tree-width at most~$2$ if and only if its blocks are either trivial
  or disoriented P-graphs.
\end{lemma}

We shall also need the following stronger statement:
\begin{lemmaE}[][category=proofs]\label{lemma:tw2-block-source}
  Let $\graph$ be a connected graph of tree-width at most~$2$, $B$ be
  a nontrivial block of $\graph$ and $x$ be a vertex of $B$. Then,
  there exists another vertex $y$ of $B$ such that $B_{(x,y)}$ (the
  graph $B$ having $1$-source $x$ and $2$-source $y$) is a disoriented
  P-graph.
\end{lemmaE}
\begin{proofE}
  Let $(T,\beta)$ be a tree decomposition of $B$ such that
  $\width{T,\beta}=\twof{B}$. The sets $\beta(n)$, $n \in \vertof{T}$
  are called \emph{bags}. If $m$ is the parent of $n$ in $T$, the set
  $\adhof{T,\beta}{n} \isdef \beta(m) \cap \beta(n)$ is the
  \emph{adhesion} of $n$ in $(T,\beta)$ (by convention,
  $\adhof{T,\beta}{n} = \emptyset$ if $n$ is the root of $T$).

  Because $B$ is connected and has at least two
  vertices, there must be an edge from $x$ to some other vertex $y$
  and there must be some node whose bag contains both $x$ and $y$. We
  can w.l.o.g. assume that this node is the root of $T$ (otherwise we
  can reorient the edges of $T$ such that this node becomes the
  root). We can further assume w.l.o.g. that the bag of the root
  contains exactly $\{x,y\}$ (otherwise we add a new root node to the
  tree, connected to the old root, whose bag contains exactly
  $\{x,y\}$). Further, w.l.o.g, we can choose $(T,\beta)$ such that
  $\beta(n) \setminus \beta(m) \neq \emptyset$, for each pair $(m,n)$
  of parent and child nodes of $T$ (if $\beta(n) \subseteq \beta(m)$
  then the edge of $T$ between $m$ and $n$ can be contracted).  Since
  $B$ has no cutvertex, it follows that
  $\cardof{\adhof{T,\beta}{n}}\geq 2$, for all non-root nodes $n \in
  \vertof{T}$.  This is a consequence of the fact that the adhesion of
  each non-root node of $T$ is a separator of $B$, see
  e.g. \cite[Lemma 11.3]{DBLP:series/txtcs/FlumG06}.  Since
  $\cardof{\beta(n)} \leq 3$, the bag of each non-root node can be
  uniquely decomposed into its adhesion $\set{u_1,u_2}$ and a third
  element $v \not\in \set{u_1,u_2}$.  Where no confusion arises, we
  shall use the names $u_1,u_2$ (for the adhesion) and $v$ (for the
  remaining vertex) throughout the proof.  Moreover, w.l.o.g, we can
  choose $(T,\beta)$ such that $\adhof{T,\beta}{n} \neq
  \adhof{T,\beta}{m}$ for each pair $(m,n)$ of parent and child of
  $T$; if $\adhof{T,\beta}{n} = \adhof{T,\beta}{m}$, then $n$ can be
  attached as a child to the parent of $m$. For every pair of vertices
  $u_1,u_2$ of $B$ we define the graph $B^{\set{u_1,u_2}}$ of sort
  $\set{1,2}$ as the subgraph of $B$ consisting only of the vertices
  $u_1$ and $u_2$ and all the edges of $B$ between these two
  vertices. For every non-root node $n \in \vertof{T}$, with
  $\adhof{T,\beta}{n}=\set{u_1,u_2}$, we denote by
  $\subtree{B}{n}^{(u_1,u_2)}$ the subgraph of $B$ of sort
  $\set{1,2}$, with $u_1$ resp. $u_2$ as the first resp. second
  source, induced by $\bigcup \set{\beta(p) \mid p \text{ is a
      descendant of } n}$ minus the edges of $B^{\set{u_1,u_2}}$.

  \begin{fact}\label{fact:tree-decomp-dsp}
    For every non-root node $n \in \vertof{T}$, with
    $\adhof{T,\beta}{n}=\set{u_1,u_2}$, we have
    $\subtree{B}{n}^{(u_1,u_2)}\in\universeOf{DSP}$.
  \end{fact}
  \begin{proof}\emph{Proof.}
    By induction on the structure of the subtree $\subtree{T}{n}$ of
    $T$ rooted at $n$, where $\set{u_1,u_2,v}\isdef\beta(n)$ and
    $\set{u_1,u_2}\isdef\adhof{T,\beta}{n}$. Let $p_1, \ldots, p_k$ be
    the children of $n$ in $T$. Then, for each $i \in \interv{1}{k}$,
    $\adhof{T,\beta}{p_i}$ is either $\set{u_1,v}$ or $\set{u_2,v}$.
    Let $P_{u_1,v} \uplus P_{v,u_2} = \set{p_1,\ldots,p_k}$ be the
    sets such that $\adhof{T,\beta}{p} = \set{u_1,v}$
    resp. $\adhof{T,\beta}{p} = \set{u_2,v}$, for all $p \in
    P_{u_1,v}$ resp. $p \in P_{v,u_2}$.  We now observe that:
    \begin{align*}
    \subtree{B}{n}^{(u_1,u_2)} = & ~\Big(
          \big(\pop^{\algof{G}}_{p \in P_{u_1,v}} \subtree{B}{p}^{(u_1,v)} ~\pop^{\algof{G}}~ B^{\set{u_1,v}} \big) ~\sop^{\algof{SP}} \\
          & ~\big(\pop^{\algof{G}}_{p \in P_{v,u_2}} \subtree{B}{p}^{(v,u_2)}  ~\pop^{\algof{G}}~ B^{\set{v,u_2}} \big) \big)\Big)
    \end{align*}
  \end{proof}

  Back to the proof, let $n$ be the root of $T$, and let $p_1, \ldots,
  p_k$ be the children of $n$ in $T$. We recall that
  $\beta(n)=\set{x,y}$.  We now observe that:
  \begin{align*}
    B_{(x,y)} = B^{\set{x,y}} ~\pop^{\algof{G}}~
  \big(\pop^{\algof{G}}_{i \in [1,k]} \subtree{B}{p_i}^{(x,y)}
  \big)
  \end{align*}
  This establishes that $B_{(x,y)} \in \universeOf{DSP}$.
  Now, we use the fact that $B$ is a block (i.e., has no cutvertex) to
  obtain that $B_{(x,y)}$ is $\circ$-atomic: Assume that $B_{(x,y)}$
  is not $\circ$-atomic.  Then, by Lemma~\ref{lemma:sp-decomp}, there
  are some graphs $\graph_1,\graph_2$ of sort $\set{1,2}$ with
  $B_{(x,y)} = \graph_1 \circ \graph_2$.  In particular, the
  $2$-source of $\graph_1$ (which is equal to the $1$-source of
  $\graph_2$) is a cutvertex of $B_{(x,y)}$.  However, this
  contradicts the definition of a block, which cannot contain a
  cutvertex.
\end{proofE}

\begin{figure}[t!]
  \vspace*{-\baselineskip}
  \centerline{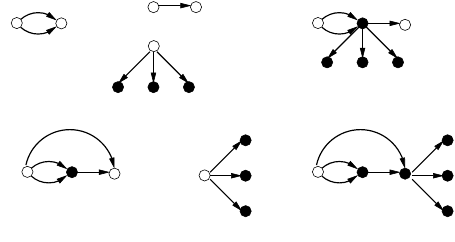}
  \vspace*{-.5\baselineskip}
  \caption{Ternary serial-composition (a) and binary
    $\rop$-composition (b) for graph of tree-width at most~$2$.}
  \label{fig:tw2}
  \vspace*{-\baselineskip}
\end{figure}

The class of connected graphs of tree-width at most~$2$ is the following two-sorted algebra:

\begin{definition}\label{def:tw2}
  The class $\algof{G}_2$ has sorts $\sorts_2 =\{ \{1\},\{1,2\}\}$ and
  signature: \begin{align*}
    \fsignature_2 \isdef
    \set{\pop, \sop, \rop, \emptygraph_{\set{1}}} \cup
    \set{\sgraph{b}_{i,3-i} \mid b \in \alphabetTwo,i=1,2},
  \end{align*}
  where we recall that we overload $\pop$, which denotes the operations $\pop_{\slabs,\slabs}$ for $\slabs \in \sorts_2 =\{ \{1\},\{1,2\}\}$.
  The function symbols from $\fsignature_2$ are interpreted
  as follows: \begin{compactitem}[-]
  \item $\sop^{\algof{G}_2}(x,y,z)$, where $\sortof{x}=\sortof{y}=\set{1,2}$ and
    $\sortof{z}=\set{1}$, joins the $2$-source of $x$ with the
    $1$-sources of $y$ and $z$; the joined vertex is not a source of
    the result, of sort $\set{1,2}$,
  \item $x \rop^{\algof{G}_2} y$, where $\sortof{x}=\set{1,2}$ and
    $\sortof{y}=\set{1}$, joins the $2$-source of $x$ with the
    $1$-source of $y$; the joined vertex is not a source of the
    result, of sort $\set{1}$.
  \item The rest of the function symbols from $\fsignature_2$ are
  interpreted same as in the graph algebra $\algof{G}$.
  \end{compactitem}

\end{definition}

Figure \ref{fig:tw2} shows the interpretation of $\sop$ (a) and $\rop$
(b) in $\algof{G}_2$.  Note that $\algof{G}_2$ is a derived algebra of
the graph algebra $\algof{G}$\footnote{ $\sop^{\algof{G}_2}(x,y,z)
\isdef \rename{(2,3)}^\algof{G}(\restrict{\set{1,3}}^\algof{G}(x
\pop^\algof{G} \rename{[2,3,1]}^\algof{G}(y) \pop^\algof{G}
\rename{(1,2)}^\algof{G}(z)))$ and $x \rop^{\algof{G}_2} y \isdef
\restrict{\set{1}}(x \pop^\algof{G} \rename{(1,2)}^\algof{G}(y))$.}.
We denote by $\universeOf{G}_2^\slabs$ the restriction of the domain
$\universeOf{G}_2$ of $\algof{G}_2$ to the elements of sort
$\slabs\in\set{\set{1},\set{1,2}}$.  The notion of disoriented P-graph
(resp. S-graph) naturally extends from $\universeOf{DSP}$ to
$\universeOf{G}^{\set{1,2}}_2$.

We define regular grammars for the $\algof{G}_2$ class.  For purposes
of presentation, we write $X$, $Y$ (resp. $\mathcal{X}$ and
$\mathcal{Y}$) for the non-terminals (resp. sets of nonterminals) of
sort $\set{1}$ and $P$, $S$ (resp. $\mathcal{P}$ and $\mathcal{S}$)
for the non-terminals (resp. sets of nonterminals) of sort
$\set{1,2}$. Likewise, we will use the template names $\X$, $\Y$, $\P$
and $\S$ for footprints. Since the sorts of nonterminals are clear
from this naming convention, we omit writing them.

\begin{definition}\label{def:regular-tw2-grammar}
  A \emph{regular} tree-width $2$ grammar is a stratified grammar
  $\grammar=(\mathcal{X}\uplus\mathcal{Y}\uplus\mathcal{P}\uplus\mathcal{S},\rules)$,
  where $\fpof{\grammar} \subseteq \afp_2$ and:
  \begin{align*}
    \afp_2 \isdef & ~\{\X \rightarrow \pexp{\Y}{\geq 0},~ \P \rightarrow \pexp{\S}{\geq 2}, ~\Y \rightarrow \P \rop \X,~ \\
    & ~~\S \rightarrow \sop(\P,\S,\X),~ \S \rightarrow \sop(\P,\P,\X),~ \rightarrow \X\} ~\cup \\
    & ~\set{\P \rightarrow \sgraph{b}_{i,3-i},~ \S \rightarrow \sgraph{b}_{i,3-i} \mid b \in \alphabetTwo,~ i=1,2}
  \end{align*}
\end{definition}

First, we establish the existence of a universal aperiodic regular
tree-width $2$ grammar. Intuitively, a complete derivation of this
grammar mimicks the block tree decomposition of a graph
(Lemma~\ref{lemma:tw2-block}).  To this end, the sources of a
disoriented series-parallel block $B$ are chosen such that the
$1$-source of $B$ is the cutvertex that is the parent of $B$ in the
block-tree. Note that such a choice is always possible (Lemma
\ref{lemma:tw2-block-source}).

\begin{lemmaE}[][category=proofs]\label{lemma:universal-synt-reg-tw2}
  Let $\grammar_{2}\isdef(\set{X,Y,P,S},\rules)$ be the regular
  tree-width $2$ grammar with the below rules, for $b \in
  \alphabetTwo$ and $i=1,2$:
  \vspace*{-\baselineskip}
  \begin{center}
    \begin{minipage}{2cm}
      \begin{align*}
        \rightarrow & ~X \text{ (\ref{it4:syntactic-regular})} \\
        X \rightarrow & ~X \pop Y \text{ (\ref{it1:syntactic-regular})} \\
        X \rightarrow & ~\emptygraph_{\set{1}} \text{ (\ref{it2:syntactic-regular})} \\
      \end{align*}
    \end{minipage}
    \begin{minipage}{3cm}
      \begin{align*}
        Y \rightarrow & ~P \rop X \text{ (\ref{it31:syntactic-regular})} \\
        S \rightarrow & ~\sop(P,S,X) \text{ (\ref{it31:syntactic-regular})} \\
        S \rightarrow & ~\sop(P,P,X) \text{ (\ref{it31:syntactic-regular})} \\
      \end{align*}
    \end{minipage}
    \begin{minipage}{3cm}
      \begin{align*}
        S \rightarrow & ~\sgraph{b}_{i,3-i} \text{ (\ref{it31:syntactic-regular})} \\
        P \rightarrow & ~\sgraph{b}_{i,3-i} \text{ (\ref{it32:syntactic-regular})} \\
        P \rightarrow & ~P \pop S \text{ (\ref{it1:syntactic-regular})} \\
        P \rightarrow & ~S \pop S \text{ (\ref{it2:syntactic-regular})}
      \end{align*}
    \end{minipage}
  \end{center}
  The sets $\alangof{}{\text{$\algof{G}_2$}}{\grammar_2}$,
  $\universeOf{G}_2^{\set{1}}$ and $\{ \graph \in \universeOf{G}
  \text{ connected} \mid \twof{\graph} \le 2, \sortof{\graph} =
  \set{1}\}$ are equal.
\end{lemmaE}
\begin{proofE}
  ``$\alangof{}{\text{$\algof{G}_2$}}{\grammar_2} \subseteq \universeOf{G}_2^{\set{1}}$''
  Clearly, all graphs resulting from derivations of the $\algof{G}_2$-grammar $\grammar_2$ belong to $\universeOf{G}_2$ --- the universe of the algebra $\algof{G}_2$.
  Further, all graphs resulting from the axiom $X$ have sort $\set{1}$ and hence belong to  $\universeOf{G}_2^{\set{1}}$.

  \noindent``$\universeOf{G}_2^{\set{1}} \subseteq \{ \graph \text{ connected} \mid \twof{\graph} \le 2, \sortof{\graph} = \set{1}\}$'' Each graph $\graph \in
  \alangof{}{\text{$\algof{G}_2$}}{\grammar_2}$ has tree-width at most
  $2$, because the grammar $\grammar_2$ only uses operations from the
  derived algebra $\algof{G}_2$, whose function symbols are interpreted
  using operations from $\algof{G}$ that involve at most $3$
  sources. Moreover, all these operations preserve connectivity and
  each $X$-derivation yields terms of sort $\set{1}$.

  \noindent``$\{ \graph \text{ connected} \mid \twof{\graph} \le 2,
  \sortof{\graph} = \set{1}\} \subseteq
  \alangof{}{\text{$\algof{G}_2$}}{\grammar_2}$'' Let $\graph$ be a
  connected graph of tree-width at most~$2$ of sort $\set{1}$, and let
  $T$ be its block tree.  In case the $1$-source is a cutvertex, we
  can assume without loss of generality (by rotating the block tree),
  that the $1$-source is the root of the block tree.  In case the
  $1$-source is not a cutvertex, we can assume without loss of
  generality (by rotating the block tree), that the block containing
  the $1$-source is the root of the block tree; we then add the
  $1$-source as the new root to the block tree.  We introduce the
  following notation: For every node of the block tree consisting of a
  single vertex $v\in\vertof{\graph}$ (that is $v$ is either a
  cutvertex or the $1$-source of $\graph$), we denote by $\graph_v$
  the subgraph of $\graph$ that is induced by $v$ and its descendants
  in the block tree, where $\graph_v$ is considered a as graph of sort
  $\set{1}$ with $1$-source $v$.  Note that we have $\graph_v =
  \graph$ for the $1$-source $v$ of $\graph$, because $v$ is the root
  of the block tree.

  We will now prove that $\graph_v \in
  \alangof{X}{\text{$\algof{G}_2$}}{\grammar_2}$
  for every node of the block tree consisting of a
  single vertex $v\in\vertof{\graph}$.
  Recall that this is sufficient to prove the claim because we also consider $\graph_v = \graph$ for the $1$-source $v$ of $\graph$.
  If $\graph_v$ has one vertex and no edges, we obtain $\graph_v \in  \alangof{X}{\text{$\algof{G}_2$}}{\grammar_2}$ by applying the rule $X \rightarrow \emptygraph_{\set{1}}$ from $\grammar_2$.
  Else, by
  Lemma \ref{lemma:nontrivial-block}, each block of $\graph_v$ is
  nontrivial.
  Let $B_1,\ldots,B_\ell$ be the blocks that are the children of $v$ in the block tree.
  For each $j \in \interv{1}{\ell}$, let $c_1^j,
  \ldots, c_{n_j}^j$ be the vertices that are the children of $B_j$ in $T$.
  By the inductive hypothesis, there exist complete derivations $X \step{\grammar_2}^* \gamma_{i,j}$ such that $\gamma_{i,j}^{\algof{G}_2} = \graph_{c_i^j}$, where each $\graph_{c_i^j}$ is of sort $\set{1}$, having $c_i^j$ as its
  $1$-source.
  By Lemma~\ref{lemma:tw2-block-source}, for each $j \in
  \interv{1}{\ell}$ there is some vertex $w_j$ of $B_j$, such that $(B_j)_{(v,w_j)}$, the graph $B_j$ considered as graph of sort $\set{1,2}$, with $v$ (resp. $w_j$) as its $1$-source
  (resp. $2$-source), is a $\sop$-atomic disoriented series-parallel
  graph.
  By Lemma~\ref{lemma:universal-synt-reg-sp} (naturally
  extended to disoriented series-parallel graphs) there is a derivation $P \step{\grammar_\algof{DSP} }^* \beta_j$ of the universal series-parallel grammar   $\grammar_\algof{DSP}$ that witnesses $(B_j)_{(v,w_j)} \in  \alangof{}{\algof{DSP}}{\grammar_\algof{DSP}}$ (we recall that $\grammar_\algof{DSP}$ contains, in addition to the rules of $\grammar_\algof{SP}$, the rules $S \rightarrow \sgraph{b}_{2,1}$
  and $P \rightarrow \sgraph{b}_{2,1}$, for each $b \in \alphabetTwo$).

  We now take the derivations $P \step{\grammar_\algof{DSP} }^* \beta_j$, for every $1 \le j \le \ell$ and build modified derivations $P \step{\grammar_2}^* \theta_j$ such that $\delta_j^{\algof{G}_2}$ is the graph $(B_j)_{(v,w_j)}$ to which we attach all the graphs $\graph_{c_i^j}$ at the respective cutvertices $c_i^j$, for every $i$ such that $c_i^j \neq w_j$.
  We obtain the derivation $P \step{\grammar_2}^* \theta_j$ from  $P \step{\grammar_\algof{DSP} }^* \beta_j$ as follows:
    \begin{compactitem}[-]
    \item each step $S \step{\grammar_\algof{DSP}} P \sop S$ is replaced by the derivation:
      \begin{compactitem}[*]
      \item $S \step{\grammar_2} \sop(P,S,X)
        \step{\grammar_2}^* \sop(P,S,\gamma_{i,j})$, if the join vertex of the serial composition is the cutpoint $c_i^j$, for some $i \in \interv{1}{n_j}$, or
      \item $S \step{\grammar_2} \sop(P,S,X) \step{\grammar_2}
        \sop(P,S,\emptygraph_{\set{1}})$, otherwise.
      \end{compactitem}
    \item each step $S \step{\grammar_\algof{DSP}} P \sop P$ is
      replaced by the derivation: \begin{compactitem}[*]
      \item $S \step{\grammar_2} \sop(P,P,X) \step{\grammar}^*
        \sop(P,P,\gamma_{i,j})$, if the join vertex of the serial
        composition is the cutpoint $c_i^j$, for some $i \in
        \interv{1}{n_j}$, or
      \item $S \step{\grammar_2} \sop(P,P,X) \step{\grammar_2}
        \sop(P,P,\emptygraph_{\set{1}})$, otherwise.
      \end{compactitem}
    \end{compactitem}
  For each $j \in \interv{1}{\ell}$, we build the complete derivation
  \begin{align*}
    Y \step{\grammar_2} P \rop X \step{\grammar_2}^* \theta_j \rop \delta_j,
  \end{align*}
  where $\delta_j = \emptygraph_{\set{1}}$ in case that $c_i^j \neq w_j$ for all $i$, and $\delta_j = \gamma_{i,j}$ in case that $c_i^j = w_j$ for some $i$.
  Since $B_1, \ldots, B_\ell$ are the children of $v$
  in the block tree, each graph $\graph_j$ consisting of $B_j$ and all its descendants in the block tree is a graph of sort $\set{1}$
  with $v$ as $1$-source such that, moreover $\graph = \graph_1
  \pop^{\algof{G}_2} \ldots \pop^{\algof{G}_2} \graph_\ell =
  (\theta_1 \rop \delta_1 \pop \ldots \pop \theta_\ell \rop \delta_\ell)^{\algof{G}_2}$.
  We consider the complete derivation
  \begin{align*}
    X \step{\grammar_2}^* X \pop \pexp{Y}{\ell} \step{\grammar_2}
    \emptygraph_{\set{1}} \pop \pexp{Y}{\ell} \step{\grammar_2}^* \emptygraph_{\set{1}} \pop (\pop_{j\in\interv{1}{\ell}} \theta_j \rop \delta_j),
  \end{align*}
  which proves $\graph \in \alangof{X}{\text{$\algof{G}_2$}}{\grammar_2}$.
\end{proofE}

Note that $\grammar_2$ is an aperiodic regular tree-width $2$ grammar,
because $\fpof{\grammar_2} \subseteq \afp_2$ and the two rules of the
form (\ref{it1:syntactic-regular}) have period $1$. As in the case of
trees, the result of Lemma \ref{lemma:universal-synt-reg-tw2} proves
the ``only if'' direction of Theorem~\ref{thm:main} for the
$\algof{G}_2$ class (Corollary \ref{cor:completeness}). The ``if''
direction is dealt with by the construction of a recognizer for the
language of a given regular tree-width $2$ grammar (Theorem~\ref{thm:tw2-reg}). Just as before, the size of the recognizer is
doubly-exponential in the size of the input regular grammar.

\begin{figure}[t!]
  \centerline{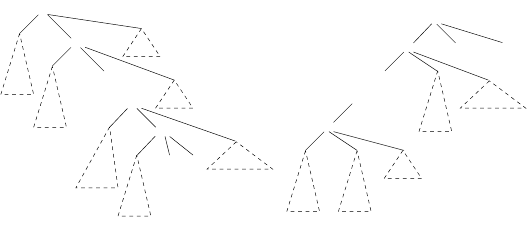}
  \vspace*{-\baselineskip}
  \caption{A partial derivation of
    an S-graph $\graph = \sop(\ldots\sop(\sop(v_1,v_2,t_1),v_3,t_2),\ldots,v_n,t_{n-1})^{\algof{G}_2}$ (a) A
    view of $\graph$ corresponding to that partial derivation
    (b)}
  \vspace*{-\baselineskip}
  \label{fig:tw2-view}
\end{figure}

\newcommand{\twalg}{\text{$\algof{G}_2$}}

The definition of views for graphs of tree-width at most~$2$
(Definition \ref{def:tw2-view}) relies on the following
associative-like re-arrangement of serial composition terms.
Formally, we consider an equivalence relation between
$\fsignature_2$-terms, such that $u \assoc w$ if and only if $w$ can
be obtained from $u$ (or viceversa) by zero or more applications of
the following axiom:
\begin{align}\label{eq:assoc-tw2}
\sop(x_1,\sop(x_2,y,z_2),z_1) = \sop(\sop(x_1,x_2,z_1),y,z_2)
\end{align}
It is easy to see that both sides of the (\ref{eq:assoc-tw2}) axiom
evaluate to the same graph, for any substitution of $x_i,z_i,y$, for
$i=1,2$, with ground terms. The $\assoc$ relation allows to regroup
the ground $\fsignature_2$-terms which evaluate to the P-components of
an S-graph into a single ground $\fsignature_2$-term $u$ that
evaluates to the serial composition of its P-components
(Figure~\ref{fig:tw2-view}). We denote by $u \astep{\grammar}^* w$ the
existence of a derivation $u \step{\grammar}^* v$ such that $v \assoc
w$.

\begin{proofsTextEnd}
\begin{lemma}\label{lemma:tw2-assoc-start}
  Let $S \in \mathcal{S}$, $Q \in \mathcal{P} \uplus \mathcal{S}$ and $X \in \mathcal{X}$ be nonterminals.
  Then, for every derivation $S \astep{\grammar}^* \sop(u,Q,X)$, there are non-terminals $P_1, \ldots, P_n$ and $X_1, \ldots, X_{n-1}$, derivations $P_i \step{\grammar}^* v_i$ and $X_j
  \step{\grammar}^* t_j$, and rules $S_i \rightarrow \sop(P_i,S_{i+1},X_i)$, for $1 \le i < n-1$, with $S=S_1$, and a rule $S_{n-1} \rightarrow \sop(P_n,Q,X)$ such that
  $u \assoc \sop(\ldots\sop(v_1,v_2,t_1)\ldots,v_n,t_{n-1})$.
\end{lemma}
\begin{proof}\emph{Proof.}
  We consider the derivation $S \step{\grammar}^* t[Q,X]$, where $t[Q,X]$ is a term such that $t[Q,X] \assoc \sop(u,Q,X)$.
  We note that the only rules of a regular tree-width 2 grammar that generate series compositions are the rules of shape $S_1 \rightarrow \sop(P,S_2,X)$ and $S \rightarrow \sop(P_1,P_2,X)$.
  Further, because $Q$ occurs at the right-most position of a series-composition of terms, we must have (up to a reordering of the rule applications) that the derivation $S \step{\grammar}^* t[Q,X]$ has the form
  \begin{align*}
    S_1 \step{\grammar} & ~\sop(P_1,S_2,X_1) \step{\grammar} \sop(P_1,\sop(P_2,S_3,X_2),X_1) \step{\grammar} \\
    & \ldots \\
    \step{\grammar} & ~\sop(P_1,\sop(P_2,\sop(\ldots,\sop(P_{n},Q,X),\ldots),X_2),X_1) \\
    \step{\grammar}^* & ~\sop(v_1,\sop(v_2,\sop(\ldots,\sop(v_{n},Q,X),\ldots),t_2),t_1)
  \end{align*}
  where $S=S_1$, and $P_i \step{\grammar}^* v_i$, $X_j
  \step{\grammar}^* t_j$ are derivations, for all $i \in
  \interv{1}{n}$ and $j \in \interv{1}{n-1}$.
  We obtain $t[Q,X] =
  \sop(v_1,\sop(v_2,\sop(\ldots,\sop(v_{n},Q,X),\ldots),t_2),t_1)
  \assoc \sop(u,Q,X)$.
  Hence, $u \assoc \sop(\ldots\sop(v_1,v_2,t_1)\ldots,v_n,t_{n-1})$, see Figure~\ref{fig:tw2-view} (b).
\end{proof}

\begin{lemma}\label{lemma:tw2-assoc}
  Let $\grammar=(\mathcal{X} \uplus \mathcal{Y} \uplus \mathcal{P} \uplus \mathcal{S}, \rules)$ be a regular tree-width $2$ grammar, $S\in\mathcal{S}$ and $Q \in
  \mathcal{P} \uplus \mathcal{S}$ be nonterminals, $\graph \in \universeOf{G}^{\set{1,2}}_2$ be an S-graph and $K
  \in \universeOf{G}^{\set{1}}_2$ be a graph.
  Then, for all graphs $\graph_1,\graph_2\in\universeOf{G}^{\set{1,2}}_2$ and
  $H\in\universeOf{G}^{\set{1}}_2$ such that $\graph =
  \sop^{\algof{G}_2}(\graph_1,\graph_2,H)$, the following are
  equivalent:\begin{compactenum}[1.]
  \item\label{it1:lemma:tw2-assoc} there exists a derivation $S
    \astep{\grammar}^* \sop(u,Q,X)$, where $u$ is a ground
    $\fsignature_2$-term such that $u^{\algof{G}_2} = \graph$,
  \item\label{it2:lemma:tw2-assoc} there exists derivations $S
    \astep{\grammar}^* \sop(u_1,S',w)$ and $S' \astep{\grammar}^*
    \sop(u_2,Q,X)$, for some nonterminal $S' \in
    \mathcal{Y}_{\set{1,2}}$ and ground $\fsignature_2$-terms $u_1,
    u_2$ and $w$ such that $u_i^{\algof{G}_2} = \graph_i$, for
    $i=1,2$ and $w^{\algof{G}_2} = H$.
  \end{compactenum}
\end{lemma}
\begin{proof}\emph{Proof.}
  ``(\ref{it1:lemma:tw2-assoc}) $\Rightarrow$
  (\ref{it2:lemma:tw2-assoc})'' We assume that there exists a
  derivation $S \astep{\grammar}^* \sop(u,Q,X)$, where $u$ is a ground
  $\spsignature$-term such that $u^{\algof{G}_2} =\graph$.  By
  Lemma~\ref{lemma:tw2-assoc-start}, there are non-terminals $P_1,
  \ldots, P_n$ and $X_1, \ldots, X_{n-1}$, derivations $P_i
  \step{\grammar}^* v_i$ and $X_j \step{\grammar}^* t_j$, and rules
  $S_i \rightarrow \sop(P_i,S_{i+1},X_i)$, for $1 \le i < n-1$, with
  $S=S_1$, and a rule $S_{n-1} \rightarrow \sop(P_n,Q,X)$ such that $u
  \assoc \sop(\ldots\sop(v_1,v_2,t_1)\ldots,v_n,t_{n-1})$. Since
  $u^{\algof{G}_2} = \sop^{\algof{G}_2}(\graph_1,\graph_2,H)$, we must
  have: \begin{align*}
    \graph_1 = & ~(\sop(\ldots\sop(v_1,v_2,t_1)\ldots,v_{i},t_{i-1}))^{\algof{G}_2} \\
    \graph_2 = & ~(\sop(\ldots\sop(v_{i+1},v_{i+2},t_{i+1})\ldots,v_{n},t_{n-1}))^{\algof{G}_2} \\
    H = & ~t_i^{\algof{G}_2} \text{, for some } i \in \interv{1}{n-1}
  \end{align*}
  We define: \begin{align*}
    u_1 \isdef & ~\sop(\ldots\sop(v_1,v_2,t_1)\ldots,v_{i},t_{i-1}) \\
    u_2 \isdef & ~\sop(\ldots\sop(v_{i+1},v_{i+2},t_{i+1})\ldots,v_{n},t_{n-1}) \\
    w_1 \isdef & ~t_i
  \end{align*}
  The existence of the derivations $S \astep{\grammar}^*
  \sop(u_1,S_i,w_1)$ and $S_i \astep{\grammar}^* \sop(u_2,Q,X)$
  follows from the decomposition of the above derivation up to and
  from $S_i$ onwards, modulo re-orderings via associativity of the
  subderivations.

  \vspace*{.5\baselineskip}
  \noindent``(\ref{it2:lemma:tw2-assoc}) $\Rightarrow$
  (\ref{it1:lemma:tw2-assoc})''
  We assume that there exist derivations $S \astep{\grammar}^* \sop(u_1,S',w)$ and $S' \astep{\grammar}^* \sop(u_2,Q,X)$, where $u_i$ are ground $\spsignature$-terms such that $u_i^{\algof{G}_2} =\graph_i$ and $w^{\algof{G}_2}  = H$.
  By Lemma~\ref{lemma:tw2-assoc-start}, there are non-terminals $P_1, \ldots, P_n$ and $X_1, \ldots, X_{n-1}$, derivations $P_i \step{\grammar}^* v_i$ and $X_j
  \step{\grammar}^* t_j$, and rules $S_i \rightarrow \sop(P_i,S_{i+1},X_i)$, for $1 \le i < n-1$, with $S=S_1$, $S' = S_k$, and $w = t_k$, and a rule $S_{n-1} \rightarrow \sop(P_n,Q,X)$ such that $u_1 \assoc
  \sop(\ldots\sop(v_1,v_2,t_1)\ldots,v_i,t_{i-1})$ and $u_2 \assoc
  \sop(\ldots\sop(v_{i+1},v_{i+2},Q)\ldots,v_n,t_{n-1})$.
  Hence, we obtain $\sop(u_1,u_2,w_1) \assoc
  \sop(v_1,\sop(v_2,\sop(\ldots,\sop(v_{n},Q,X),\ldots),t_2),t_1)$,
  which proves $S \astep{\grammar}^* \sop(\sop(u_1,u_2,w),Q,X)$.
  Since $\graph = \sop^{\algof{G}_2}(\graph_1,\graph_2,H) =
  \sop^{\algof{G}_2}(u_1^{\algof{G}_2},u_2^{\algof{G}_2},w^{\algof{G}_2})$,
  we define $u \isdef \sop(u_1,u_2,w)$, thus obtaining $S
  \astep{\grammar}^* \sop(u,Q,X)$.
\end{proof}
\end{proofsTextEnd}

In order to simplify the development of the recognizer we consider
regular tree-width 2 grammars in the following alternative form:

\begin{definition}
\label{def:regular-tw2-grammar-alternative}
  A regular tree-width 2 grammar in \emph{alternative} form is a
  stratified grammar
  $\grammar=(\mathcal{X}\uplus\mathcal{Y}\uplus\mathcal{P}\uplus\mathcal{S},\rules)$
  such that $\fpof{\grammar} \subseteq \afp'_2$, where:
  \begin{align*}
    \afp'_2 \isdef & ~\{\X \rightarrow \pexp{\Y}{\geq 0},~ \P \rightarrow \pexp{\S}{\geq 1}, ~\Y \rightarrow \P \rop \X, \\
    & ~~\S \rightarrow \sop(\P,\S,\X), ~\S \rightarrow \sop(\P,\P,\X),~ \rightarrow \X\} ~\cup \\
    & ~\set{\S \rightarrow \sgraph{b}_{i,3-i} \mid b \in \alphabetTwo,~ i=1,2}
  \end{align*}
  such that for every rule $P \rightarrow S \in \rules$ there is some
  rule $S \rightarrow \sgraph{b}_{i,3-i} \in \rules$ and $S$ does
  not occur anywhere else in $\rules$.
\end{definition}
Each regular grammar $\grammar$ can be transformed into a regular
grammar in alternative form
$\widehat{\grammar}=(\mathcal{X}\uplus\mathcal{Y}\uplus\mathcal{P}\uplus\widehat{\mathcal{S}},\widehat{\rules})$,
by setting
$\widehat{\mathcal{S}}\isdef\mathcal{S}\uplus\set{S_{b_{i,3-i}}\mid
  b\in\alphabetTwo,~ i=1,2}$, for some fresh nonterminals
$S_{b_{i,3-i}} \not\in \mathcal{S}$, and replacing each rule $P
\rightarrow \sgraph{b}_{i,3-i} \in \rules$ by the rules $P \rightarrow
S_{b_{i,3-i}}, S_{b_{i,3-i}} \rightarrow \sgraph{b}_{i,3-i} \in
\widehat{\rules}$.  It is straightforward that
$\alangof{}{SP}{\widehat{\grammar}} =
\alangof{}{SP}{\grammar}$\footnote{In every derivation one can replace
the derivation steps $P \rightarrow S_{b_{i,3-i}}, S_{b_{i,3-i}}
\rightarrow \sgraph{b}_{i,3-i}$ with $P \rightarrow
\sgraph{b}_{i,3-i}$, and viceversa.} and that $\widehat{\grammar}$ is
aperiodic iff $\grammar$ is aperiodic.  Moreover, the size of the
grammar is preserved, i.e., $\size{\grammar} =
\Theta(\size{\widehat{\grammar}})$.

Intuitively, the profile for a graph $\graph \in \universeOf{G}_2$
maintains sufficient information about $\graph$ in order to determine
how $\graph$ can be parsed during some derivation; the cases for
$\graph \in \universeOf{G}^{\set{1}}_2$ resp. $\graph \in
\universeOf{G}^{\set{1,2}}_2$ build on the ideas developed for trees
and series-parallel graphs \ifLongVersion\else(Appendix
\ref{app:sp})\fi. The views of the graphs from $\universeOf{G}_2$ for
a regular tree-width $2$ grammar are defined below:

\begin{definition}\label{def:tw2-view}
  Let $\graph$ be a graph of tree-width $\le 2$ and
  $\grammar=(\mathcal{X}\uplus\mathcal{Y} \uplus \mathcal{P} \uplus
  \mathcal{S},\rules)$ be a normalized regular tree-width $2$ grammar
  in alternative form. We distinguish the cases:
  \begin{compactenum}[1.]
  \item\label{it1:tw2-view} $\graph \in \universeOf{G}^{\set{1}}_2$: a
    \emph{view of $\graph$ for $\grammar$} is a multiset
    $\mset{Y_1,\ldots, Y_n} \in \mpow{\mathcal{Y}}$ for which there
    exist complete derivations $Y_i \step{\grammar}^* v_i$, for all $i
    \in \interv{1}{n}$, such that $\graph = (v_1 \pop \ldots \pop
    v_n)^{\algof{G}_2}$. A \emph{reduced} view of $\graph$ is a
    multiset $\trunk{m}_\grammar$, where $m$ is a view of $\graph$ for
    $\grammar$.
  \item\label{it21:tw2-view} $\graph \in \universeOf{G}^{\set{1,2}}_2$
    is a P-graph: a \emph{view of $\graph$ for $\grammar$} is a
    multiset $\mset{S_1,\ldots,S_n} \in \mpow{\mathcal{S}}$ for which
    there exist complete derivations $S_i \step{\grammar}^* v_i$, for
    all $i \in \interv{1}{n}$, such that $\graph = (v_1 \pop \ldots
    \pop v_n)^{\algof{G}_2}$. The reduced views are defined as before
    (\ref{it1:tw2-view}).
  \item\label{it22:tw2-view} $\graph \in \universeOf{G}^{\set{1,2}}_2$
    is an S-graph: a \emph{view of $\graph$ for $\grammar$}
    is either: \begin{compactenum}
    \item\label{it221:tw2-view} a triple $(S,Q,X) \in
      \mathcal{S} \times (\mathcal{P} \uplus
      \mathcal{S}) \times \mathcal{X}$ for which
      there exists a derivation $S \astep{\grammar}^* \sop(u,Q,X)$,
      for some ground $\fsignature_2$-term $u$ such that
      $u^{\algof{G}_2} = \graph$, or
    \item\label{it222:tw2-view} a pair $(S,\bot) \in \mathcal{S}
      \times \set{\bot}$ for which there exists a derivation $S
      \step{\grammar}^* u$, for some ground $\fsignature_2$-term $u$
      such that $u^{\algof{G}_2} = \graph$.
    \end{compactenum}
    The views of S-graphs are reduced, by definition.
  \end{compactenum}
  The \emph{profile} of $\graph$ is the set of reduced views, denoted
  $h_\grammar(\graph)$.
\end{definition}

We define the finite algebra $\algof{A}_\grammar$ having the sorted
domain $\universeOf{A}^{\slabs} \isdef \set{h_\grammar(\graph) \mid
  \graph \in \universeOf{G}_2, \sortof{\graph}=\slabs}$, for $\slabs
\in \{\set{1},\set{1,2}\}$. A P-profile (resp. S-profile) is the
profile of a P-graph (resp. S-graph) of sort $\set{1,2}$. The function
symbols from $\fsignature_2$ are interpreted in $\algof{A}_\grammar$
by the least sets built according to the inference rules from
Figure~\ref{fig:tw2-alg}. The names of the inference rules distinguish
between P-profile and S-profiles, e.g., $\P \pop_{\set{1,2}} \S$ defines
$\pop^\twalg$ when one argument is a P-profile and the other is an
S-profile, $\sop^1(\P,\S)$ and $\sop^2(\P,\S)$ define $\sop^\twalg$
when the first and second arguments are a P- and an S-profile,
respectively. While the inference rules for parallel composition
resemble the ($\pop$) rule for trees (Figure \ref{fig:tree-alg}), the
inference rules for serial composition rely on finding grammar rules
from $\grammar$ that match either the multiset arguments using the
$\leadsto_\grammar$ relation (for P-profile arguments) or the
nonterminals from the tuples (for S-profile arguments).

The size of $a \in \universeOf{A}_\grammar$ is $\sizeof{a} \isdef
\sum_{m \in a} \cardof{m}$ if $\sortof{a} = \set{1}$ or $\sortof{a} =
\set{1,2}$ and $a$ is a P-profile, or $\sum_{s \in a} \lenof{s}$ if
$a$ is an S-profile.


\begin{figure}[t!]
  \vspace*{-\baselineskip}
{\small\begin{center}
  \begin{minipage}{4cm}
    \begin{prooftree}
      \AxiomC{$\begin{array}{c}
          \\
        \end{array}$}
      \RightLabel{$\mathtt{vertex}$}
      \UnaryInfC{$\emptymset \in \emptygraph_{\set{1}}^{\algof{A}_\grammar}$}
    \end{prooftree}
  \end{minipage}
  \begin{minipage}{4cm}
    \begin{prooftree}
      \AxiomC{$k_i \in a_i,~ i=1,2$}
      \RightLabel{$\pop_{\set{1}}$}
      \UnaryInfC{$\trunk{k_1 + k_2}_\grammar \in a_1 \pop^{\algof{A}_\grammar} a_2$}
    \end{prooftree}
  \end{minipage}

  \vspace*{.5\baselineskip}
  \begin{minipage}{4cm}
    \begin{prooftree}
      \AxiomC{$S \rightarrow \sgraph{b}_{i,3-i} \in \rules$}
      \RightLabel{$\mathtt{edge}$}
      \UnaryInfC{$\mset{S} \in \sgraph{b}^{\algof{A}_\grammar}$}
    \end{prooftree}
  \end{minipage}

  \vspace*{.5\baselineskip}
  \begin{minipage}{7.1cm}
    \begin{prooftree}
      \AxiomC{$\begin{array}{c}
          P \leadsto_{\grammar} m \in a_1,~ X \leadsto_{\grammar} k \in a_2 \\
          Y \rightarrow P \rop X \in \rules
          \end{array}$}
      \RightLabel{$\rop$}
      \UnaryInfC{$\mset{Y} \in a_1 \rop^{\algof{A}_\grammar} a_2$}
    \end{prooftree}
  \end{minipage}

  \vspace*{.5\baselineskip}
  \begin{minipage}{7cm}
    \begin{prooftree}
      \AxiomC{$\begin{array}{c} \\
          m_i \in a_i,~ i=1,2
        \end{array}$}
      \RightLabel{$\P \pop_{\set{1,2}} \P$}
      \UnaryInfC{$\trunk{m_1 + m_2}_\grammar \in a_1 \pop^{\algof{A}_\grammar} a_2$}
    \end{prooftree}
  \end{minipage}

  \vspace*{.5\baselineskip}
  \begin{minipage}{7cm}
    \begin{prooftree}
      \AxiomC{$m_1 \in a_1,~ (S,\bot) \in a_2$}
      \RightLabel{$\P \pop_{\set{1,2}} \S$}
      \UnaryInfC{$\trunk{m_1 + \mset{S}}_\grammar \in a_1 \pop^{\algof{A}_\grammar} a_2$}
    \end{prooftree}
  \end{minipage}

  \vspace*{.5\baselineskip}
  \begin{minipage}{7cm}
    \begin{prooftree}
      \AxiomC{$\set{(S_i,\bot) \in a_i}_{i=1,2}$}
      \RightLabel{$\S \pop_{\set{1,2}} \S$}
      \UnaryInfC{$\trunk{\mset{S_1} + \mset{S_2}}_\grammar \in a_1 \pop^{\algof{A}_\grammar} a_2$}
    \end{prooftree}
  \end{minipage}

  \vspace*{.5\baselineskip}
  \begin{minipage}{7.1cm}
    \begin{prooftree}
      \AxiomC{$\begin{array}{c}
          \set{P_i \leadsto_{\grammar} m_i \in a_i}_{i=1,2},~ X_1 \leadsto_{\grammar} k \in a_3 \\
          S \rightarrow \sop(P_1,S_1,X_1) \in \rules \\
          S_1 \rightarrow \sop(P_2,Q,X) \in \rules
        \end{array}$}
      \RightLabel{$\sop^1(\P,\P)$}
      \UnaryInfC{$(S,Q,X) \in \sop^{\algof{A}_\grammar}(a_1,a_2,a_3)$}
    \end{prooftree}
  \end{minipage}

  \begin{minipage}{7.1cm}
    \begin{prooftree}
      \AxiomC{$\begin{array}{c}
          \\
          \set{P_i \leadsto_{\grammar} m_i \in a_i}_{i=1,2},~
          X \leadsto_{\grammar} k \in a_3 \\
          S \rightarrow \sop(P_1,P_2,X) \in \rules
        \end{array}$}
      \RightLabel{$\sop^2(\P,\P)$}
      \UnaryInfC{$(S,\bot) \in \sop^{\algof{A}_\grammar}(a_1,a_2,a_3)$}
    \end{prooftree}
  \end{minipage}

  \vspace*{.5\baselineskip}
  \begin{minipage}{7cm}
    \begin{prooftree}
      \AxiomC{$\begin{array}{c}
          P \leadsto_{\grammar} m \in a_1, (S_1,Q,X) \in a_2 \\
          X_1 \leadsto_{\grammar} k \in a_3 \\
          S \rightarrow \sop(P,S_1,X_1) \in \rules
        \end{array}$}
      \RightLabel{$\sop^1(\P,\S)$}
      \UnaryInfC{$(S,Q,X) \in \sop^{\algof{A}_\grammar}(a_1,a_2,a_3)$}
    \end{prooftree}
  \end{minipage}

  \vspace*{.5\baselineskip}
  \begin{minipage}{7cm}
    \begin{prooftree}
      \AxiomC{$\begin{array}{c}
          P \leadsto_{\grammar} m \in a_1, (S_1,\bot) \in a_2 \\
          X_1 \leadsto_{\grammar} k \in a_3 \\
          S \rightarrow \sop(P,S_1,X_1) \in \rules
        \end{array}$}
      \RightLabel{$\sop^2(\P,\S)$}
      \UnaryInfC{$(S,\bot) \in \sop^{\algof{A}_\grammar}(a_1,a_2,a_3)$}
    \end{prooftree}
  \end{minipage}

  \vspace*{.5\baselineskip}
  \begin{minipage}{7cm}
    \begin{prooftree}
      \AxiomC{$\begin{array}{c}
          (S,S_1,X_1) \in a_1,~ P \leadsto_{\grammar} m \in a_2 \\
          X_1 \leadsto_\grammar k \in a_3 \\
          S_1 \rightarrow \sop(P,Q,X) \in \rules
        \end{array}$}
      \RightLabel{$\sop^1(\S,\P)$}
      \UnaryInfC{$(S,Q,X) \in \sop^{\algof{A}_\grammar}(a_1,a_2,a_3)$}
    \end{prooftree}
  \end{minipage}

  \begin{minipage}{7cm}
    \begin{prooftree}
      \AxiomC{$\begin{array}{c}
          \\
          (S,P,X_1) \in a_1,~ P \leadsto_{\grammar} m \in a_2 \\
          X_1 \leadsto_\grammar k \in a_3
        \end{array}$}
      \RightLabel{$\sop^2(\S,\P)$}
      \UnaryInfC{$(S,\bot) \in \sop^{\algof{A}_\grammar}(a_1,a_2,a_3)$}
    \end{prooftree}
  \end{minipage}

  \vspace*{.5\baselineskip}
  \begin{minipage}{7cm}
    \begin{prooftree}
      \AxiomC{$\begin{array}{c}
          (S,S_1,X_1) \in a_1,~ (S_1,Q,X) \in a_2 \\
          X_1 \leadsto_\grammar k \in a_3
        \end{array}$}
      \RightLabel{$\sop^1(\S,\S)$}
      \UnaryInfC{$(S,Q,X) \in \sop^{\algof{A}_\grammar}(a_1,a_2,a_3)$}
    \end{prooftree}
  \end{minipage}

  \vspace*{.5\baselineskip}
  \begin{minipage}{7cm}
    \begin{prooftree}
      \AxiomC{$\begin{array}{c}
          (S,S_1,X_1) \in a_1,~ (S_1,\bot) \in a_2 \\
          X_1 \leadsto_\grammar k \in a_3
        \end{array}$}
      \RightLabel{$\sop^1(\S,\S)$}
      \UnaryInfC{$(S,\bot) \in \sop^{\algof{A}_\grammar}(a_1,a_2,a_3)$}
    \end{prooftree}
  \end{minipage}
\end{center}}
\caption{The interpretation of the signature $\fsignature_2$ in
  $\algof{A}_\grammar$}\label{fig:tw2-alg}
\vspace*{-1.5\baselineskip}
\end{figure}


\begin{lemmaE}[][category=proofs]\label{lemma:size-of-tw2-rec}
  Let $\grammar = (\mathcal{X}\uplus\mathcal{Y} \uplus \mathcal{P} \uplus \mathcal{S}, \rules)$ be a normalized regular tree-width $2$ grammar in alternative form.
  Then, $\cardof{\universeOf{A}_\grammar} \in  2^{2^{\poly{\size{\grammar}}}}$ and $\sizeof{a} \in
  2^{\poly{\size{\grammar}}}$ for each $a \in
  \algof{A}_\grammar$.
  Moreover, $f^{\algof{A}_\grammar}(a_1,\ldots,a_{\arityof{f}})$ can be
  computed in time $2^{\poly{\size{\grammar}}}$, for each function symbol $f \in
  \fsignature_2$ and all elements $a_1,\ldots,a_{\arityof{f}} \in \universeOf{A}_\grammar$.
\end{lemmaE}
\begin{proofE}
  As in the proof of Lemma~\ref{lemma:size-of-tree-rec}, we have at most  ${\sizeof{\grammar}}^{\poly{\cardof{\mathcal{Y}}}}$
  reduced multisets in a profile of sort $\set{1}$ and at most
  ${\sizeof{\grammar}}^{\poly{\cardof{\mathcal{S}}}}$
  multisets in a P-profile of sort $\set{1,2}$.
  Moreover, there are at most $\cardof{\mathcal{S}} \cdot
  (\cardof{\mathcal{P}} +
  \cardof{\mathcal{S}}) \cdot \cardof{\mathcal{X}}$
  triples of the form $(S,Q,X)$ and at most
  $\cardof{\mathcal{S}}$ pairs $(S,\bot)$ in an S-profile
  of sort $\set{1,2}$.
  Thus, we obtain: \begin{align*}
    \cardof{\universeOf{A}_\grammar} \in & ~2^{{\sizeof{\grammar}}^{\poly{\cardof{\mathcal{Y}}}}} + 2^{{\sizeof{\grammar}}^{\poly{\cardof{\mathcal{S}}}}} + \\
    & ~2^{\cardof{\mathcal{S}} \cdot (\cardof{\mathcal{P}} + \cardof{\mathcal{S}}) \cdot \cardof{\mathcal{X}} +  \cardof{\mathcal{S}}} \\
    \in & ~2^{2^{\poly{\sizeof{\grammar}}}}
  \end{align*}
  An element $a \in \universeOf{A}_\grammar$ is a set of either
  reduced multisets from $\mpow{\mathcal{Y}}$ or $\mpow{\mathcal{S}}$,
  or a set of triples $(S,Q,X)$ and pairs $(S,\bot)$.  As in the
  proofs of Lemmas~\ref{lemma:size-of-tree-rec}
  and~\ref{lemma:size-of-sp-rec}, we then obtain $\sizeof{a} \in
  \sizeof{\grammar}^\poly{\sizeof{\grammar}}$ for each $a \in
  \universeOf{A}_\grammar$.

  We compute an upper bound on the time needed to evaluate the function $f^{\algof{A}_\grammar}$, for each $f \in
  \fsignature_2$:
  \begin{itemize}[-]
  \item the cases $\emptygraph_{\set{1}}^{\algof{A}_\grammar}$,
    $\sgraph{b}_{1,2}^{\algof{A}_\grammar}$ and $a_1 \pop^\algof{A}
    a_2$ have similar counterparts in the proofs of Lemmas
    \ref{lemma:size-of-tree-rec} and \ref{lemma:size-of-sp-rec},
    respectively.
  \item $a_1 \rop a_2$: since $a_1$ is a P-profile of sort
    $\set{1,2}$, there are at most
    ${\sizeof{\grammar}}^{\poly{\cardof{\mathcal{S}}}}$
    many multisets $m \in a_1$ and the check $P
    \leadsto_{\grammar} m$, for a given nonterminal $P$ and
    a multiset $m \in a_1$, takes
    ${\sizeof{\grammar}}^{\poly{\cardof{\mathcal{S}}}}$
    time. Similarly, because $a_2$ is a profile of sort $\set{1}$,
    there are at most
    ${\sizeof{\grammar}}^{\poly{\cardof{\mathcal{Y}}}}$
    multisets $k \in a_2$ and the check $X \leadsto_\grammar k$, for a
    given multiset $k \in a_2$, takes time
    ${\sizeof{\grammar}}^{\poly{\cardof{\mathcal{Y}}}}$.
    Then, the application of the ($\rop$) inference rule from Figure
    \ref{fig:tw2-alg}, takes $2^{\poly{\sizeof{\grammar}}}$ time.
  \item $\sop^{\algof{A}_\grammar}(a_1,a_2,a_3)$: depending on the
    type of the profiles $a_1$ and $a_2$, the application of either of
    the inference rules ($\P\sop^i\P$), ($\P\sop^i\S$), ($\S\sop^i\P$)
    and ($\S\sop^i\S$) requires an iteration over the rules of
    $\grammar$ and the evaluation of at most three conditions $P
    \leadsto_{\grammar} m$ (resp. $X \leadsto_\grammar k$)
    for some given nonterminal $P \in \mathcal{P}$
    (resp. $X \in \mathcal{X}$) and multiset $m \in a_i$, $i=1,2$ (resp. $k\in a_3$). Since these checks take time
    ${\sizeof{\grammar}}^{\poly{\cardof{\mathcal{S}}}}$
    (resp. ${\sizeof{\grammar}}^{\poly{\cardof{\mathcal{Y}}}}$),
    the entire computation takes $2^{\poly{\sizeof{\grammar}}}$ time.
  \end{itemize}
\end{proofE}

We establish the counterparts of Lemmas
\ref{lemma:tree-view-membership} and
\ref{lemma:tree-view-homomorphism}, below:


\begin{lemmaE}[][category=proofs]\label{lemma:tw2-view-membership}
  For all $\graph_1, \graph_2 \in \universeOf{G}_2$, if
  $h_{\grammar}(\graph_1) = h_{\grammar}(\graph_2)$
  then $\graph_1\in\alangof{}{\twalg}{\grammar} \iff
  \graph_2\in\alangof{}{\twalg}{\grammar}$.
\end{lemmaE}
\begin{proofE}
  We assume that $\graph_1\in\alangof{}{\twalg}{\grammar}$ and prove
  $\graph_2\in\alangof{}{\twalg}{\grammar}$ (the other direction is
  symmetric). Since $h_{\grammar}(\graph_1) = h_{\grammar}(\graph_2)$,
  we distinguish three cases:

  \vspace*{.5\baselineskip}\noindent
  \underline{$\sortof{\graph_1}=\sortof{\graph_2}=\set{1}$}: the proof
  in this case is similar to the proof of Lemma
  \ref{lemma:tree-view-membership}.

  \vspace*{.5\baselineskip}\noindent
  \underline{$\sortof{\graph_1}=\sortof{\graph_2}=\set{1,2}$ and
    $\graph_1$, $\graph_2$ are P-graphs}: the proof in this case is
  similar to the proof of Lemma \ref{lemma:sp-view-membership} in the
  case $\graph_1,\graph_2 \in \universeOf{SP}$ are both P-graphs.

  \vspace*{.5\baselineskip}\noindent
  \underline{$\sortof{\graph_1}=\sortof{\graph_2}=\set{1,2}$ and
    $\graph_1$, $\graph_2$ are S-graphs}: Since $\graph_1 \in
  \alangof{}{\twalg}{\grammar}$, there exists a complete derivation $S
  \step{\grammar}^* u_1$ such that $u_1^\twalg=\graph_1$, for some
  axiom $\rightarrow S$ of $\grammar$. By Definition
  \ref{def:tw2-view} (\ref{it222:tw2-view}), we obtain $(S,\bot) \in
  h_\grammar(\graph_1)$, hence $(S,\bot) \in h_\grammar(\graph_2)$, by
  the equality $h_{\grammar}(\graph_1) =
  h_{\grammar}(\graph_2)$. Again by Definition~\ref{def:tw2-view}
  (\ref{it222:tw2-view}), we obtain a complete derivation $S
  \step{\grammar}^* u_2$ such that $u_2^\twalg=\graph_2$, hence
  $\graph_2 \in \alangof{}{\twalg}{\grammar}$, because $\rightarrow S$
  is an axiom.
\end{proofE}

\begin{lemmaE}[][category=proofs]\label{lemma:tw2-view-homomorphism}
  $h_{\grammar}$ is a homomorphism between $\twalg$ and
  $\algof{A}_\grammar$.
\end{lemmaE}
\begin{proofE}
  We prove the following points, for all $b \in \alphabetTwo$, $1 \leq
  i \leq 2$ and $\graph_1, \graph_2 \in \universeOf{G}_2$:

  \vspace*{.5\baselineskip}
  \noindent\underline{$h_\grammar(\emptygraph_{\set{1}}^\twalg) =
    \emptygraph_{\set{1}}^{\algof{A}_\grammar}$}: this case is similar
  to the case $\emptygraph_{\set{1}}^\algof{T}$ from the proof of
  Lemma \ref{lemma:tree-view-homomorphism}.

  \vspace*{.5\baselineskip}
  \noindent\underline{$h_\grammar(\graph_1 \pop^\twalg \graph_2) =
    h_\grammar(\graph_1) \pop^{\algof{A}_\grammar}
    h_\grammar(\graph_2)$}, where
  $\sortof{\graph_1}=\sortof{\graph_2}=\set{1}$: this case is similar
  to the case of parallel composition from the proof of Lemma
  \ref{lemma:tree-view-homomorphism}.

  \vspace*{.5\baselineskip}
  \noindent\underline{$h_\grammar(\sgraph{b}_{i,3-i}^\twalg) = \sgraph{b}_{i,3-i}^{\algof{A}_\grammar}$}: this case is similar to
  the case $\sgraph{b}_{1,2}^\algof{SP}$ from the proof of Lemma
  \ref{lemma:sp-view-homomorphism}.

  \vspace*{.5\baselineskip}
  \noindent\underline{$h_\grammar(\graph_1 \rop^\twalg \graph_2) =
    h_\grammar(\graph_1) \rop^{\algof{A}_\grammar}
    h_\grammar(\graph_2)$}: ``$\subseteq$'' Since $\graph_1
  \rop^\twalg \graph_2$ is a graph of sort $\set{1}$, let $\mset{Y}
  \in h_\grammar(\graph_1 \rop^\twalg \graph_2)$ be a reduced view of
  $\graph_1 \rop^\twalg \graph_2$. Because $\graph_1 \rop^\twalg
  \graph_2$ is not a parallel composition of two or more graphs of
  sort $\set{1}$, any element of $h_\grammar(\graph_1 \rop^\twalg
  \graph_2)$ is a singleton multiset, for some nonterminal $Y \in \mathcal{Y}$.
  Then, by Definition~\ref{def:tw2-view} (\ref{it1:tw2-view}), there exists a complete derivation $Y \step{\grammar}^* u$ such that $u^\twalg = \graph_1 \rop^\twalg
  \graph_2$. Without loss of generality, we can assume that this
  derivation has the form: \begin{align*} Y \step{\grammar} P \rop X
    \step{\grammar}^* v_1 \rop X \step{\grammar}^* v_1 \rop v_2
  \end{align*}
  for some ground $\fsignature_2$-terms $v$ and $w$ such that
  $v_i^\twalg=\graph_i$, for $i = 1,2$. By a reordering of the steps,
  if necessary, we can consider that the derivation $P
  \step{\grammar}^* v_1$ has the form:
\begin{align*}
    P \step{\grammar}^* ~\pop_{S \in \mathcal{S}} \pexp{S}{m(S)}
    \step{\grammar}^* ~\pop_{S \in \mathcal{S},i\in m(S)} w_1(S,i)
  \end{align*}
  where $w_1(S,i)$ are ground $\fsignature_2$-terms such that $\graph_1 = (\pop_{S \in \mathcal{S},i\in m(S)}
  w_1(S,i))^\twalg$.
  Then $m$ is a view of $\graph_1$ such that $P   \leadsto_{\grammar} m$.
  We obtain $\trunk{m}_\grammar \in h_\grammar(\graph_1)$ and $P \leadsto_{\grammar} \trunk{m}_\grammar$,
  by Lemma \ref{fact:trunk}.
  Similarly, by reordering the steps of the
  derivation $X \step{\grammar}^* v_2$, we obtain a view $k$ of
  $\graph_2$ such that $X \leadsto_\grammar k$, hence
  $\trunk{k}_\grammar \in h_\grammar(\graph_2)$ and $X
  \leadsto_\grammar \trunk{k}_\grammar$, by Lemma \ref{fact:trunk}. We
  obtain $\mset{Y} \in h_\grammar(\graph_1) \rop^{\algof{A}_\grammar}
  h_\grammar(\graph_2)$, by applying the inference rule ($\rop$)
  from Figure \ref{fig:tw2-alg}.

  \noindent''$\supseteq$'' Let $\mset{Y} \in h_\grammar(\graph_1)
  \rop^{\algof{A}_\grammar} h_\grammar(\graph_2)$ be a multiset. By
  the inference rule ($\rop$) from Figure~\ref{fig:tw2-alg}, which
  defines $\rop^{\algof{A}_\grammar}$, each element of
  $h_\grammar(\graph_1) \rop^{\algof{A}_\grammar}
  h_\grammar(\graph_2)$ is a singleton multiset containing some
  nonterminal $Y \in \mathcal{Y}$. By the inference rule
  ($\rop$), there exists multisets $m \in h_\grammar(\graph_1)$ and $k
  \in h_\grammar(\graph_2)$ and a rule $Y \rightarrow P \rop X$ in
  $\grammar$ such that $P \leadsto_{\grammar} m$ and $X
  \leadsto_\grammar k$. Then, we obtain complete derivations: \begin{align*}
    P \step{\grammar}^* & ~\pop_{S \in \mathcal{S}} \pexp{S}{m(S)}
    \step{\grammar}^* ~\pop_{S \in \mathcal{S},i\in m(S)} w_1(S,i) \\
    X \step{\grammar}^* & ~\pop_{Y \in \mathcal{Y}} \pexp{Y}{k(Y)}
    \step{\grammar}^*  ~\pop_{Y \in \mathcal{Y},i\in k(Y)} w_2(Y,i)
  \end{align*}
  where $w_1(S,i)$ and $w_2(Y,i)$ are ground $\fsignature_2$-terms
  such that $\graph_1 = (\pop_{S \in \mathcal{S},i\in m(S)} w_1(S,i))^\twalg$ and
  $\graph_2 = (\pop_{Y \in \mathcal{Y},i\in k(Y)}
  w_2(Y,i))^\twalg$.
  By composing these derivations with the rule $Y
  \rightarrow P \rop X$, we obtain a complete derivation $Y
  \step{\grammar}^* u$ such that $u^\twalg=\graph_1
  \rop^\twalg \graph_2$. Since $\graph_1 \rop^\twalg \graph_2 \in
  \universeOf{A}_{\set{1}}$, we obtain $\mset{Y} \in
  h_\grammar(\graph_1 \rop^\twalg \graph_2)$, by Definition
  \ref{def:tw2-view} (\ref{it1:tw2-view}).

  \vspace*{.5\baselineskip}
  \noindent\underline{$h_\grammar(\graph_1 \pop^\twalg \graph_2) =
    h_\grammar(\graph_1) \pop^{\algof{A}_\grammar}
    h_\grammar(\graph_2)$}, where
  $\sortof{\graph_1}=\sortof{\graph_2}=\set{1,2}$: this case is
  similar to the case of parallel composition from the proof of Lemma
  \ref{lemma:sp-view-homomorphism}.

  \vspace*{.5\baselineskip}
  \noindent\underline{$h_\grammar(\sop^\twalg(\graph_1,\graph_2,H)) =
    \sop^{\algof{A}_\grammar}(h_\grammar(\graph_1),h_\grammar(\graph_2),h_\grammar(H))$}:
  Note that $\graph \isdef \sop^\twalg(\graph_1,\graph_2,H)$ is an S-graph, hence
  $h_\grammar(\graph)$ is an S-profile,
  whose elements are of the form $(S,Q,X)$ or $(S,\bot)$, where $S \in \mathcal{S}$, $Q \in \mathcal{P} \uplus
  \mathcal{S}$ and $X \in \mathcal{X}$.
  In the following, we tackle the cases of elements of the form $(S,Q,X)$,
  the other cases being similar and left to the reader. Based on the
  types of $\graph_1$ and $\graph_2$, we distinguish the following
  cases: \begin{itemize}[-]
  \item $\graph_1$ and $\graph_2$ are P-graphs: ``$\subseteq$'' Let
    $(S,Q,X) \in h_\grammar(\graph)$ be a reduced view of $\graph$. By
    Definition~\ref{def:tw2-view} (\ref{it221:tw2-view}), there exists
    a derivation $S \astep{\grammar}^* \sop(u,Q,X)$, for some ground
    $\fsignature_2$-term $u$ such that $\graph=u^\twalg$. By Lemma
    \ref{lemma:tw2-assoc}, there exist derivations $S
    \astep{\grammar}^* \sop(u_1,S_1,w)$ and $S_1 \astep{\grammar}^*
    \sop(u_2,Q,X)$, for some nonterminal $S_1 \in \mathcal{S}$ and
    ground $\fsignature_2$-terms $u_1, u_2$ and $w$ such that
    $u_i^{\algof{G}_2} = \graph_i$, for $i=1,2$ and $w^{\algof{G}_2} =
    H$. Then, by Lemma~\ref{lemma:tw2-assoc-start}, and because
    $\graph_1$ and $\graph_2$ are P-graphs, there are rules $S
    \rightarrow \sop(P_1,S_1,X_1)$ and $S_1 \rightarrow
    \sop(P_2,Q,X)$, and complete derivations $P_i \step{\grammar}^*
    u_i$ and $X_1 \step{\grammar}^* w$, for some non-terminals $P_1,
    P_2$ and $X_1$.  By reordering these derivations, if necessary, we
    obtain views $m_i$ of $\graph_i$ and $k$ of $H$ such that $P_i
    \leadsto_{\grammar} m_i$ and $X_1 \leadsto_\grammar k$.  By Lemma
    \ref{fact:trunk}, we obtain $P_i \leadsto_{\grammar}
    \trunk{m_i}_\grammar \in h_\grammar(\graph_i)$ and $X_1
    \leadsto_\grammar \trunk{k}_\grammar \in h_\grammar(H)$.  Hence,
    we obtain $(S,Q,X) \in
    \sop^{\algof{A}_\grammar}(h_\grammar(\graph_1),
    h_\grammar(\graph_2), h_\grammar(H))$, by applying the inference
    rule ($\sop^1(\P,\P)$) from Figure \ref{fig:tw2-alg}.

    \noindent''$\supseteq$'' Let $(S,Q,X) \in
    \sop^{\algof{A}_\grammar}(h_\grammar(\graph_1),h_\grammar(\graph_2),h_\grammar(H))$
    be a triple. Since $h_\grammar(\graph_1)$ and
    $h_\grammar(\graph_2)$ are P-profiles, by the inference rule
    ($\sop^1(\P,\P)$) from Figure \ref{fig:tw2-alg}, we have rules $S
    \rightarrow \sop(P_1,S_1,X_1), S' \rightarrow \sop(P_2,Q,X) \in
    \rules$ and reduced views $m_i \in h_\grammar(\graph_i)$ such that
    $P_i \leadsto_{\grammar} m_i$, for $i=1,2$ and $k \in
    h_\grammar(H)$ such that $X_1 \leadsto_\grammar k$. Then, there
    exist views $m'_i$ of $\graph_i$, for $i=1,2$ and $k'$ of $H$. The
    existence of these views implies the existence of complete
    derivations $P_i \step{\grammar}^* u_i$ and $X_1
    \step{\grammar}^* w$ such that $u_i^\twalg=\graph_i$,
    for $i=1,2$ and $w^\twalg=H$. Then, we obtain derivations $S
    \astep{\grammar}^* \sop(u_1,S_1,w)$ and $S_1
    \astep{\grammar}^* \sop(u_2,Q,X)$, hence there exists a
    derivation $S \astep{\grammar}^* \sop(u,Q,X)$, for some
    ground $\fsignature_2$-term $u$ such that
    $u^\twalg=\sop^\twalg(\graph_1,\graph_2,H)=\graph$. Hence,
    $(S,Q,X) \in h_\grammar(\graph)$, by Definition~\ref{def:tw2-view}
    (\ref{it221:tw2-view}).
  \item $\graph_1$ is a P-graph and $\graph_2$ is an S-graph:
    ``$\subseteq$'' Let $(S,Q,X) \in h_\grammar(\graph)$ be a reduced
    view of $\graph$. Similar to the previous case, we obtain
    derivations $S \astep{\grammar}^* \sop(u_1,S_1,w)$ and
    $S_1 \astep{\grammar}^* \sop(u_2,Q,X)$, for some
    nonterminal $S_1 \in \mathcal{S}$ and ground
    $\fsignature_2$-terms $u_1, u_2$ and $w$ such that
    $u_i^{\algof{G}_2} = \graph_i$, for $i=1,2$ and $w^{\algof{G}_2} =
    H$.
    Then, by Lemma~\ref{lemma:tw2-assoc-start}, and because $\graph_1$ is a P-graph, there is a rule $S \rightarrow \sop(P_1,S_1,X_1)$, and complete derivations $P_1 \step{\grammar}^* u_1$ and $X_1 \step{\grammar}^* w$, for some non-terminals $P_1$ and $X_1$.
    By reordering these derivations, if necessary, we obtain views $m_1$ of $\graph_1$ and $k$ of $H$ such that $P_1 \leadsto_{\grammar} m_1$ and $X_1 \leadsto_\grammar k$.
    By Lemma \ref{fact:trunk}, we obtain $P_1 \leadsto_{\grammar} \trunk{m_1}_\grammar \in h_\grammar(\graph_1)$ and $X_1 \leadsto_\grammar \trunk{k}_\grammar \in h_\grammar(H)$.
    Moreover, since $\graph_2$ is an S-graph, we have $(S_1,Q,X) \in h_\grammar(\graph_2)$, by Definition~\ref{def:tw2-view} (\ref{it221:tw2-view}).
    Then, we obtain $(S,Q,X) \in \sop^{\algof{A}_\grammar}(h_\grammar(\graph_1),
    h_\grammar(\graph_2), h_\grammar(H))$, by applying the inference
    rule ($\sop^1(\P,\S)$) from Figure \ref{fig:tw2-alg}.

    \noindent''$\supseteq$'' Let $(S,Q,X) \in
    \sop^{\algof{A}_\grammar}(h_\grammar(\graph_1),h_\grammar(\graph_2),h_\grammar(H))$
    be a triple. Since $h_\grammar(\graph_1)$ is a P-profile and
    $h_\grammar(\graph_2)$ is an S-profile, by the inference rule
    ($\sop^1(\P,\S)$) from Figure \ref{fig:tw2-alg}, we have a rule $S
    \rightarrow \sop(P,S_1,X_1) \in \rules$ and reduced views $m_1 \in
    h_\grammar(\graph_1)$, $(S_1,Q,X) \in h_\grammar(\graph_2)$ and $k
    \in h_\grammar(H)$, such that $P_1 \leadsto_{\grammar}
    m_1$ and $X_1 \leadsto_\grammar k$. Then, there exist views $m'_1$
    of $\graph_1$ and $k'$ of $H$, from which we obtain complete
    derivations $P_1 \step{\grammar}^* u_1$ and $X_1
    \step{\grammar}^* w$ such that $u_1^\twalg=\graph_1$ and
    $w^\twalg=H$. By combining these derivations with the rule $S
    \rightarrow \sop(P,S_1,X_1)$, we obtain a derivation $S
    \astep{\grammar}^* \sop(u_1,S_1,w_1)$. Moreover, because
    $(S_1,Q,X)$ is a view of $\graph_2$, there exists a derivation
    $S_1 \astep{\grammar} \sop(u_2,Q,X)$ such that
    $u_2^\twalg=\graph_2$. By Lemma \ref{lemma:tw2-assoc}, we obtain a
    derivation $S \astep{\grammar}^* \sop(u,Q,X)$ such that
    $u^\twalg=\graph$, hence $(S,Q,X) \in h_\grammar(\graph)$, by
    Definition~\ref{def:tw2-view} (\ref{it221:tw2-view}).
  \item $\graph_1$ is an S-graph and $\graph_2$ is a P-graph:
    ''$\subseteq$'' Let $(S,Q,X) \in h_\grammar(\graph)$ be a reduced
    view of $\graph$. Similar to the previous cases, we obtain
    derivations $S \astep{\grammar}^* \sop(u_1,S_1,w)$ and
    $S_1 \astep{\grammar}^* \sop(u_2,Q,X)$, for some
    nonterminal $S_1 \in \mathcal{S}$ and ground
    $\fsignature_2$-terms $u_1, u_2$ and $w$ such that
    $u_i^{\algof{G}_2} = \graph_i$, for $i=1,2$ and $w^{\algof{G}_2} =
    H$. Because $\graph_1$ is an S-graph, we can assume w.l.o.g. that
    the derivation $S \astep{\grammar}^* \sop(u_1,S_1,w)$ is
    of the form: \begin{align*}
      S \astep{\grammar}^* \sop(u_1,S_1,X_1) \step{\grammar}^* \sop(u_1,S_1,w)
    \end{align*}
    for a complete derivation $X_1 \step{\grammar}^* w$ such
    that $w^\twalg=H$. By a reordering of the steps of this
    derivation, if necessary, we obtain a view $k$ of $H$ such that
    $X_1 \leadsto_\grammar k$. Moreover, by Lemma \ref{fact:trunk}, we
    obtain $X_1 \leadsto_\grammar \trunk{k}_\grammar \in
    h_\grammar(H)$. Since $S \astep{\grammar}^*
    \sop(u_1,S_1,X_1)$, we have $(S,S_1,X_1) \in
    h_\grammar(\graph_1)$, by Definition~\ref{def:tw2-view}
    (\ref{it221:tw2-view}).
    Then, by Lemma~\ref{lemma:tw2-assoc-start}, and because $\graph_2$ is a P-graph, there is a rule $S_1 \rightarrow \sop(P,Q,X)$, and a complete derivation $P \step{\grammar}^* u_2$.
    By a reordering of the steps of this derivation, if necessary, we
    obtain a view $m$ of $\graph_2$ such that $P
    \leadsto_{\grammar} m$.
    Moreover, by Lemma
    \ref{fact:trunk}, we obtain $P \leadsto_{\grammar}
    \trunk{m}_\grammar \in h_\grammar(\graph_2)$. Then, we obtain
    $(S,Q,X) \in \sop^{\algof{A}_\grammar}(h_\grammar(\graph_1),
    h_\grammar(\graph_2), h_\grammar(H))$, by applying the inference
    rule ($\sop^1(\S,\P)$) from Figure \ref{fig:tw2-alg}.

    \noindent''$\supseteq$'' Let $(S,Q,X) \in
    \sop^{\algof{A}_\grammar}(h_\grammar(\graph_1),h_\grammar(\graph_2),h_\grammar(H))$
    be a triple. Since $h_\grammar(\graph_1)$ is an S-profile and
    $h_\grammar(\graph_2)$ is a P-profile, by the inference rule
    ($\sop^1(\S,\P)$) from Figure \ref{fig:tw2-alg}, we have a rule
    $S_1 \rightarrow \sop(P,Q,X) \in \rules$ and reduced views
    $(S,S_1,X_1)$ of $\graph_1$, $m$ of $\graph_2$ and $k$ of $H$ such
    that $P \leadsto_{\grammar} m$ and $X_1
    \leadsto_\grammar k$. By Lemma \ref{fact:trunk}, there exist views
    $m'$ of $\graph_2$ and $k'$ of $H$ such that $P
    \leadsto_{\grammar} m$ and $X_1 \leadsto_\grammar k$.
    Since $(S,S_1,X_1) \in h_\grammar(\graph_1)$, by Definition
    \ref{def:tw2-view} (\ref{it221:tw2-view}), there exists a
    derivation $S \astep{\grammar}^* \sop(u_1,S_1,X_1)$ such
    that $u_1^\twalg=\graph_1$. From the rule $S_1 \rightarrow
    \sop(P,Q,X)$ and the view $m'$ of $\graph_2$, we obtain a
    derivation $S_1 \astep{\grammar}^* \sop(u_2,Q,X)$ such
    that $u_2^\twalg=\graph_2$. From the view $k'$ of $H$ we obtain a
    derivation $X_1 \step{\grammar}^* w$ such that
    $w^\twalg=H$. By Lemma \ref{lemma:tw2-assoc}, there exists a
    derivation $S \astep{\grammar}^* \sop(u,Q,X)$ such that
    $u^\twalg=\sop^\twalg(\graph_1,\graph_2,H)=\graph$, hence $(S,Q,X)
    \in h_\grammar(\graph)$, by Definition~\ref{def:tw2-view}
    (\ref{it221:tw2-view}).
  \item $\graph_1$ and $\graph_2$ are S-graphs: ''$\subseteq$'' Let
    $(S,Q,X) \in h_\grammar(\graph)$ be a reduced view of $\graph$. As
    in the previous case, there exist reduced views $(S,S_1,X_1) \in
    h_\grammar(\graph_1)$ and $k \in h_\grammar(H)$, such that $X_1
    \leadsto_\grammar k$. Moreover, as in the previous case, we obtain
    a derivation $S_1 \astep{\grammar}^* \sop(u,Q,X)$, for a
    ground $\fsignature_2$-term $u$ such that $u^\twalg=\graph_2$,
    hence $(S_1,Q,X) \in h_\grammar(\graph_2)$, by Definition
    \ref{def:tw2-view} (\ref{it221:tw2-view}). Then, we obtain
    $(S,Q,X) \in \sop^{\algof{A}_\grammar}(h_\grammar(\graph_1),
    h_\grammar(\graph_2), h_\grammar(H))$, by applying the inference
    rule ($\sop^1(\S,\S)$) from Figure \ref{fig:tw2-alg}.

    \noindent''$\supseteq$'' Let $(S,Q,X) \in
    \sop^{\algof{A}_\grammar}(h_\grammar(\graph_1),h_\grammar(\graph_2),h_\grammar(H))$
    be a triple. Since both $h_\grammar(\graph_1)$ and
    $h_\grammar(\graph_2)$ are S-profiles, by the inference rule
    ($\sop^1(\S,\S)$) from Figure \ref{fig:tw2-alg}, we have reduced
    views $(S,S_1,X_1) \in h_\grammar(\graph_1)$, $(S_1,Q,X) \in
    h_\grammar(\graph_2)$ and $k \in h_\grammar(H)$ such that $X_1
    \leadsto_\grammar k$. By Lemma \ref{fact:trunk}, there exists a
    view $k'$ of $H$ such that $X_1 \leadsto_\grammar k'$ and a
    complete derivation $X_1 \step{\grammar}^* w$ such that
    $w^\twalg=H$. By Definition~\ref{def:tw2-view}
    (\ref{it221:tw2-view}), we obtain derivations $S
    \astep{\grammar}^* \sop(u_1,S_1,X_1)$ and $S_1
    \astep{\grammar}^* \sop(u_2,Q,X)$ such that
    $u_i^\twalg=\graph_i$, for $i=1,2$. By combining the derivations
    $S \astep{\grammar}^* \sop(u_1,S_1,X_1)$ and $X_1
    \step{\grammar}^* w$ we obtain $S
    \astep{\grammar}^* \sop(u_1,S_1,w)$, hence a derivation
    $S \astep{\grammar}^* \sop(u,Q,X)$ such that
    $u^\twalg=\sop^\twalg(\graph_1,\graph_2,H)=\graph$, by Lemma
    \ref{lemma:tw2-assoc}. Then, $(S,Q,X) \in h_\grammar(\graph)$, by
    Definition~\ref{def:tw2-view} (\ref{it221:tw2-view}).
  \end{itemize}
\end{proofE}

Furthermore, Lemma \ref{lemma:aperiodic-tree-rec} carries over from
the class of trees, because the idempotent power $\idemof{a}$ of each
S-profile $a$ has the aperiodic property $\idemof{a}
\pop^{\algof{A}_\grammar} a = \idemof{a}$, by the definition of
$\pop^{\algof{A}_\grammar}$ for S-profiles.

The main result of this subsection is the ``if'' direction of Theorem~\ref{thm:main} for the class of graphs of tree-width at most~$2$.

\begin{theoremE}[][category=proofs]\label{thm:tw2-reg}
  The language of each (aperiodic) regular tree-width $2$ grammar $\grammar$ is
  recognizable by an (aperiodic) recognizer $(\algof{A},B)$.
  Moreover,
  $\cardof{\universeOf{A}} \in 2^{2^{\poly{\size{\grammar}}}}$
  and $f^{\algof{A}}(a_1,\ldots,a_{\arityof{f}})$ can be computed in time
  $2^{\poly{\size{\grammar}}}$, for each function symbol $f \in
  \spsignature$ and elements $a_1,\ldots,a_n \in \universeOf{A}$.
\end{theoremE}
\begin{proofE}
  Let $\grammar=(\mathcal{X}\uplus\mathcal{Y}\uplus\mathcal{P}\uplus\mathcal{S},\rules)$ be a regular (aperiodic) tree-width $2$ grammar.
  Using
  Lemma~\ref{lemma:normal-form}, we obtain a normalized (aperiodic) grammar
  $\grammar'=(\mathcal{X}\uplus\mathcal{Y}\uplus\mathcal{P}\uplus\mathcal{S},\rules')$.
  We point out that $\grammar$ and $\grammar'$ have the same set of  nonterminals.
  In particular, we have
  $\alangof{}{\algof{\twalg}}{\grammar'} =
  \alangof{}{\algof{\twalg}}{\grammar}$ and $\sizeof{\grammar'} \in
  2^{\poly{\cardof{\mathcal{X}}+\cardof{\mathcal{Y}}+\cardof{\mathcal{P}}+\cardof{\mathcal{S}}} \cdot \log(\sizeof{\grammar})}$.
  We then consider the corresponding normalized regular grammar in alternative form $\widehat{\grammar'}=(\mathcal{X}\uplus\mathcal{Y}\uplus\mathcal{P}\uplus\widehat{\mathcal{S}},\widehat{\rules'})$.
  We recall that $\alangof{}{\algof{\twalg}}{\widehat{\grammar'}} = \alangof{}{\algof{\twalg}}{\grammar'}$ and $\size{\grammar'} = \Theta(\size{\widehat{\grammar'}})$.

  We build the recognizer $(\algof{A}_{\widehat{\grammar'}},B_{\widehat{\grammar'}})$
  based on the grammar $\widehat{\grammar'}$.
  By Lemma
  \ref{lemma:tw2-view-membership}, we obtain
  $\alangof{}{\twalg}{\widehat{\grammar'}} = h_{\widehat{\grammar'}}^{-1}(B_{\widehat{\grammar'}})$,
  where $B_{\widehat{\grammar'}} \isdef \set{h_{\widehat{\grammar'}}(\graph) \mid \graph
    \in \alangof{}{\twalg}{\widehat{\grammar'}}}$.
    By Lemma
  \ref{lemma:tw2-view-homomorphism},
  $(\algof{A}_{\widehat{\grammar'}},B_{\widehat{\grammar'}})$ is a recognizer for $\alangof{}{\twalg}{\widehat{\grammar'}}$.
  Because of $\alangof{}{\twalg}{\widehat{\grammar'}} = \alangof{}{\twalg}{\grammar}$,
  $\algof{A}_{\widehat{\grammar'}}$ is also a recognizer for the language of $\grammar$ in $\twalg$.
  Moreover, $\algof{A}_{\widehat{\grammar'}}$
  is aperiodic if $\widehat{\grammar'}$ is aperiodic, by the counterpart of
  Lemma \ref{lemma:aperiodic-tree-rec} for $\algof{G}_2$.
  We further recall that $\widehat{\grammar'}$ is aperiodic iff $\grammar'$ is aperiodic.
  The complexity upper bounds follow in the
  same way as in the proof of Theorem~\ref{thm:tree-reg}.
\end{proofE}

\section{Logical Definability of Aperiodic Recognizable Sets}
\label{sec:logic}

This section focuses on the relation between definability in Monadic
Second Order Logic (\mso) and recognizability of sets of graphs from
the three classes considered in the previous. The equivalence between
recognizability and \mso-definability for sets of words and ranked
sets of trees are well-known \cite{Buechi90,Doner70}. For unranked
sets of trees, recognizability equals definability in Counting
\mso\ (\cmso), i.e., \mso\ extended with modulo constraints on
cardinalities of sets~\cite[Theorem 5.3]{CourcelleI}. For graphs, in
general, \cmso-definability implies recognizability in the graph
algebra $\algof{G}$~\cite[Theorem 4.4]{CourcelleI}, but not
viceversa\footnote{There are uncountably many recognizable sets of
graphs and only countably many \cmso-definable ones~\cite[Proposition
  2.14]{CourcelleI}.}. Moreover, the equivalence between
recognizability in the graph algebra and \cmso-definability has been
established for classes of graphs defined by bounded
tree-width~\cite{10.1145/2933575.2934508,journals/lmcs/BojanczykP22},
also known as partial $k$-trees~\cite{BODLAENDER19981}.

Since recognizability in $\algof{G}$ implies recognizability in each
class $\class$ (i.e., derived algebra of $\algof{G}$),
\cmso-definability implies recognizability in $\class$. However, the
contrapositive does not hold, even if the domain of $\class$ has
tree-width bounded by an integer $k$, because recognizability in
$\class$ does not necessarily imply recognizability in the class
$\algof{G}^{\leq k}$. The equivalence is recovered if the class
$\class$ is \emph{parsable}, i.e., for each $\graph \in
\universeOf{C}$, some ground $\fsignature_\class$-term $t$ such that
$t^\class = \graph$ can be derived using a finite set of
\mso\ formul{\ae}~\cite[Theorem 4.8(2)]{CourcelleV}.

The main result of this section is the following refinement of the
equivalence between \cmso-definable and recognizable subsets of
parsable classes:

\begin{theorem}\label{thm:mso-aperiodic}
  Let $\class$ be a parsable class of graphs. A set
  $\langu\subseteq\universeOf{C}$ is recognizable by an aperiodic
  recognizer if and only if it is \mso-definable.
\end{theorem}

This result is a generalization of a similar statement that has been
already proved for the class $\algof{T}$ of trees~\cite[Theorem
  5.10]{journals/corr/abs-2008-11635}. Moreover, the classes
$\algof{(D)SP}$ of (disoriented) series-parallel graphs have been
already shown to be parsable~\cite[Theorem 6.10]{CourcelleV}, hence
Theorem~\ref{thm:mso-aperiodic} applies to them. For the sake of
completeness, we show here that $\algof{G}_2$ is parsable, thus
extending the characterization of the aperiodic recognizable sets of
the $\algof{G}_2$ class.

Furthermore, each class $\btwclass{k}$ has been shown to be parsable in~\cite{journals/corr/abs-2310-04764},
based on the combined results of~\cite{10.1145/2933575.2934508,journals/lmcs/BojanczykP22}, hence
Theorem~\ref{thm:mso-aperiodic} applies to these classes as well:

\begin{corollary}\label{cor:mso-aperiodic-btw}
  For each integer $k\geq1$, a set
  $\langu\subseteq\btwuniverse{k}$ is recognizable in $\btwclass{k}$ by an
  aperiodic recognizer if and only if it is \mso-definable.
\end{corollary}

The rest of this section introduces the definitions and lemmas for the
proof of Theorem~\ref{thm:mso-aperiodic}.

\subsection{Counting Monadic Second Order Logic}
\label{subsec:cmso}

The counting monadic second order logic (\cmso) is the set of
formul{\ae} defined inductively by the following syntax:
\begin{align*}
  \psi := & \ x=y \mid a(x,y_1,\ldots,y_{\arityof{a}}),~ a \in \alphabet \mid X(x) \mid \cardconstr{X}{q}{p} \\
  & \mid \psi \land \psi \mid \neg\psi \mid \exists x ~.~ \psi \mid \exists X ~.~ \psi
\end{align*}
where $x,y_1,\ldots,y_{\arityof{a}} \in \vars$ are individual
variables, $X \in \Vars$ are set variables and $p,q \in \nat$ are
integers such that $p \in \interv{0}{q-1}$. By \mso\ we denote the
subset of \cmso\ consisting of formul{\ae} that do not contain atomic
propositions of the form $\cardconstr{X}{q}{p}$. A sentence is a
formula in which each variable occurs in the scope of a quantifier.

The semantics of \cmso\ is given by the usual satisfaction relation
$\graph \models^\store \psi$, where $\graph$ is a graph and $\store$
is a mapping of individual (resp. set) variables to vertices or edges
(resp. sets mixing vertices with edges). The relation is defined
inductively on the structure of formul{\ae}, as usual. In particular,
$\graph \models^\store a(x,y_1,\ldots,y_n)$ holds iff $\store(x)$ is
an $a$-labeled edge attached to vertices $\store(y_1), \ldots,
\store(y_n)$, and $\graph \models^\store \cardconstr{X}{q}{p}$ holds
iff $\cardof{\store(X)}=p \mod q$.

If $\phi$ is a sentence, we write $\graph \models \phi$ instead of
$\graph \models^\store \phi$. A set of graphs $\langu$ is
\mscmso-definable iff there exists a \mscmso\ sentence~$\phi$
such that $\langu = \set{\graph \mid \graph \models \phi}$. We recall
the following result relating \mso-definability with recognizability
for unranked sets of trees:

\begin{theorem}[Theorem 5.10 in \cite{journals/corr/abs-2008-11635}]
\label{thm:mso-aperiodic-trees}
  A set $\langu \subseteq \universeOf{T}$ of trees is \mso-definable
  if and only if $\langu$ is recognizable in $\algof{T}$ by an
  aperiodic recognizer.
\end{theorem}

A definable \emph{transduction} is a relation on graphs, defined by
finitely many \mso\ formul{\ae}. For an input graph $\graph$, a
definable transduction produces zero or more output graphs, as
follows: \begin{compactenum}[1.]
\item $\graph$ is accepted or rejected according to a formula
  $\varphi(X_1, \ldots, X_n)$ that, moreover, assigns sets of vertices
  and edges from $\graph$ to the variables $X_1, \ldots, X_n$; the
  transduction is said to be parameterless if $n=0$.
\item $\graph$ is copied into $k\geq1$ disjoint layers,
\item for the chosen valuation of $X_1,\ldots,X_n$, the output graph
  $\graph'$ consists of the vertices and edges $x$ from the $i$-th
  layer, as defined by formul{\ae}
  $\psi^{\mathsf{vert}/\mathsf{edge}}_i(x,X_1,\ldots,X_n)$, for $i \in
  \interv{1}{k}$. The incidence relation between the $a$-labeled edges
  $x$ and the vertices $y_1,\ldots,y_{\arityof{a}}$ of $\graph'$,
  taken from the $i_1, \ldots, i_{\arityof{a}}$ layers, are defined by
  formul{\ae}
  $\theta_{(a,i_1,\ldots,i_{\arityof{a}})}(x,y_1,\ldots,y_{\arityof{a}},X_1,\ldots,X_n)$.
\end{compactenum}
The relation $\trans \subseteq \universeOf{G} \times \universeOf{G}$
is defined by the sets $\trans(\graph)$ of outputs for the input graph
$\graph$, for all valuations of $X_1, \ldots, X_n$.

The main property of definable transductions is the \emph{Backwards
Translation Theorem}, which states that \mscmso-definable sets are
closed under inverse definable transductions. A consequence is that
definable transductions are closed under relational composition:
\begin{theorem}[Theorem 1.40 in \cite{courcelle_engelfriet_2012}]\label{thm:bt}
  If $\langu$ is an \mscmso-definable set of graphs and $\trans$ is an
  definable transduction then the set $\trans^{-1}(\langu)$ is
  \mscmso-definable. Consequently, the composition of definable
  transductions is a definable transduction.
\end{theorem}

\subsection{Parse Tree Algebras}
\label{subsec:parse-tree-algebras}

\begin{figure}[t!]
  \centerline{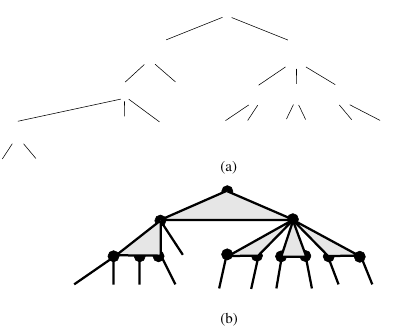}
  \vspace*{-.5\baselineskip}
  \caption{A $\fsignature_2$-term (a) and its corresponding parse tree
    (b) The order of the nodes attached to an edge is
    counter-clockwise.}
  \vspace*{-\baselineskip}
  \label{fig:val}
\end{figure}

In the rest of this section, we fix an alphabet $\alphabet$ and a
class $\class$ with signature $\fsignature_\class$ and sorts
$\sorts_\class$ (Definition~\ref{def:class}). Let
\begin{align*}
  \alphabetParse \isdef
  \set{\symbOf{f} \mid f \in \fsignature_\class \setminus
    \set{\pop_{\slabs,\slabs} \mid \slabs \in \sorts_\class}}
\end{align*}
be an alphabet of edge labels of arities $\arityof{\symbOf{f}} =
\arityof{f}+1$. A \emph{parse tree} for $\class$ is a tree, having
$\alphabetParse$-labeled edges, obtained from a
$\fsignature_\class$-term $t$ by collapsing each
$\pop$-subterm\footnote{A subterm consisting of occurrences of some
symbol $\pop_{\slabs,\slabs}$ and variables, for some $\slabs \in
\sorts_\class$.} of $t$ into a single node and transforming each
occurrence of a function symbol $f\in\fsignature_\class$ from $t$ into
an edge with label $\symbOf{f}\in\alphabetParse$.

\begin{example}
Figure \ref{fig:val} (a) shows a $\fsignature_2$-term that evaluates
to the rightmost graph on Figure \ref{fig:tw2} (b), where all edge
labels are $b\in\alphabetTwo$. Figure \ref{fig:val} (b) shows the
parse tree of that term.
\end{example}


For a given class $\class$, the \emph{parse tree algebra}
$\repalgebra_\class$ is a $(\fsignature_\class,\sorts_\class)$-algebra
of trees over the edge label alphabet $\alphabetParse$. The domain
$\repdomain_\class$ of $\repalgebra_\class$ is the set of
$\alphabetParse$-labeled trees having associated sorts from
$\sorts_\class$. The function symbols $f \in
\fsignature_\class\setminus \set{\pop_{\slabs,\slabs} \mid \slabs \in
  \sorts_\class}$ are interpreted over $\alphabetParse$-labeled trees
instead of graphs as follows:
\begin{align*}
f^{\repalgebra_\class}(\tree_1,\ldots,\tree_\arityof{f}) \isdef
~\extend{\symbOf{f}}^{\algof{T}}(\tree_1,\ldots,\tree_\arityof{f})
\end{align*}

We denote by $\treealgebra(\alphabetParse)$ the algebra of
$\treealgebra(\alphabetParse)$-labeled trees, of sort $\set{1}$, and
its universe by $\trees(\alphabetParse)$. Note that not every tree
from $\trees(\alphabetParse)$ is a parse tree, because parse trees
correspond to well-sorted terms. However, recognizability in the
multi-sorted algebra of parse trees equals recognizability in the
single-sorted tree algebra:


\begin{lemmaE}[][category=proofs]\label{lemma:red-rec}
  Let $\class$ be a class. Each set $\langu \subseteq
  \repdomain_\class$ of parse trees is (aperiodic) recognizable in the
  parse tree algebra $\repalgebra_\class$ iff $\langu$ is
  (aperiodic) recognizable in the tree algebra
  $\treealgebra(\alphabetParse)$.
\end{lemmaE}
\begin{proofE}
  ``$\Rightarrow$'' Assume that $\langu \subseteq \repdomain_\class$
  is (aperiodic) recognizable in $\repalgebra_\class$ and let
  $\algof{B}=(\set{\universeOf{B}^\slabs}_{\slabs\in\sorts_\class},
  \set{f^{\algof{B}}}_{f\in\fsignature_\class})$ be a (aperiodic)
  locally finite $(\fsignature_\class,\sorts_\class)$-algebra, such
  that $\langu = \homof{\algof{P}_\class}{B}^{-1}(D)$, for some set
  $D\subseteq\universeOf{B}$. We define the locally finite algebra
  $\overline{\algof{B}}$ having domain
  $\overline{\universeOf{B}}^{\set{1}} \isdef \set{(b,\slabs) \mid b
    \in \universeOf{B}^\slabs,\slabs\in\sorts_\class} \cup \set{\bot}$
  and a single sort $\set{1}$, that interprets each function symbol
  from $\fsignature_{\treealgebra(\alphabetParse)}$ as follows: \begin{align*}
    \extend{\symbOf{f}}^{\overline{\algof{B}}}(\overline{b}_1, \ldots, \overline{b}_{\arityof{f}}) \isdef & \begin{cases}
      (f^\algof{B}(b_1,\ldots, b_\arityof{f}),\sortof{f}), \\
      \hspace*{5mm} \text{ if } \overline{b}_i = (b_i,\slabs_i) \text{, for } i \in \interv{1}{\arityof{f}} \\
      \hspace*{5mm} \text{ and } \rankof{f} = \tuple{\slabs_1, \ldots,\slabs_\arityof{f}} \\
      \bot, \text{ otherwise}
    \end{cases} \\
    (b_1,\slabs_1) \pop^{\overline{\algof{B}}} (b_2,\slabs_2) \isdef & \begin{cases}
      (b_1 \pop^\algof{B} b_2, \slabs_1), \text{ if } \slabs_1 = \slabs_2 \\
      \bot, \text{ otherwise}
    \end{cases}
  \end{align*}
  We prove that $\algof{B}$ is aperiodic only if
  $\overline{\algof{B}}$ is aperiodic. Let $(b,\slabs) \in
  \overline{\universeOf{B}}$ be an element. Then, $\idemof{(b,\slabs)}
  = (\idemof{b},\slabs)$, by the definition of
  $\pop^{\overline{\algof{B}}}$. By the aperiodicity of $\algof{B}$,
  we obtain $\idemof{(b,\slabs)} \pop^{\overline{\algof{B}}}
  (b,\slabs) = (\idemof{b}\pop^\algof{B} b, \slabs) =
  (\idemof{b},\slabs) = \idemof{(b,\slabs)}$, thus
  $\overline{\algof{B}}$ is aperiodic, because the choice of
  $(b,\slabs)$ was arbitrary.
  We define the set
  $\overline{D} \isdef \set{(b,\slabs) \mid b \in \universeOf{B}^\slabs \cap D, \slabs \in \sorts_\class}$
  and the function
  $\overline{h} : \treealgebra(\alphabetParse) \rightarrow \overline{\universeOf{B}}$ as follows:
  \begin{equation*}
  \overline{h}(\tree) \isdef 
    (\homof{\algof{P}_\class}{B}(\tree),\slabs), \text{ for } \tree \in \universeOf{P}_\class^\slabs
  \end{equation*}
  By the definition of $\overline{h}$ and $\overline{D}$, we have
  $\homof{\algof{P}_\class}{B}(\tree) \in D \iff \overline{h}(\tree)
  \in \overline{D}$, for each $\tree \in \universeOf{P}_\class$, hence
  $\langu=\homof{\algof{P}_\class}{B}^{-1}(D)=\overline{h}^{-1}(\overline{D})$.
  We are left with proving that $\overline{h}$ a homomorphism between
  $\treealgebra(\alphabetParse)$ and $\overline{\algof{B}}$, by
  showing the following points:

  \vspace*{.5\baselineskip}\noindent\underline{$\overline{h}(\extend{\symbOf{f}}^{\treealgebra(\alphabetParse)}(\tree_1,\ldots,\tree_{\arityof{f}}))
    =
    \extend{\symbOf{f}}^{\overline{\algof{B}}}(\overline{h}(\tree_1),
    \ldots, \overline{h}(\tree_n))$}, for $f \in \fsignature_\class
  \setminus \set{\pop_{\slabs,\slabs} \mid \slabs\in\sorts_\class}$:
  Assume that
  $\rankof{f}=\tuple{\sortof{\tree_1},\ldots,\sortof{\tree_{\arityof{f}}}}$
  and $\slabs = \sortof{f} =  \sortof{\extend{\symbOf{f}}^{\treealgebra(\alphabetParse)}(\tree_1,\ldots,\tree_{\arityof{f}})}$. Using the definition of $\extend{\symbOf{f}}$, we obtain:
  \begin{align*}
    & \overline{h}(\extend{\symbOf{f}}^{\treealgebra(\alphabetParse)}(\tree_1,\ldots,\tree_{\arityof{f}})) \\
    = & ~(\homof{\algof{P}_\class}{B}(\extend{\symbOf{f}}^{\treealgebra(\alphabetParse)}(\tree_1,\ldots,\tree_{\arityof{f}})), \slabs) \\
    = & ~(\homof{\algof{P}_\class}{B}(f^{\algof{P}_\class}(\tree_1,\ldots,\tree_{\arityof{f}})), \slabs) \\
    = & ~(f^\algof{B}(\homof{\algof{P}_\class}{B}(\tree_1), \ldots, \homof{\algof{P}_\class}{B}(\tree_{\arityof{f}})), \slabs) \\
    = & ~\extend{\symbOf{f}}^{\overline{\algof{B}}}((\homof{\algof{P}_\class}{B}(\tree_1),\sortof{\tree_1}), \ldots, (\homof{\algof{P}_\class}{B}(\tree_{\arityof{f}}),\sortof{f})) \\
    = & ~\extend{\symbOf{f}}^{\overline{\algof{B}}}(\overline{h}(\tree_1), \ldots, \overline{h}(\tree_n))
  \end{align*}

  \vspace*{.5\baselineskip}\noindent\underline{$\overline{h}(\tree_1
    \pop^{\treealgebra(\alphabetParse)} \tree_2) =
    \overline{h}(\tree_1) \pop^{\treealgebra(\alphabetParse)}
    \overline{h}(\tree_1)$}: Assume that we have $\slabs \isdef \sortof{\tree_1}=\sortof{\tree_2}$ for some $\slabs \in \sorts_\class$. 
    We compute, using the definition of $\pop^{\overline{\algof{B}}}$:
  \begin{align*}
    \overline{h}(\tree_1 \pop^{\treealgebra(\alphabetParse)} \tree_2) = & (\homof{\algof{P}_\class}{B}(\tree_1 \pop^{\treealgebra(\alphabetParse)} \tree_2), \slabs) \\
    = & (\homof{\algof{P}_\class}{B}(\tree_1) \pop^{\algof{B}} \homof{\algof{P}_\class}{B}(\tree_2), \slabs) \\
    = & (\homof{\algof{P}_\class}{B}(\tree_1), \slabs) \pop^{\overline{\algof{B}}} (\homof{\algof{P}_\class}{B}(\tree_2), \slabs) \\
    = & \overline{h}(\tree_1) \pop^{\overline{\algof{B}}} \overline{h}(\tree_2)
  \end{align*}


  \vspace*{.5\baselineskip}\noindent``$\Leftarrow$'' Assume that
  $\langu$ is recognizable in $\treealgebra(\alphabetParse)$ and let
  $\overline{\algof{B}} = (\universeOf{B}^{\set{1}},
  \set{\extend{\symbOf{f}}^{\overline{\algof{B}}}}_{f\in\fsignature_\class})$
  be a locally-finite single-sorted algebra such that $\langu =
  \homof{\treealgebra(\alphabetParse)}{\overline{\algof{B}}}^{-1}(\overline{D})$,
  for some set $\overline{D}\subseteq\overline{\universeOf{B}}$. We
  define the domain $\universeOf{B}^\slabs \isdef \set{(b,\slabs) \mid
    b \in \overline{\universeOf{B}}}$ of sort $\slabs$, for each
  $\slabs \in \sorts_\class$ and the
  $(\fsignature_\class,\sorts_\class)$-algebra $\algof{B} \isdef
  (\set{\universeOf{B}^\slabs}_{\slabs\in\sorts_\class},
  \set{f^\algof{B}}_{f \in \fsignature_\class})$ such that, for each
  $f \in \fsignature_\class$: \begin{align*}
    f^\algof{B}((b_1,\slabs_1), \ldots,
    (b_{\arityof{f}},\slabs_{\arityof{f}})) \isdef &
    ~(\extend{\symbOf{f}}^{\overline{\algof{B}}}(b_1,\ldots,b_\arityof{f}), \sortof{f}), \\
    \text{ if } f \in \fsignature_\class \setminus \set{\pop_{\slabs,\slabs} \mid \slabs \in \sorts_\class}, &
    ~\rankof{f}=\tuple{\slabs_1,\ldots,\slabs_{\arityof{f}}} \\
      (b_1,\slabs) \pop^\algof{B}_{\slabs,\slabs} (b_2,\slabs) \isdef & ~(b_1 \pop^{\overline{\algof{B}}} b_2, \slabs)
      \text{, for all } \slabs \in \sorts_\class
  \end{align*}
  We prove that $\overline{\algof{B}}$ is aperiodic only if
  $\algof{B}$ is aperiodic. Let $(b,\slabs) \in \universeOf{B}$ be an
  element. Then, $\idemof{(b,\slabs)}=(\idemof{b},\slabs)$, by the
  definition of $\pop^\algof{B}_{\slabs,\slabs}$. By the aperiodicity
  of $\overline{\algof{B}}$, we obtain $\idemof{(b,\slabs)}
  \pop^\algof{B}_{\slabs,\slabs} (b,\slabs) = (\idemof{b}
  \pop^{\overline{\algof{B}}} b,\slabs) = (\idemof{b},\slabs) =
  \idemof{(b,\slabs)}$.

  We define the set $D \isdef \set{(b,\slabs) \mid b \in \overline{D},
    \slabs \in \sorts_\class}$ and the function $h :
  \universeOf{P}_\class \rightarrow \universeOf{B}$ as $h(\tree)
  \isdef
  (\homof{\treealgebra(\alphabetParse)}{\overline{\algof{B}}}(\tree),
  \sortof{\tree})$. Then, we immediately obtain that $h(\tree)\in D
  \iff
  \homof{\treealgebra(\alphabetParse)}{\overline{\algof{B}}}(\tree)\in\overline{D}$,
  for each $\tree \in \universeOf{P}_\class$, i.e.,
  $\langu=\homof{\treealgebra(\alphabetParse)}{\overline{\algof{B}}}^{-1}(\overline{D})=h^{-1}(D)$. We
  are left with proving that $h$ is a homomorphism between
  $\algof{P}_\class$ and $\algof{B}$, by showing the following points:

  \vspace*{.5\baselineskip}\noindent\underline{$h(f^{\algof{P}_\class}(\tree_1,\ldots,\tree_{\arityof{f}}))
    = f^\algof{B}(h(\tree_1), \ldots, h(\tree_\arityof{f}))$}, for $f
    \in \fsignature_\class \setminus \set{\pop_{\slabs,\slabs} \mid
      \slabs\in\sorts_\class}$: Let $\slabs \isdef
  \sortof{h(f^{\algof{P}_\class}(\tree_1,\ldots,\tree_{\arityof{f}}))}$. We
  compute: \begin{align*}
    & ~h(f^{\algof{P}_\class}(\tree_1,\ldots,\tree_{\arityof{f}})) \\
    = & ~(\homof{\treealgebra(\alphabetParse)}{\overline{\algof{B}}}(f^{\algof{P}_\class}(\tree_1,\ldots,\tree_{\arityof{f}})), \slabs) \\
    = & ~(\extend{\symbOf{f}}^{\overline{\algof{B}}}(\homof{\treealgebra(\alphabetParse)}{\overline{\algof{B}}}(\tree_1),\ldots,\homof{\treealgebra(\alphabetParse)}{\overline{\algof{B}}}(\tree_{\arityof{f}}),\slabs) \\
    = & ~f^\algof{B}((\homof{\treealgebra(\alphabetParse)}{\overline{\algof{B}}}(\tree_1),\sortof{\tree_1}), \ldots, (\homof{\treealgebra(\alphabetParse)}{\overline{\algof{B}}}(\tree_{\arityof{f}}),\sortof{f})) \\
    = & ~f^\algof{B}(h(\tree_1), \ldots, h(\tree_{\arityof{f}}))
    \end{align*}

  \vspace*{.5\baselineskip}\noindent\underline{$h(\tree_1
    \pop^{\algof{P}_\class}_{\slabs,\slabs} \tree_2) = h(\tree_1)
    \pop^{\algof{B}}_{\slabs,\slabs} h(\tree_2)$}: We denote
  $\sortof{\tree_1} = \sortof{\tree_2} = \slabs$ and
  compute: \begin{align*}
    h(\tree_1 \pop^{\algof{P}_\class}_{\slabs,\slabs} \tree_2) = & ~(\homof{\treealgebra(\alphabetParse)}{\overline{\algof{B}}}(\tree_1 \pop^{\algof{P}_\class}_{\slabs,\slabs} \tree_2), \slabs) \\
    = & ~(\homof{\treealgebra(\alphabetParse)}{\overline{\algof{B}}}(\tree_1) \pop^{\overline{\algof{B}}}_{\slabs,\slabs} \homof{\treealgebra(\alphabetParse)}{\overline{\algof{B}}}(\tree_2),\slabs) \\
    = & ~(\homof{\treealgebra(\alphabetParse)}{\overline{\algof{B}}}(\tree_1),\slabs) \pop^{\algof{B}} (\homof{\treealgebra(\alphabetParse)}{\overline{\algof{B}}}(\tree_2),\slabs) \\
    = & ~h(\tree_1)  \pop^{\algof{B}} h(\tree_2)
  \end{align*}
\end{proofE}

Each parse tree $\tree \in \universeOf{P}_\class$ obtained by
compacting a $\fsignature_\class$-term $t$, has an associated graph,
denoted $\hval_\class(\tree) \isdef t^\class$. The following result
ensures that $\hval_\class$ is well-defined:

\begin{propositionE}[][category=proofs]\label{prop:unique-representations}
  For every parse tree $\tree \in \universeOf{P}_\class$ there is only
  one $\fsignature_\class$-term up to the associativity and
  commutativity of the operations $\pop_{\slabs,\slabs}$, with $\slabs
  \in \sorts_\class$, such that $\tree = t^{\repalgebra_\class}$.
\end{propositionE}
\begin{proofE}
  The proof is by induction on the size of $\tree$. Because every
  parse tree $\tree \in \universeOf{R}_\class$ is representable, we
  have that $\tree = \pop_{\slabs,\slabs}^{\repalgebra_\class}
  f_i^{\repalgebra_\class}(\tree_1^i,\ldots,\tree_\arityof{f_i}^i)$
  for some sort $\slabs \in \sorts_\class$, function symbols $f_i \in
  \fsignature_\class \setminus \set{\pop_{\slabs,\slabs} \mid \slabs
    \in \sorts_\class}$ and parse trees $\tree_j^i$.  We note that the
  choice of the function symbols $f_i$ is unique up to reordering
  because these are all the edges that are attached to the root of the
  parse tree. By the inductive hypothesis, we have that there are some
  terms $t_j^i$ such that $\tree_j^i = (t_j^i)^{\repalgebra_\class}$,
  which are unique up to the associativity and commutativity of the
  operations $\pop_{\slabs,\slabs}$.  Hence, $t = \pop_{\slabs,\slabs}
  f_i(\tree_1^i,\ldots,\tree_\arityof{f_i}^i)$ is a term with $\tree =
  t^{\repalgebra_\class}$ that is unique up to the associativity and
  commutativity of the operations $\pop_{\slabs,\slabs}$.
\end{proofE}

It is manifest that $\sortof{\hval_\class(\tree)}=\sortof{\tree}$, for
each $\tree \in \universeOf{P}_\class$ and that $\hval_\class$ is a
(sort-preserving) homomorphism between $\algof{P}_\class$ and
$\class$. Since both algebras are representable, this homomorphism is
unique, thus called the evaluation homomorphism. It is known that the
evaluation homomorphism is a definable transduction, for each class
$\btwclass{k}$~\cite[Proposition 7.48]{courcelle_engelfriet_2012}. The
following is a generalization of this result to arbitrary classes:

\begin{propositionE}[][category=proofs]\label{prop:eval-hom-def-transduction}
  For each class $\class$, the function $\hval_\class$ is a
  parameterless definable transduction.
\end{propositionE}
\begin{proofE}
  Let $k\geq1$ be the minimum integer such that each sort that occurs
  in some function symbol from $\fsignature_\class$ is a subset of
  $\interv{1}{k+1}$. Since $\fsignature_\class$ is finite, such an
  integer exists. Let $h : \universeOf{P}_\class \rightarrow
  \universeOf{P}_{\btwclass{k}}$ be homomorphism that replaces each
  $\symbOf{f}$-labeled edge of a parse tree for $\class$ with the
  parse tree corresponding to the $\btwsignature{k}$-term $t_f$ that
  interprets the function symbol $f$ in $\btwclass{k}$. Since each
  such term is finite, the $h$ function is a definable
  transduction. By Theorem~\ref{thm:bt}, $\hval_\class = h \circ
  \hval_{\btwclass{k}}$ is a definable transduction, because
  $\hval_{\btwclass{k}}$ is a definable transduction,
  by~\cite[Proposition 7.48]{courcelle_engelfriet_2012}.
\end{proofE}

Then, recognizability of a set in the class $\class$ of graphs equals
recognizability of the inverse $\hval$-image in the algebra
$\algof{P}_\class$:

\begin{lemmaE}[][category=proofs]\label{lem:parse-tree-class-relation}
  Let $\class$ be a class and $\langu \subseteq \universeOf{C}$ be a
  set of graphs.  Then, $\mathcal{S} \isdef \hval_\class^{-1}(\langu)$
  is (aperiodic) recognizable in $\repalgebra_\class$ if and only if
  $\langu$ is (aperiodic) recognizable in $\class$.
\end{lemmaE}
\begin{proofE}
``$\Rightarrow$'' We assume that $\mathcal{S} = \hval_\class^{-1}(\langu)$ is
  recognizable and let $(\algof{I},J)$ be a
  $(\fsignature_\class,\sorts_\class)$-recognizer such that
  $\mathcal{S} = \homof{\algof{P}_\class}{I}^{-1}(J)$. By Lemma
  \ref{lemma:syntactic-congruence-homomorphism}, the algebra
  ${\algof{P}_\class}_\requiv{\mathcal{S}}$ is locally finite and
  there exists a surjective homomorphism $g$ from $\mathcal{I}$ to
  ${(\algof{P}_\class)}_\requiv{\mathcal{S}}$ such that $h_\requiv{S}
  = g \circ \homof{\algof{P}_\class}{I}$. By Lemma
  \ref{lemma:factorization-of-related-syntactic-congruences}, the
  algebras ${(\algof{P}_\class)}_\requiv{\mathcal{S}}$ and
  $\class_\requiv{\langu}$ are isomorphic, hence
  $\class_\requiv{\langu}$ is locally finite, thus $\langu$ is
  recognizable in $\class$, witnessed by the homomorphism
  $h_\requiv{\langu}$. We assume now that $\mathcal{I}$ is an
  aperiodic algebra and prove that
  ${(\algof{P}_\class)}_\requiv{\mathcal{S}}$ is aperiodic. Provided
  that ${(\algof{P}_\class)}_\requiv{\mathcal{S}}$ is aperiodic, the
  aperiodicity of $\class_\requiv{\langu}$ follows from the
  isomorphism between ${(\algof{P}_\class)}_\requiv{\mathcal{S}}$ and
  $\class_\requiv{\langu}$. Let $a \in \universeOf{P}_\class$. By the
  surjectivity of $g$, there exists $i \in \universeOf{I}$ such that
  $g(i) = [a]_\requiv{\mathcal{S}}$ and let $\slabs$ be the sort of
  $i$, which is also the sort of $[a]_\requiv{\mathcal{S}}$. Hence,
  $g(\idemof{i}) =\idemof{[a]_\requiv{\mathcal{S}}}$ and we compute:
  $\idemof{[a]_\requiv{\mathcal{S}}}
  \pop^{{(\algof{P}_\class)}_\requiv{\mathcal{S}}}_{\slabs,\slabs} [a]_\requiv{\mathcal{S}} = g(\idemof{i})
  \pop^{{(\algof{P}_\class)}_\requiv{\mathcal{S}}}_{\slabs,\slabs}
  g(i) = g(\idemof{i} \pop^{\mathcal{I}}_{\slabs,\slabs} i) =
  g(\idemof{i}) = [a]_\requiv{\mathcal{S}}$. Since the choice of
  $a \in \universeOf{P}_\class$ was arbitrary, we obtain that
  ${(\algof{P}_\class)}_\requiv{\mathcal{S}}$ is aperiodic.

  \vspace*{.5\baselineskip}\noindent''$\Leftarrow$'' We assume that
  $\langu$ is recognizable and let $(\mathcal{I},J)$ be a
  $(\fsignature_\class,\sorts_\class)$-recognizer such that $\langu =
  \homof{\class}{I}^{-1}(J)$. Using a similar argument as for the
  previous direction, we obtain that $\class_\requiv{\langu}$ is a
  locally finite algebra that, moreover, is aperiodic if $\mathcal{I}$
  is aperiodic. Since ${(\algof{P}_\class)}_\requiv{\mathcal{S}}$ and
  $\class_\requiv{\langu}$ are isomorphic, by the argument from the
  previous direction, we obtain that
  ${(\algof{P}_\class)}_\requiv{\mathcal{S}}$ is locally finite and,
  moreover, aperiodic if $\class_\requiv{\langu}$ is aperiodic. Then,
  $\mathcal{S}$ is (aperiodic) recognizable, witnessed by
  $h_\requiv{\mathcal{S}}$.
\end{proofE}

\subsection{Parsable Classes}
\label{subsec:parsable}

In general, the inverse of the evaluation homomorphism $\hval$ is not
a definable transduction. The parsable classes are the ones for which
this is the case. The following definition is closely related to the
parsable presentations of sets of graphs, introduced by
Courcelle~\cite[Definition 4.2]{CourcelleV}:

\begin{definition}\label{def:parsable}
  The class $\class$ is \emph{parsable} iff there exists a definable
  transduction $\parsefunc \subseteq \hval_{\class}^{-1}$ such that
  $\dom{\parsefunc}=\universeOf{C}$.
\end{definition}

Parsability ensures the equivalence between recognizability and
\cmso-definability:

\begin{theorem}[Theorem 4.8(2) in \cite{CourcelleV}]\label{thm:parsable-cmso-rec}
  Let $\class$ be a parsable class and $\langu\subseteq\universeOf{C}$
  be a set. Then, $\langu$ is \cmso-definable if and only if $\langu$
  is recognizable in $\class$.
\end{theorem}

We prove the main result of this section before showing its
applicability to the class $\algof{G}_2$. Note that this result was
already proved for the class of trees (Theorem~\ref{thm:mso-aperiodic-trees}).

\paragraph{Proof of Theorem~\ref{thm:mso-aperiodic}}
``$\Rightarrow$'' Let $(\algof{A},\universeOf{B})$ be an aperiodic
$(\fsignature_\class,\sorts_\class)$-recognizer for $\langu$ such that
$\langu=\homof{\class}{A}^{-1}(B)$. Then, $\homof{\class}{A} \circ
\hval_\class$ is a homomorphism between $\algof{P}_\class$ and
$\algof{A}$ and $\mathcal{S} \isdef \hval^{-1}_\class(\langu) =
(\homof{\class}{A} \circ \hval_\class)^{-1}(B)$ is aperiodic
recognizable in $\algof{P}_\class$. By Lemma \ref{lemma:red-rec},
$\mathcal{S}$ is aperiodic recognizable in
$\treealgebra(\alphabetParse)$, hence $\mathcal{S}$ is \mso-definable,
by Theorem~\ref{thm:mso-aperiodic-trees}. Since $\class$ is parsable,
there exists a definable transduction $\parsefunc \subseteq
\hval_{\class}^{-1}$ such that $\dom{\parsefunc}=\universeOf{C}$. We
obtain $\langu=\parsefunc^{-1}(\mathcal{S})$, thus $\langu$ is
\mso-definable, by Theorem~\ref{thm:bt}.

\noindent ``$\Leftarrow$'' Assume that $\langu$ is \mso-definable and
$\mathcal{S} = \hval^{-1}_\class(\langu)$ be the set of parse trees
that evaluate to some graph from $\langu$. By Proposition
\ref{prop:eval-hom-def-transduction}, $\hval_\class$ is a definable
transduction, hence $\mathcal{S}$ is an \mso-definable set of trees,
by Theorem~\ref{thm:bt}. By Theorem~\ref{thm:mso-aperiodic-trees},
$\mathcal{S}$ is aperiodic recognizable in
$\treealgebra(\alphabetParse)$, hence also aperiodic recognizable in
$\algof{P}_\class$, by Lemma \ref{lemma:red-rec}. By Lemma
\ref{lem:parse-tree-class-relation}, $\langu$ is aperiodic
recognizable in $\class$. \qed

It is already known that the classes of (disoriented) series-parallel
graphs are parsable~\cite[Theorems 6.10 and 6.12]{CourcelleV}. For
self-completeness reasons, we re-prove the following theorem using our
notation:

\begin{theoremE}[][category=proofs]\label{thm:sp-reg-parsable}
  The classes $\algof{SP}$ and $\algof{DSP}$ are parsable.
\end{theoremE}
\begin{proofE}
  We first describe the parsing transduction $\parsefunc_\algof{SP}$
  for the $\algof{SP}$ class. The parsing transduction for the
  $\algof{DSP}$ class will be the composition $\parsefunc_\algof{SP}
  \circ \trans_\mathit{rev}$ of a transduction $\trans_\mathit{rev}$
  that reverses an arbitrarily chosen set of edges and
  $\parsefunc_\algof{SP}$. Note that the output of
  $\trans_\mathit{rev}$ is rejected by $\universeOf{SP}$ if it does
  not belong to $\universeOf{SP}$.

  We recall the following definitions from \cite[Section
    6]{CourcelleV}. Let $\graph \in \universeOf{SP}$ be an (oriented)
  series-parallel graph and $x,y \in \vertof{\graph}$ be (not
  necessarily distinct) vertices of $\graph$. \begin{enumerate}[1.]
  \item $x \leq y$ iff there exists an oriented path from $x$ to $y$
    in $\graph$ and $x < y$ iff $x \leq y$ and $x \neq y$.
  \item for $x < y$, $\graph[x,y]$ is the subgraph of $\graph$
    consisting of all vertices $x \leq z \leq y$ and all edges between
    them.
  \item $\vertavoid{\graph}$ is the set of vertices of $\graph$ that are
    not on some path from $\sourceof{\graph}(1)$ to
    $\sourceof{\graph}(2)$. Note that every vertex from
    $\vertof{\graph}$ is on some path from $\sourceof{\graph}(1)$ to
    $\sourceof{\graph}(2)$, however the vertices from
    $\vertavoid{\graph}$ are avoided by least one such path.
  \item for each $x \in \vertavoid{\graph}$ and $y \in \vertof{\graph}$,
    we define $\mathsf{left}(x) \isdef y$ (resp. $\mathsf{right}(x)
    \isdef y$) iff the following hold: \begin{enumerate}[(a)]
    \item\label{it1:left-vertex} $y < x$ (resp. $x < y$).
    \item\label{it2:left-vertex} every path from $\sourceof{\graph}(1)$ to
      $\sourceof{\graph}(2)$ that crosses $x$ also crosses $y$.
    \item\label{it3:left-vertex} there is some path from $\sourceof{\graph}(1)$ to
      $\sourceof{\graph}(2)$ that crosses $y$ but not $x$.
    \item any vertex $y'$ that satisfies conditions
      (\ref{it1:left-vertex}-\ref{it3:left-vertex}) for $x$ is such
      that $y' \leq y$ (resp. $y \leq y'$).
    \end{enumerate}
   \item $\graph[x]\isdef\graph[\mathsf{left}(x),\mathsf{right}(x)]$,
     for each vertex $x \in \vertavoid{\graph}$. The existence and
     unicity of $\mathsf{left}(x)$ and $\mathsf{right}(x)$ is proved
     in \cite[Lemma 6.6]{CourcelleV}.
   \item
     $\graph[e]\isdef\graph[\edgerelof{\graph}(e,1),\edgerelof{\graph}(e,2)]$,
     for each edge $e \in \edgeof{\graph}$.
   \item $\singledge{\graph}{e}$ is the graph consisting of the edge
     $e$ and vertices $\edgerelof{\graph}(e,1)$ and
     $\edgerelof{\graph}(e,2)$, for each edge $e \in
     \edgeof{\graph}$. Note that $\graph[e]$ and
     $\singledge{\graph}{e}$ are not the same, e.g., if $\graph =
     (\sgraph{b}_{1,2} \pop \sgraph{b}_{1,2})^\algof{SP}$. Moreover,
     each $\graph[e]$, resp. $\singledge{\graph}{e}$, is a P-component
     of $\graph$.
  \end{enumerate}
  We call a P-component of $\graph$ any $\sop$-atomic subgraph of
  $\graph$, i.e., either a single-edge graph or a parallel
  compositions of at least two other subgraphs. An S-component of
  $\graph$ is a subgraph $C$ of $\graph$ that is a serial composition
  of at least two other subgraphs of $\graph$ such that $C
  \sop^\algof{SP} K$ or $K \sop^\algof{SP} C$ are not subgraphs of
  $\graph$ for any other graph $K$. By the corrected\footnote{The
  original statement of \cite[Lemma 6.7]{CourcelleV} considers only
  $\set{\graph[x] \mid x \in \edgeof{\graph} \cup \vertavoid{\graph}}$
  and ignores $\set{\singledge{\graph}{e} \mid e \in
    \edgeof{\graph}}$.} statement of~\cite[Lemma 6.7]{CourcelleV}, the
  set of P-components of some series-parallel subgraph of $\graph$
  equals $\set{\graph[x] \mid x \in \edgeof{\graph} \cup
    \vertavoid{\graph}} \cup \set{\singledge{\graph}{e} \mid e \in
    \edgeof{\graph}}$. Moreover, the functions $\mathsf{left},
  \mathsf{right} : \vertavoid{\graph} \rightarrow \vertof{\graph}$ are
  definable using \mso.

  To prove that the class $\algof{SP}$ is parsable, we describe a
  definable transduction $\parsefunc_\algof{SP}$ that builds from each
  input graph $\graph$ (i.e., $\graph \in \universeOf{G}$ is any
  graph, not just some series-parallel graph) a parse tree $\tree \in
  \universeOf{P}_\algof{SP}$ such that $\hval_\algof{SP}(\tree) =
  \graph$. First, the input graph $\graph$ is rejected if $\graph
  \not\in \universeOf{SP}$. The latter condition can be checked by
  encoding the conditions of \cite[Lemma 6.1]{CourcelleV} using
  \mso\ formul{\ae}. The transduction uses a parameter $X$ that is
  assigned to a subset of $\vertavoid{\graph} \cup \edgeof{\graph}$
  having the following \mso-definable properties (we abuse notation
  and write $x \in X$ whenever $x$ belongs to the valuation of
  $X$): \begin{itemize}[-]
  \item $\graph[x]\neq\graph[y]$, for all $x \neq y \in X$,
  \item for each $x \in \vertavoid{\graph} \cup \edgeof{\graph}$ there
    exists $y \in X$ such that $\graph[x]=\graph[y]$.
  \end{itemize}

  \begin{figure}[t!]
    \centerline{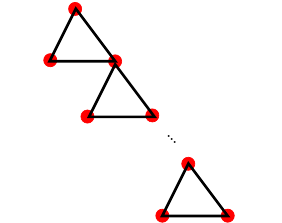}
    \vspace*{-.5\baselineskip}
    \caption{Chaining of $\symbOf{\sop}$ edges in $\algof{SP}$ parse
      tree. The nodes of the parse tree are shown in red and the edges
      in black.}
    \vspace*{-\baselineskip}
    \label{fig:sp-parse}
  \end{figure}
  
  The transduction $\parsefunc_\algof{SP}$ uses $k=3$ layers.  We
  denote by $(x,i)$ an element $x \in X$ taken from the layer $i \in
  \set{1,2,3}$. The nodes of the output parse tree $\tree$ are the
  vertices $x \in X$ from the $1^\mathit{st}$ layer. The edges of
  $\tree$ are elements from the $3^\mathit{rd}$ layer. We use the
  $2^\mathit{nd}$ layer only for the additional vertices used by the
  edges that encode serial composition subterms. The intuition is that
  a node $(x,i)$ is the root of a subtree that evaluates to a subgraph
  of $\graph$ occurring as the $i^{\mathit{th}}$ argument of a binary
  $\sop$ operation, being thus attached to the $(i+1)^\mathit{th}$
  position of a $\symbOf{\sop}$-labeled edge of $\tree$, for
  $i=1,2$. Moreover, the distinction between $\graph[e]$ and
  $\singledge{\graph}{e}$, for some edge $e \in \edgeof{\graph}$ is
  mirrored by the distinction between $(e,1)$ (i.e., a node of
  $\tree$) and $(e,3)$ (i.e., an edge of $\tree$). The parse tree
  $\tree$ of $\graph$ (i.e., the result of
  $\parsefunc_\algof{SP}(\graph)$ for the given valuation of $X$) is
  defined as follows: \begin{itemize}[-]
  \item $\vertof{\tree} \isdef \set{(x,1),(x,2) \mid x \in X}$,
  \item $\edgeof{\tree}$ consists of either: \begin{itemize}[*]
  \item unary $\symbOf{\sgraph{b}}_{1,2}$-labeled edges $(x,3)$,
    attached to a single vertex, for each $x \in X \cap
    \edgeof{\graph}$. These edges correspond to the P-components of
    $\graph$ of the form $\singledge{\graph}{e}$, for $e \in
    \edgeof{\graph}$. The choice of the vertex
    $\edgerelof{\tree}((x,3),1)$ depends on the position of the
    single-edge graph $\singledge{\graph}{x}$,
    namely: \begin{enumerate}[(i)]
    \item $\singledge{\graph}{x}$ occurs alone or on the
      $1^\mathit{st}$ position on some S-component of $\graph$, in
      which case $\edgerelof{\tree}((x,3),1)=(y,1)$, where $y$ is the
      representative of the innermost enclosing P-component of that
      S-component,
    \item $\singledge{\graph}{x}$ occurs on some S-component, but not
      on the first or last position, in which case
      $\edgerelof{\tree}((x,3),1)=(x,1)$,
    \item $\singledge{\graph}{x}$ occurs on the last position on some
      S-component, in which case $\edgerelof{\tree}((x,3),1)=(y,2)$,
      where $y$ is the representative of the previous P-component
      within that S-component.
    \end{enumerate}
    This case distinction can be encoded in \mso.
  \item ternary $\symbOf{\sop}$-labeled edges correspond to the
    P-components of $\graph$ that occur in some S-component. More
    specifically, for each $x \in X$, we have $\graph[x] = C_1
    \pop^\algof{SP} \ldots \pop^\algof{SP} C_m$, where $C_i =
    \graph[x_{i,1}] \sop^\algof{SP} \ldots \sop^\algof{SP}
    \graph[x_{i,n_i}]$ and $x_{i,j} \in X$, for all $i \in
    \interv{1}{m}$ and $j \in \interv{1}{n_i}$. Then, for each $i \in
    \interv{1}{m}$ and $j \in \interv{1}{n_i}$, $\tree$ has a
    $\symbOf{\sop}$-labeled edge $(x_{i,j},3)$ attached to the
    following vertices (see Figure \ref{fig:sp-parse} for an
    illustration): \begin{align*}
      \edgerelof{\tree}((x_{i,j},3),1) \isdef & ~\begin{cases} (x,1) \text{, if } j = 1 \\
        (x_{i,j-1},2) \text{, if } 2 \leq j \leq n_i-1
      \end{cases} \\
      \edgerelof{\tree}((x_{i,j},3),k) \isdef & ~(x_{i,j},k-1) \text{, for } 1 \leq j < n_i-1, k=2,3 \\
      \edgerelof{\tree}((x_{i,n_i-1},3),2) \isdef &~(x_{i,n_i-1},1) \\
      \edgerelof{\tree}((x_{i,n_i-1},3),3) \isdef &~(x_{i,n_i},1) \\
    \end{align*}
    Note that the relation ``\emph{$\graph[x_{i,j}]$ is to the left of
    $\graph[x_{i,\ell}]$ in the S-component $C_i$}'', for all $1 \leq
    j < \ell \leq n_i$, can be defined in \mso, thus the distinction
    between the two cases in the definition of
    $\edgerelof{\tree}((x_{i,j},3),1)$ above is \mso-definable.
  \end{itemize}
  \end{itemize}
  The root of $\tree$ is $(x,1)$, where $x \in X$ is the
  unique element such that $\graph[x]$ is maximal in the subgraph
  relation among all $y \in X$. The proof of the fact that
  $\hval_{\algof{SP}}(\tree)=\graph$ is routine, by induction on the
  decomposition of $\graph$ given by Lemma \ref{lemma:sp-decomp}.

  To prove that the class $\algof{DSP}$ is parsable, we define the
  copyless transduction $\trans_\mathit{rev}$ having parameter $X$,
  assigned to a set of edges of the input graph that will be reversed
  in the output graph $\graph'$, i.e., for each edge $x \in X$ we have
  $\labof{\graph'}(x) \isdef \labof{\graph}(x)^r$, where $b^r$ is a
  unique edge label corresponding to $b \in \alphabetTwo$ and
  $\edgerelof{\graph'}(x,i) \isdef \edgerelof{\graph}(x,3-i)$, for $i
  = 1,2$. Then, $\parsefunc_\algof{DSP}$ is the composition of
  $\trans_\mathit{rev}$ with the modified version of
  $\parsefunc_\algof{SP}$, where the edges with labels $b^r$ are
  encoded by unary $\symbOf{\sgraph{b}}_{2,1}$ edges in the output
  parse tree.
\end{proofE}

Moreover, each class $\btwclass{k}$ is parsable, by the combined
results
of~\cite{10.1145/2933575.2934508,journals/lmcs/BojanczykP22}. Note
that, even though the connected graphs from the domain of
$\btwclass{2}$ are the exactly the ones from the domain of
$\algof{G}_2$, the signatures of the two classes are different, since
$\algof{G}_2$ is a derived algebra of $\btwclass{2}$. Finally, we
prove that:

\begin{theoremE}[][category=proofs]\label{thm:tw-two-reg-parsable}
  The class $\algof{G}_2$ is parsable.
\end{theoremE}
\begin{proofE}
  To prove that $\algof{G}_2$ is parsable, we build a definable
  transduction $\parsefunc_2$ that produces, from any connected graph
  $\graph$ of tree-width $\leq2$ a parse tree $\tree \in
  \universeOf{P}_{\twalg}$ such that $\hval_{\twalg}(\tree) =
  \graph$. We define $\parsefunc_2 \isdef \trans_2 \circ \trans_1$,
  where the definable transductions $\trans_1$ and $\trans_2$ are
  described below. We refer to the definitions introduced in the proof
  of Theorem~\ref{thm:sp-reg-parsable}.

  First, $\trans_1$ checks whether the input graph $\graph$ has
  tree-width at most~$2$, by checking the absence of a $K_4$ minor,
  and assigns its parameter $X$ to a set of edges from $\graph$.  The
  output $\graph'$ of $\trans_1$ is the same as $\graph$, except for
  the edges from $X_1$ that are reversed. The reversed edges have the
  labels $b^r$, where the original labels are $b \in \alphabetTwo$.

  Second, $\trans_2$ produces a parse tree $\tree$ of its input graph
  $\graph$ in a similar way as the $\parsefunc_\algof{SP}$
  transduction from the proof of Theorem
  \ref{thm:sp-reg-parsable}. Note that each block of $\graph$ is a
  P-component of $\graph$ and, moreover, the condition that
  $\graph[x]$ does not contain a cutvertex can be encoded in \mso.
  The block tree of $\graph$ can be represented by the \mso-definable
  ``\emph{descendent-of}'' relation between the representatives of its
  blocks, which in fact, describes a forest. We assume w.l.o.g. that
  the root of the block tree of $\graph$ is a cutvertex of $\graph$,
  or the $1$-source of $\graph$, in case $\graph$ consists of a single
  block, i.e., $\graph$ is $2$-connected.

  \begin{figure}[t!]
    \centerline{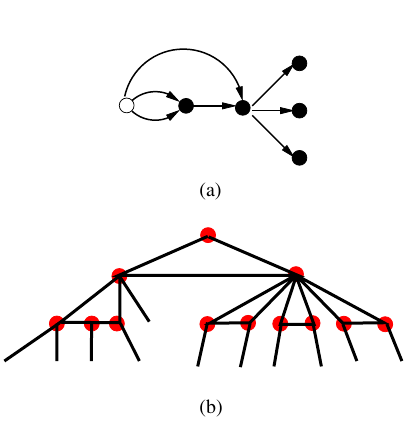}
    \vspace*{-.5\baselineskip}
    \caption{A graph from $\universeOf{G}_2$ (a) and its corresponding
      parse tree (b). The braces indicate the representatives of
      P-components $\graph[v_2]$ and $\graph[e_1]$. All other
      P-components are single edges $\singledge{\graph}{e_i}$ having
      $e_i$ as representative, for $i = 1, \ldots, 7$. The nodes of the
      parse tree are shown in red and the edges in black. }
    \vspace*{-\baselineskip}
    \label{fig:parse}
  \end{figure}

  The transduction $\trans_2$ uses $k=7$ layers and assigns its
  parameter $X$ to a set of vertices and edges from its input $\graph$
  (i.e., the output of $\trans_1$) that are the unique representatives
  of the P-components of $\graph$. In addition $X$ contains the
  $1$-source of $\graph$, which can be identified by an \mso\ formula.

  The output parse tree $\tree$ is defined as follows. The nodes of
  $\tree$ are $\vertof{\tree} \isdef \set{(x,1), (x,2), (x,3) \mid x
    \in X}$. We use the layer $j \in \set{1,2,3}$ to represent the
  nodes attached to the $\symbOf{\sop}$-labeled edges on the
  $(j+1)^\mathit{th}$ position of an edge. The layer $j \in
  \set{4,5,6,7}$ distinguishes the four types of edges, as
  follows: \begin{itemize}[-]
  \item unary $\symbOf{\emptygraph}_{\set{1}}$-labeled edges of the
    form $(x,4)$, for some $x \in X$, attached to a single vertex
    $\edgerelof{\tree}((x,4),1)=(x,j)$, which is attached to the
    $j^\mathit{th}$ position of another edge, for $j \in \set{2,3}$.
  \item unary $\symbOf{\sgraph{b}}_{i,3-i}$-labeled edges $(x,5)$
    correspond to the edges $x \in X \cap \edgeof{\graph}$, where $i =
    1$ if $\labof{\graph}(x)=b$ and $i=2$ if $\labof{\graph}(x)=b^r$,
    i.e., the orientation of the edge depends on its label,
    $b^r$-labeled edges being exactly the ones reversed by
    $\trans_1$. These edges are attached to a single vertex
    $\edgerelof{\tree}((x,6),1) \isdef (x,j)$, where the vertex
    $(x,j)$ is attached to some other edge on the $j^\mathit{th}$
    position, for $j \in \set{2,3,4}$.
  \item ternary $\symbOf{\rop}$-labeled edges $(x,6)$, where $x \in X$
    is the representative of a block of $\graph$. Each
    $\symbOf{\rop}$-labeled edge is attached to the following
    vertices: \begin{align*}
      & \edgerelof{\tree}((x,6),1) \isdef \\
      & \begin{cases}
        (y,3) \text{, where } y \text{ is the representative of the P-component} \\[-1mm]
        \hspace*{9mm} \text{ to which } x \text{ is attached in some S-component} \\
        (y,1) \text{, where } y \text{ is the $1$-source of } \graph
      \end{cases}
    \end{align*}
    Note that the relation that finds the representative of the
    P-component to which the block $x$ is attached in an S-component
    (as the $3^\mathit{rd}$ argument of a serial composition) can be
    encoded in \mso, similar to the left-neighbour relation in an
    S-component (see proof of Theorem~\ref{thm:sp-reg-parsable}).
  \item $\symbOf{\sop}$-labeled edges of arity $4$ are defined similar
    to the ternary $\symbOf{\sop}$-labeled edges from the proof of
    Theorem~\ref{thm:sp-reg-parsable}. In particular, each S-component
    of a P-component $\graph[x]$ of $\graph$ is of the form
    $\sop^\twalg(\graph[x_1], \sop^\twalg(\graph[x_2], \ldots
    \sop^\twalg(\graph[x_{n-1}],\graph[x_{n}],H_{n-1}) \ldots, H_2),
    H_1)$, where $\graph[x_1]$, $\ldots$, $\graph[x_n]$ are
    P-components of $\graph[x]$ and $H_1, \ldots, H_{n-1}$ are
    (possibly empty) parallel compositions of blocks of sort
    $\set{1}$. Then, for each $i\in\interv{1}{n}$, $\tree$ has a
    $\symbOf{\sop}$-labeled edge $(x_i,7)$ attached to the following
    vertices: \begin{align*} \edgerelof{\tree}((x_i,7),1) \isdef &
      ~\begin{cases} (x,1) \text{, if } i = 1 \\
        (x_{i-1},2) \text{, if } i \geq 2
      \end{cases} \\
      \edgerelof{\tree}((x_i,7),j) \isdef & ~(x_i,j-1) \text{, for } j \in \set{2,3,4}
    \end{align*}
    Note that the subtree corresponding to each non-empty parallel
    composition of blocks $H_i$ is a parallel composition of
    $\symbOf{\rop}$-labeled edges whose $1^\mathit{st}$ vertex is
    $(x_i,1)$, by the definition of $\symbOf{\rop}$-labeled edges. The
    relation between $x_i$ and the representative of each block from
    $H_i$ can be defined using \mso.
  \end{itemize}
  We refer to Figure \ref{fig:parse} for an illustration of this
  definition. We note that the case distinctions above are definable
  using \mso.  The root of $\tree$ is $(r,1)$, where $r \in Y$ exists
  and is unique. The proof that $\hval_{\twalg}(\tree)=\graph$ goes by
  induction on the decomposition of $\graph$ given by Lemmas
  \ref{lemma:sp-decomp} and \ref{lemma:tw2-block}.
\end{proofE}

\section{Conclusion}

We introduce regular grammars for classes of graphs of tree-width at
most $2$. These grammars provide finite and compact representations
for the recognizable sets in each class. The inclusion of a
context-free language into a regular language can be decided in
doubly-exponential time. A further restriction of regular grammars
describes exactly the \mso-definable sets thereof, that are, moreover,
recognized by aperiodic algebras.

As future work, we consider a more precise assessment of the size of
the minimal recognizer algebra for the classes considered in this
paper, as well as addressing the problem of finding regular grammars
for the class of graphs of tree-width at most $3$.

\section*{Acknowledgements}
Marius Bozga and Radu Iosif wish to acknowledge the support of the
French National Research Agency project Non-Aggregative Resource
COmpositions (NARCO) under grant number ANR-21-CE48-0011.  Florian
Zuleger wishes to acknowledge the support of the FWF project AUTOSARD:
``Automated Sublinear Amortised Resource Analysis of Data Structures''
No.~P36623 and the project VASSAL: ``Verification and Analysis for
Safety and Security of Applications in Life'' funded by the European
Union under Horizon Europe WIDERA Coordination and Support
Action/Grant Agreement No. 101160022
\includegraphics[width=.03\textwidth]{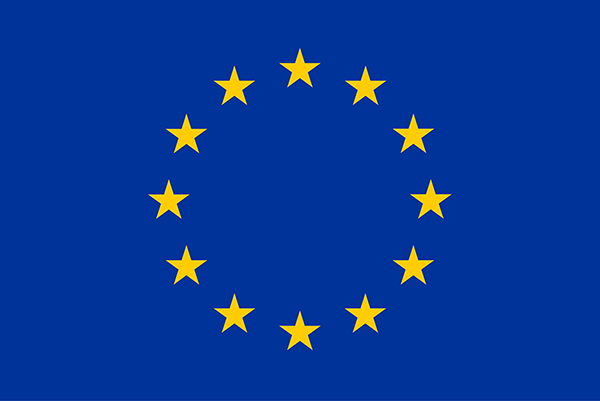}. All
authors thank the anonymous reviewers for their insightful comments.

\bibliographystyle{IEEEtran}
\bibliography{refs}

\ifLongVersion\else
\appendices

\section{Proofs}
\label{app:proofs}
\printProofs[proofs]

\section{Additional Material on Recognizability}
\label{app:rec}
\printProofs[rec]

\section{Series Parallel Graphs}
\label{app:sp}
\printProofs[sp]
\fi

\end{document}